\documentclass[12pt,reqno]{amsart}
\usepackage[top=1.25in,bottom=1in,left=1in,right=1in]{geometry}
\usepackage[utf8]{inputenc} 
\usepackage{amsmath}
\usepackage{amssymb}
\usepackage{amsthm}
\usepackage{bbm}
\usepackage{xcolor}
\usepackage{graphicx}
\usepackage{float}
\usepackage{enumerate}
\usepackage{natbib}
\usepackage{setspace,lipsum}
\usepackage{hyperref}

\usepackage{mathrsfs}
\usepackage{comment}
\usepackage{apptools}
\usepackage{url}
\newcommand{\RN}[1]{  \textup{\uppercase\expandafter{\romannumeral#1}}}

\usepackage{etoolbox}
\makeatletter
\patchcmd{\@maketitle}
  {\ifx\@empty\@dedicatory}
  {\ifx\@empty\@date \else {\vskip3ex \centering\footnotesize\@date\par\vskip1ex}\fi
   \ifx\@empty\@dedicatory}
  {}{}
\patchcmd{\@adminfootnotes}
  {\ifx\@empty\@date\else \@footnotetext{\@setdate}\fi}
  {}{}{}
\makeatother

\makeatother

\AtAppendix{
\numberwithin{equation}{section}
\numberwithin{definition}{section}
\numberwithin{theorem}{section}
\numberwithin{proposition}{section}
\numberwithin{lemma}{section}
\numberwithin{corollary}{section}}

\allowdisplaybreaks
\usepackage{graphicx}
\usepackage{amsmath}
\usepackage{bbm}
\usepackage{amssymb}
\usepackage{amsthm}
\usepackage{wrapfig}
\usepackage{multicol}
\usepackage{float}
\usepackage{pdfpages}
\usepackage{natbib}
\usepackage{longtable}
\usepackage{makecell}
\usepackage{booktabs}

\usepackage{setspace}
\usepackage[utf8]{inputenc}
\usepackage{subfiles}
\usepackage{blindtext}
\usepackage{enumitem}
\usepackage{lipsum}
\usepackage[title,toc,titletoc]{appendix}
\usepackage{etoolbox}
\usepackage{tikz}
\usepackage{mathrsfs}
\usepackage{diagbox}
\usepackage{bm}
\newcolumntype{P}[1]{>{\centering\arraybackslash}p{#1}}
\newcolumntype{M}[1]{>{\centering\arraybackslash}m{#1}}
\usepackage[font=footnotesize, textfont=it]{caption}
\usepackage{multirow}
\usepackage{subfig}
\usepackage{anyfontsize}
\newcommand\myfontsize{\fontsize{10pt}{12pt}\selectfont}

\newcommand{\E}{\mathbb{E}}

\newcommand{\bM}{\mathbf{M}}
\newcommand{\bT}{\mathbf{T}}
\newcommand{\bbT}{\mathbb{T}}

\newcommand{\R}{\mathbb{R}}

\newcommand{\I}{\mathbb{I}}

\renewcommand{\P}{\mathbb{P}}
\newcommand{\td}{\widetilde}

\newcommand{\lb}{\left(}
\newcommand{\rb}{\right)}
\newcommand{\T}{\mathcal{T}}

\newcommand{\cF}{\mathcal{F}}

\newcommand{\bfpiw}{\bm{\pi}^{w}}
\newcommand{\bfpip}{\bm{\pi}^{p}}
\newcommand{\bfpin}{\bm{\pi}^{n}}
\newcommand{\Bin}{\mathrm{Bin}}

\renewcommand{\td}{\widetilde}

\newcommand{\cW}{\mathcal{W}}
\newcommand{\one}{\mathbf{1}}

\newcommand{\w}{w}

\usepackage{colortbl}

\newtheorem{claim}{Claim}

\newtheorem{proposition}{Proposition}
\newtheorem{theorem}{Theorem}
\newtheorem{corollary}{Corollary}
\newtheorem{lemma}{Lemma}

\newtheorem{assumption}{Assumption}

\theoremstyle{remark}
\newtheorem{definition}{Definition}
\newtheorem{example}{Example}
\theoremstyle{remark}

\DeclareMathOperator*{\argmin}{argmin}
\DeclareMathOperator*{\arginf}{arginf}

\definecolor{myyellow}{HTML}{FFD500}
\definecolor{myblue}{HTML}{005BBB}
\title[The Transfer Performance of Economic Models]{The Transfer Performance of Economic Models}
\thanks{We thank Jiafeng Chen, Stefano DellaVigna, Nic Fishman, Matthew Gentzkow, Andrei Iakovlev, Michael Jordan, Brandon Kaplowitz, Paulo Natenzon, Kyohei Okumura, Pietro Ortoleva, Michael Ostrovsky, Ashesh Rambachan, Bas Sanders, Jesse Shapiro, Eisho Takatsuji, and Davide Viviano for helpful comments. We also thank  NSF grants 1851629 and 1951056 for financial support, and  the  Quest high performance computing facility at Northwestern University for computational resources and staff assistance. }

\author[]{Isaiah Andrews$^\S$}
\thanks{$^\S$Department of Economics, Harvard}
\author[]{Drew Fudenberg$^\ddag$}
\thanks{$^\ddag$Department of Economics, MIT}
\author[]{Lihua Lei$^\#$}
\thanks{$^\#$Stanford Graduate School of Business}
\author[]{Annie Liang$^\dag$}
\thanks{$^\dag$Department of Economics, Northwestern University}
\author[]{Chaofeng Wu$^*$}
\thanks{$^*$Department of Computer Science, Northwestern University}
\date{\today}

\begin{document}

\pagenumbering{gobble}

% \maketitle
%[]
\begin{abstract}
Economists often use models estimated on data from a particular domain to make predictions in another. We provide a tractable formulation for this ``out-of-domain" prediction problem and define the \emph{transfer error} of a model  based on its performance in  a new domain. 
We derive finite-sample forecast intervals that are guaranteed to cover realized transfer errors with a user-selected probability when domains are i.i.d. We apply these intervals in an application to compare the transferability of economic models and black box algorithms, finding that black box algorithms outperform economic models for prediction within a domain, but generalize more poorly across domains.
\end{abstract}
\maketitle
\thispagestyle{empty} \vspace{-1em}

\newpage \pagenumbering{arabic} \setcounter{page}{1}

\section{Introduction}

Economists routinely  make predictions in environments where data is unavailable,  relying on evidence from related but distinct contexts. For example, an economist at a development agency may need to predict diffusion of microfinance takeup in one Indian village given data on diffusion in others. Or an economist at an insurance company may need to predict willingness-to-pay for certain insurance plans given data on willingness-to-pay for others.

Traditionally, economists address this challenge by estimating economic models on existing data and using the estimated parameters to make predictions in new domains. But the  predictive success of machine learning algorithms in several economics applications
\citep[e.g.][]{HartfordWrightLeytonBrown,Hofmanetal,banerjee2023selecting} raises the question of whether the economic model is essential in this process, or whether predictions would be improved by training and porting a flexible machine learning algorithm instead. This is an empirical question: On the one hand, machine learning methods are capable of uncovering novel patterns that existing models miss \citep{FudenbergLiang,Petersonetal,LudwigMullainathan}; on the other, many researchers believe that structured economic models capture fundamental regularities that generalize more reliably across domains \citep{CoveneyDoughertyHighfield,Athey2017,Beery_2018_ECCV,Manski2021}.  Understanding whether economic models do in fact transfer better across domains is important to understanding their  future role within economic analysis and policymaking.

Our paper provides a conceptual framework that formalizes the ``out-of-domain'' transfer problem, and allows systematic comparison between  economic modeling and machine learning algorithms. Relative to the large body of work on external validity described in Section \ref{sec:RelatedLit}, our out-of-domain test considers performance within a class of prediction problems rather than conditioning on a particular realization of the data available to the economist. For example, instead of assessing whether a structural model of information diffusion will transfer well from a specific Indian village to another, we ask whether the structural model transfers well ``in general'' across Indian villages. Formally, we consider an ex-ante perspective in which the economist is aware of a set of domains that are relevant (such as Indian villages) but does not know which specific transfer problem will realize from this class of possible transfer problems. 

Our main theoretical contribution is the construction of finite-sample forecast intervals that characterize how well a model or algorithm performs in new domains.\footnote{\label{fn:Forecast} We use the term ``forecast interval,'' rather than ``confidence interval,'' to reflect the random nature of the target, namely the \emph{realized} (rather than expected, median, etc.) transfer error, but they can also be viewed as confidence intervals for these random targets.} In the main text, we provide intervals that are valid for a benchmark setting where all domains are equally likely to realize for training and testing (corresponding to an assumption that domains are exchangeable) as well as for a setting in which the testing domain is qualitatively different from the training domains. (Appendix \ref{subsec:general_dist_shift}
generalizes our approach and results even further.) The definition of ``performance'' is left to the user, and our results apply for a large class of measures that includes our motivating example of predictive accuracy, as well as others, such as how well a qualitative finding based on parameter estimates generalizes across domains. 

Finally, we use our framework to compare the generalizability of economic models and black box machine learning methods when predicting certainty equivalents for lotteries. While economic models of risk preference have been  studied extensively from the perspectives of how well they fit economic data \citep{HarlessCamerer1994,HeyOrme1994,bruhin2010risk,bernheim2020empirical} and how their predictive performance compares to that of machine learning algorithms \citep{PeysakhovichNaecker2017,Plonskyetal,FKLM}, the existing tests have primarily been within the context of a single domain.\footnote{An important exception is \citet{Einavetal2012}, which examines how general risk preferences are across different domains of choice under uncertainty.} We offer a new perspective on how well these models transfer across domains, finding that although machine learning algorithms outperform the economic models when trained and tested on (disjoint) data generated under the same conditions, the economic models generalize better. 
Our analysis suggests the primary reason for this is not that the machine learning algorithms overfit, but rather that  economic models of risk preference are more effective at extrapolating from observed cases to new ones.

We now describe the main parts of the paper in more detail. Section \ref{sec:Setup} describes our conceptual framework, which extends the familiar notion of out-of-sample evaluation to out-of-domain evaluation. In the standard out-of-sample test, a model's free parameters are estimated on a training sample, and the predictions of the estimated model are evaluated on a test sample, where the observations in the training and test samples are disjoint but drawn from the same distribution.  We depart from this framework by supposing that the distribution of the data varies across a set of ``domains,'' e.g. different subject pools or choice frames. 

We adopt the perspective of an external \emph{analyst} who wants to assess the efficacy of a procedure that a \emph{researcher} uses to make predictions in a new domain. For example, the researcher's procedure might be to estimate an economic model on data from one domain  and use the estimated model to make predictions in another. Alternatively, the researcher's procedure might involve pooling data across  domains for estimation,  or instead training and porting a black box algorithm. The procedure's exact performance depends on the domains that are used  for estimation---the \emph{training domains}---and the domain that is realized for prediction---the \emph{target domain}. The analyst's goal is  to develop a forecast interval for the random performance of this procedure across the many possible realizations of those domains. This forecast interval can then be used to compare procedures for generalization.

Section \ref{sec:Baseline} constructs forecast intervals for the baseline setting where the distributions governing different domains are  exchangeable. This means that while there may be ex-post differences between the domains on which the model is estimated and on which it is applied, these differences are not ex-ante known to the analyst. 
We propose the following protocol: The analyst first collects data across as many domains as possible, henceforth the analyst's \emph{meta-data set}. The analyst then repeatedly splits these into training and test domains, and runs the researcher's procedure for each of these domain realizations. For example, if the researcher's procedure involves transferring an economic model, then the analyst would estimate the researcher's model on the training domains and use it to make predictions in the test domain. Intuitively, the transfer performance between  domains in the analyst's meta-data set can serve as a proxy for the transfer performance to the new (unobserved) target domain. Pooling these transfer performances across different choices of training and test domains yields an empirical distribution of transfer errors.  We show that the quantiles of this distribution can be used to compute valid finite-sample forecast intervals for the transfer error on the 
target domain. (Because in most applications, models are estimated on data from a small number of domains,   we exclusively consider finite-sample results in this paper.)

Section \ref{sec:RelaxIID}  extends our methods for the scenario where the training domains are known to be systematically different from the target domain. Specifically, we suppose that  the training domains are exchangeable, but allow the target domain to be governed by a qualitatively different distribution. We fully generalize our procedure and results when the analyst either knows or can bound  a likelihood ratio relating the target and training distributions. We further develop two partial orders for comparing how well models generalize in environments where the analyst lacks any knowledge at all about the likelihood ratio. Though this ordering is inherently incomplete, we demonstrate that it can be used to compare economic models and black box algorithms in our subsequent application. Each of these procedures (and all  other methods described in this paper) are implemented in an \texttt{R} package (\texttt{transferUQ}), available on Github.\footnote{\href{https://github.com/lihualei71/transferUQ}{https://github.com/lihualei71/transferUQ}} 

Our statistical framework and results extend the recent literature on conformal inference  
\citep[e.g. ][]{vovk2005algorithmic, tibshirani2019conformal, lei2021conformal}   and randomization inference
\citep[e.g.][]{ritzwoller2024randomization} by replacing the assumption of (weighted) exchangeable observations with that of (weighted) exchangeable domains.  Importantly, while conformal inference aims to predict a single data point from other data points, our goal is to generate forecast intervals 
on a function of multiple data points, involving both observed and unobserved data.
Our problem thus involves a multi-dimensional structure that falls outside of the standard framework and requires substantively new arguments.

Section \ref{sec:Application} applies our procedures to compare the  transferability of economic models and black box algorithms for predicting certainty equivalents for binary lotteries. In this application, we evaluate how well a model (or algorithm) predicts certainty equivalents reported by one subject pool when estimated on data from another. We evaluate two models of risk preferences (expected utility and cumulative prospect theory) and two popular black box machine learning algorithms (random forest and kernel regression). We find that although the black box algorithms outperform the economic models out-of-sample when trained and tested on data from the same domain, the economic models generalize more reliably across domains. Specifically, while the forecast intervals for the black box algorithms and economic models overlap, the forecast intervals for the black box methods are  wider, and their upper bounds are substantially higher. 

Why do the black boxes perform worse at transfer prediction? A natural explanation, based on intuition from conventional out-of-sample testing, is that black boxes are very flexible and hence learn idiosyncratic details that do not generalize across subject pools. But when we restrict the analysis to a subset of samples involving the same set of lotteries,  the economic models and black box algorithms have nearly indistinguishable forecast intervals. Alternatively, when we define the domains so that they involve disjoint lotteries, the improvement of the transfer performance of economic models over black box algorithms is even larger than in our main analysis. This tells us that black box algorithms are not universally worse at transfer, but rather their relative performance depends on specific characteristics of the transfer problem.   In particular, black boxes seem to transfer worse when the primary source of variation across samples is a shift in the marginal distribution over features (i.e., which lotteries appear in the sample), rather than a shift in the distribution of outcomes conditional on features (i.e., the distribution of certainty equivalents given fixed lotteries).

\subsection{Related Literature.} \label{sec:RelatedLit}

Our results apply for a broad class of definitions of transfer error, but our primary motivation is evaluating how well an economic model estimated on data from one domain predicts in another. \citet{Hofmanetal} gives an in-depth argument for why this important, calling for more work on the question ``how well does a predictive model fit to one data distribution generalize to another?" for social science models. This is exactly what we consider.

Predictive accuracy has long had a central role in experimental economics.\footnote{As discussed by \citet{HarlessCamerer1994}, the poor predictive performance of expected utility theory was a primary motivation for the development of  alternative models in behavioral economics. Both \citet{HarlessCamerer1994} and \citet{HeyOrme1994} provide early assessments of alternative theories on the basis of predictive performance.} While models have often been assessed based on how well they fit data from a given domain,
recent papers also examine how well predictions transfer across domains.  For example, \citet{kuzmics} estimates various game-theoretic models on $2\times 2$ normal-form games and evaluates their predictive performance on $3 \times 3$ normal-form games;  \citet{Natenzon2019} estimates discrete choice models on data for four choice menus and evaluates their predictive performance on a fifth menu; and \citet {fudenberg2024predicting} evaluates the cross-domain predictive accuracy of models of how players learn  in the infinitely-repeated prisoner's dilemma, where the domains are various sets of payoff matrices for the stage game.  This paper provides a general framework that nests these transfer problems, and develops formal statistical results for assessing transfer performance.
 
Transfer prediction is also an important consideration when comparing economic models with black box algorithms. Several papers have compared the predictive performance of machine learning algorithms with that of  more structured economic models in out-of-sample tests \citep{PeysakhovichNaecker2017,Notietal,Plonskyetal,CamererNaveSmith,FudenbergLiang,Kandori,Ke}, concluding that machine learning algorithms are more predictive. In the application we consider, black box methods turn out to be less effective at transferring predictions across domains. This result joins \citet{gechter2019evaluating}'s finding that structural models deliver better policy recommendations for conditional cash transfer policies in new contexts than black box methods do. We hope that our methods will be used to assess the transferability of economic models in other problems as well, contributing to a more comprehensive picture of how well economic models generalize.

Finally, our theoretical framework and results lie at the intersection of several literatures in economics, computer science, and statistics. These literatures consider several related but distinct objectives:  synthesizing evidence across different domains,  improving the quality of extrapolation from one domain to another, and  quantifying the extent to which insights from one domain generalize to another. 

The first objective, synthesizing results across different domains, is a particular focus of the literature on meta-analysis \citep{CardKrueger,dellavignapope,BandieraFischerPratYtsma,ImaiRutterCamerer,Vivalt}. Our goal   is instead to assess the cross-domain forecast accuracy of a model. These problems are related, and \cite{Meager2019} and \cite{Meager2022} in particular provide posterior predictive intervals for new domains in the context they consider.  Unlike our approach, the predictive intervals reported in those papers are valid only under a parametric model for the distribution of effects across domains  and a distributional assumption on the domain-specific effect estimates.

The second objective is to extrapolate results from one domain to another.  Within computer science, the literature on domain generalization (\citealt{DG1} and \citealt{DG2}) develops models that generalize well to new unseen domains \citep{DGSurvey}. Similarly, several papers within economics  (e.g., \citealt{Hotzetal2005} and \citealt{DehejiaPopElechesSamii}) use knowledge about the distribution of covariates to extrapolate out-of-domain. 
Our focus is not on developing new models or algorithms with good out-of-domain guarantees, but rather on developing forecast intervals for the out-of-domain performance of models and algorithms that are used in practice.

The third related objective in the literature is to understand the extent to which results obtained in one domain hold more generally, i.e. their external validity. Unlike  papers that assess the generalizability of insights from randomized control trials \citep[e.g.][]{Deaton,Imbens2010,dellavignalinos} or laboratory experiments \citep[e.g.][]{LevittList,AlUbaydliList},
here we consider generalizability across random domains.\footnote{Another set of papers study the generalizability of instrumental variables estimates \citep[e.g.][]{AngristFernandezVal,BertanhaImbens} and causal effects \citep[e.g.][]{PearlBareinboim2014,Tafti}.} 
Section \ref{sec:RelaxIID}  and Appendix \ref{subsec:general_dist_shift} relaxes the exchangeability assumption to allow bounded deviations from exchangeability; our results there connect to the literature on sensitivity analysis \citep[e.g.][]{rosenbaum2005sensitivity, aronow2013interval, AndrewsOster, sahoo2022learning}.

\section{ Framework}  \label{sec:Setup}

Consider a fixed procedure for extrapolating predictions across domains, e.g., estimating a structural economic model on data from one domain and using the estimated model  to make predictions in another. We adopt the perspective of an external analyst who wants to evaluate the effectiveness of  this procedure.  The analyst is not focused on extrapolation from one specific domain to another (e.g., from an American dataset to a German dataset), but would rather like to understand whether the procedure generally performs well across a class of transfer tasks (e.g., extrapolating across countries). To this end, the analyst evaluates transfer error from an ex-ante perspective without knowing which domains are used to  estimate and evaluate the model. The analyst seeks to construct forecast intervals for the procedure's error when transferring from a (random) set of training domains to a new (random) target domain. This speaks to the question of whether one procedure for extrapolation (such as transferring an estimated economic model) generally performs better than another (such as transferring a trained black box algorithm).

This section proceeds as follows: Section \ref{sec:TransferErrors} describes the extrapolation procedure that the analyst would like to evaluate, and a large class of measures for the  procedure's \emph{transfer error}. Section \ref{sec:Analyst} formalizes the analyst's problem.

\subsection{Transfer errors} \label{sec:TransferErrors}
Let $\mathcal{X}$ be a set of covariate vectors and $\mathcal{Y}$ be a set of outcomes. An \emph{observation} is a pair $(x,y) \in \mathcal{X} \times \mathcal{Y}$, and a \emph{sample} is a set of observations $S = \{(x_i,y_i)\}_{i=1}^m$. We consider samples $S_d$ indexed to \emph{domains} $d=1,2,\dots$, such as in the following  examples:

\begin{example}[Different Subject Pools]
Each sample $S_d$ corresponds to data observed for subjects from a given subject pool, where the subject pools possibly differ in their demographic characteristics. For example, $S_1$ may contain data from Caltech undergraduates, while $S_2$ contains data from a representative Prolific subject pool. \label{ex:SubjectPools}
\end{example}

\begin{example}[Different Choice Frames] Each sample $S_d$ corresponds to data collected under a particular framing of choice questions. For example, $S_1$ might contain the reported certainty equivalents for compound lotteries, and    $S_2$ the reported certainty equivalents for equivalent simple lotteries.

\end{example}

\begin{example}[Different Choice Menus] \label{Ex:ChoiceMenu} There is a finite set of goods $a_1,a_2,\dots,a_m$, and each sample $S_d$ includes observed choices from a different subset of available goods. For example, $S_1$ might contain all choices from binary menus and $S_2$ all choices from those menus that include  $a_1$.
\end{example}

\noindent For now we take these samples as given; Section \ref{sec:Analyst} lays out the underlying statistical model for how these samples are generated, which we will use to prove our results.

A researcher observes samples from some set of training domains $d \in \mathcal{T}$, and uses these samples $S_{\mathcal{T}} \equiv (S_d)_{ d\in \mathcal{T}}$ to make predictions in a new (yet unseen) target domain $d^*$. We will refer to $S_{\mathcal{T}}$ as the \emph{training samples} and to $S_{d^*}$ as the \emph{target sample}.

The number of training domains $r\equiv \vert \mathcal{T}\vert$ is a parameter of the research procedure, and should reflect what is done in practice. In economics, it is common to transfer quantitative conclusions from a single domain to another, e.g., for parameter calibration in structural models \citep{Greenwoodetal1997, McKayetal2016, Oswald2019} and extrapolation of treatment effect estimates beyond the experimental population \citep{mogstad2018identification, tipton2018review, cattaneo2021extrapolating, maeba2022extrapolation}. In this case $r=1$, and the relevant question is whether extrapolating from one sample leads to good predictions in the new domain. If instead  data is gathered from $r>1$ different domains and the observations are aggregated and used to estimate a model (as in the meta-analyses of \citealt{Meager2019,Meager2022}), the relevant question may be how well the estimated model on the aggregated data generalizes to a new domain, and $r>1$ is appropriate.

Our results apply to the following class of measures for transfer performance.

\begin{definition} A \emph{transfer error} is any quantity $e_{\mathcal{T},d^*}$ that can be written as a function of the training data $S_{\mathcal{T}}$, the target sample $S_{d^*}$, and potentially an independent noise variable.
\label{def:TransferError}
\end{definition}

Our leading examples are transfer errors that measure how well a fixed model or algorithm transfers across domains. That is, suppose the training samples $S_{\mathcal{T}}$ are used to select a \emph{prediction rule} $f_{S_{\mathcal{T}}}: \mathcal{X} \rightarrow \mathcal{Y}$, e.g., by estimating a parametric model or by fitting a black box algorithm.\footnote{That is, let $\mathcal{S}$ denote the set of all finite sets of finite samples, and let $\mathcal{Y}^{\mathcal{X}}$ be the set of all prediction rules. Then a ``model'' is a mapping $\rho: \mathcal{S} \rightarrow \Delta(\mathcal{Y}^{\mathcal{X}})$ and we write $f_{S_{\mathcal{T}}} = \rho(S_{\mathcal{T}})$ for the realized prediction rule.} The accuracy of the prediction rule is evaluated using a loss function $\ell: \mathcal{Y} \times \mathcal{Y} \rightarrow \mathbb{R}_+$, where
\[e(f,S) = \frac{1}{\# S} \sum_{(x,y) \in S} \ell(f(x),y)\]
denotes the average loss when using $f$ to predict $y$ given $x$ for observations $(x,y) \in S$. One specification of a transfer error is then 
\begin{equation} \label{eq:TransferError}
    e_{{\mathcal{T}},d^*} = e(f_{S_{\mathcal{T}}},S_{d^*})
\end{equation}
i.e., the raw error of the model when it is estimated on the training samples and used to predict outcomes in the target sample.

\begin{example}[Transferring Models of Risk Preferences] The covariates $\mathcal{X}$ describe different lotteries, i.e., each covariate vector $x$ includes a description of (say) two possible prizes and their corresponding probabilities. The outcome $y$ is the average willingness-to-pay for this lottery. A firm acquires willingness-to-pay data from consumers in Illinois for a given set of lotteries, and uses this data to estimate a model of risk preferences, e.g., estimating parameter values for an Expected Utility model with CRRA preferences. The firm then uses this estimated model to predict willingness-to-pay from consumers in California for a different set of lotteries. The measure in (\ref{eq:TransferError}) assesses the accuracy of those predictions. \label{ex:Risk}
\end{example}

Any normalization of equation (\ref{eq:TransferError}) with respect to a function of the target sample is also a transfer error. For example, we might normalize (\ref{eq:TransferError}) with respect to the in-sample error of the model when trained on the target sample, 
\begin{equation} \label{def:TransferDeterioration}
e_{\bold{T},n+1}=\frac{ e(f_{S_\bold{T}},S_{n+1})} {e(f_{S_{n+1}},S_{n+1})}.
\end{equation}
This quantity reveals how much less accurate the model is than if it had been directly trained on the target sample. 

\begin{example}[The Value of Re-Estimating Diffusion Models] The covariates $\mathcal{X}$ describes the network of relationships across households in a village, and the identity of households which are seeded with information about a microfinance program. The outcome $y$ is the average takeup rate of the program across households. A development economist observes the takeup decisions in a single village in India following an experiment in which certain households are seeded with information about the program. The economist uses this data to estimate a structural model of information diffusion, and then  predicts the average takeup rate  in a new village using the estimated model. The transfer error in (\ref{def:TransferDeterioration}) assesses how much less accurate this prediction is compared to if the economist could re-estimate the structural model on data from this new village.
\end{example}

Although we focus on the transfer errors defined in (\ref{eq:TransferError}) and (\ref{def:TransferDeterioration}), 
Definition \ref{def:TransferError} is substantially broader. Appendix \ref{app:OtherTransfer} describes several other specifications of transfer errors, including the stability of parameters and errors in counterfactual predictions.

\subsection{The analyst's problem} \label{sec:Analyst}

We now consider the perspective of an external analyst, who would like to evaluate the transfer guarantees of the procedure described above.  Rather than assessing the transfer error $e_{\mathcal{T},d^*}$ for a specific set of training domains $\mathcal{T}$ and  target domain $d^*$, the analyst considers a random version of this quantity, where the samples used for training and evaluation of the model are not yet known. In Example \ref{ex:Risk}, this corresponds to an analyst who is interested in how well the CRRA model transfers across arbitrary locations,  as opposed to from  one specific location to another.

Formally, the analyst has access to \emph{metadata} consisting of $n$ samples
\[\bold{M} = \{S_1, \dots, S_{d}, \dots  S_n\}.\]
We assume that $n>r$; that is, the analyst can collect a larger number of samples than were used by the researcher. The analyst models the researcher's set $\mathcal{T}$ of $r$ training domains as drawn uniformly over all subsets of $\{1, \dots,n\}$ of size $r$. We use $\bold{T}$ to mean the random variable whose realization is $\mathcal{T}$, so that $S_{\bold{T}}=(S_d)_{d\in \bold{T}}$ is the researcher's (random) vector of training samples. The target domain (on which predictions will be made) is a final sample $S_{n+1}$, which, unlike the metadata, is not observed by the analyst.  The quantity of interest is $e_{\bold{T},n+1}$, i.e., the random transfer error when the researcher extrapolates predictions from $S_{\bold{T}}$ to $S_{n+1}$. Figure \ref{fig:illustration} depicts this transfer error for the model transfer specification of (\ref{eq:TransferError}).

\begin{figure}%[H]
\begin{center}
    \includegraphics[scale=0.4]{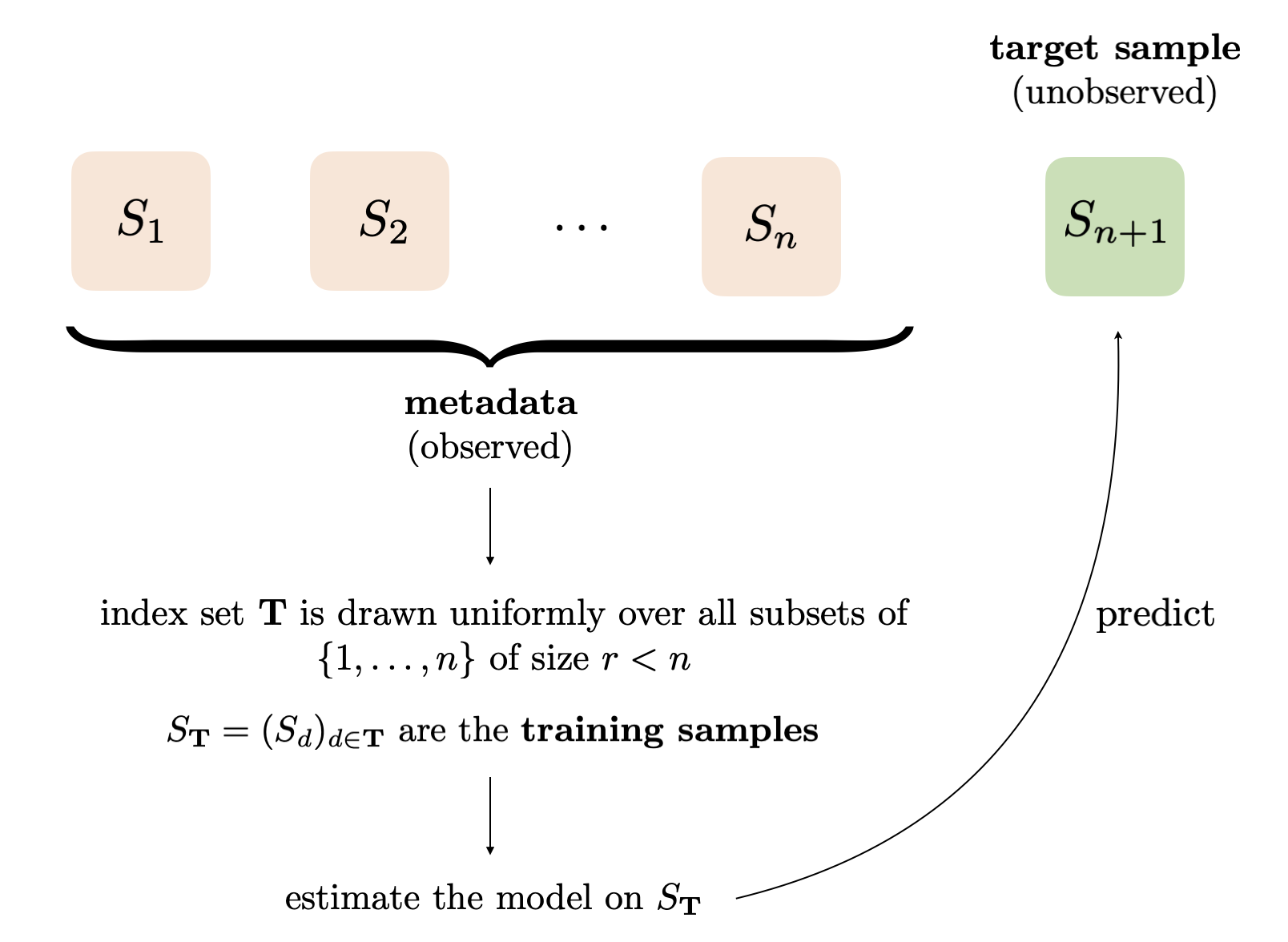}
\end{center}
\caption{This figure depicts the transfer error $e_{\bold{T},n+1}$ as defined according to (\ref{eq:TransferError}). It measures the prediction error of a model estimated on the samples $S_{\bold{T}}$ and evaluated on the sample $S_{n+1}$, where $S_{\bold{T}}$ consists of $r$ random training samples from the metadata, and $S_{n+1}$ is an unobserved sample from a new domain.} \label{fig:illustration}
\end{figure}

The analyst's goal is to develop forecast intervals for the transfer error $e_{\bold{T},n+1}$, i.e., interval-valued functions of the meta-data $\bold{M}$ which cover  $e_{\bold{T},n+1}$ with the prescribed probability, regardless of the distribution $\mu$ that governs samples across domains. Section \ref{sec:Baseline} provides these intervals for a baseline setting in which domains are exchangeable, and Section \ref{sec:RelaxIID} generalizes them to a setting in which the observed domains are systematically different from the target domain.

\section{Baseline Setting: Exchangeability
} \label{sec:Baseline}

We begin with a baseline setting in which the different domains are governed by exchangeable distributions. Specifically, we suppose that in the analyst's statistical model, the samples $S_1,S_2, \dots$ are generated in the following way: 

\begin{assumption} There is a fixed (but unknown) meta-distribution  $\mu\in\Delta(\mathcal{P} \times\mathbb{N})$ over joint distributions $\mathcal{P} \equiv \Delta(\mathcal{X} \times \mathcal{Y})$ and sample sizes $\mathbb{N}$, where each sample $S_d$ is generated by first drawing a distribution and sample size $(P_d,m_d) \sim \mu$, and then independently drawing $m_d$ observations  $(x,y)$ from $P_{d}$.\footnote{All of our results  extend unchanged if samples from the different domains are ex-ante exchangeable rather than i.i.d..} \label{assp:IID}\end{assumption}

In Examples \ref{ex:SubjectPools}-\ref{Ex:ChoiceMenu}, this assumption implies that (from the analyst's perspective) the subject pools, choice frames, or choice menus differentiating the samples are themselves drawn i.i.d. from a fixed distribution. Assumption \ref{assp:IID} is standard in conformal inference \citep{vovk2005algorithmic}, permutation testing \citet{romano1990behavior}, and randomization inference \citep{ritzwoller2024randomization}.\footnote{It can also be understood as a  Bayesian hierarchical model \citep{Meager2019, Meager2022} or a version of cluster sampling \citep{liang1986longitudinal,bugni2023inference}.If framed in this way, the analyst's goal is to do predictive inference for new clusters. When $\mu$  assigns probability $1$  to  a single distribution in $p\in\mathcal{P}$ or when $\mu$ assigns probability $1$ to $m=1$,  this reduces to i.i.d.\ sampling of observations from a fixed joint distribution, but our focus is on settings where neither of these is the case. } 
 In contrast,  the literature on external validity (see Section \ref{sec:RelatedLit})  typically assumes that the distributions governing behavior in different domains are close in some distance metric \citep{adjaho2023externally}, share a common support over $\mathcal{X}$ or $\mathcal{Y}$ \citep{sahoo2022learning,lei2023policy}, or can be estimated using background covariates \citep{tipton2018review}. Relative to these assumptions, our approach has the advantage of allowing for arbitrary and unknown relationships between the realized distributions governing domains, but 
  it  rules out ex-ante predictable patterns in how the joint distribution varies across samples (such as time trends). 

  Section \ref{sec:Procedure} presents our forecast intervals for this setting, and Section \ref{sec:ResultsIID} proves the validity and tightness of these intervals.
  
\subsection{Baseline
procedure} \label{sec:Procedure}
Thanks to Assumption \ref{assp:IID},  the observed samples in the metadata, $\{S_1, \dots, S_n\}$, can act as surrogates for the unseen target sample, $S_{n+1}$. 
  As before, let $e_{\mathcal{T},d}^\bM$ denote the (observed) transfer error from any selection of training samples $\mathcal{T} \subseteq \{1,\dots,n\}$ to any surrogate target sample $d \in \{1,\dots,n\}\backslash \mathcal{T}$ from the metadata (where we now make the dependence of this quantity on $\bM$ explicit). 
We use  $\bbT_{r+1,n}$ to denote the set of $\frac{n!}{(n-r-1)!}$ unique pairs $(\mathcal{T},d)$ that can be constructed in this way.
Then 
\begin{equation} \label{def:FM}
F_{\bM}=\frac{(n-r-1)!}{n!} \sum_{(\T,d) \in \bbT_{r+1,n}} \delta_{e_{\T,d}^{\bM}}
\end{equation}
is the empirical distribution of transfer errors in the pooled sample 
$\left\{ e^{\bM}_{\T,d} : (\T,d) \in \bbT_{r+1,n}\right\}$ as we vary which samples in the metadata are used for training and testing. (Throughout $\delta$ denotes the 
Dirac measure).  In the case where $r=1$, so that a single sample is used for training, the observed transfer errors can be represented as a matrix as depicted in Figure \ref{fig:matrix}, and $F_{\bold{M}}$ is their empirical distribution.

\begin{figure}[h]
    \centering
    \includegraphics[scale=0.23]{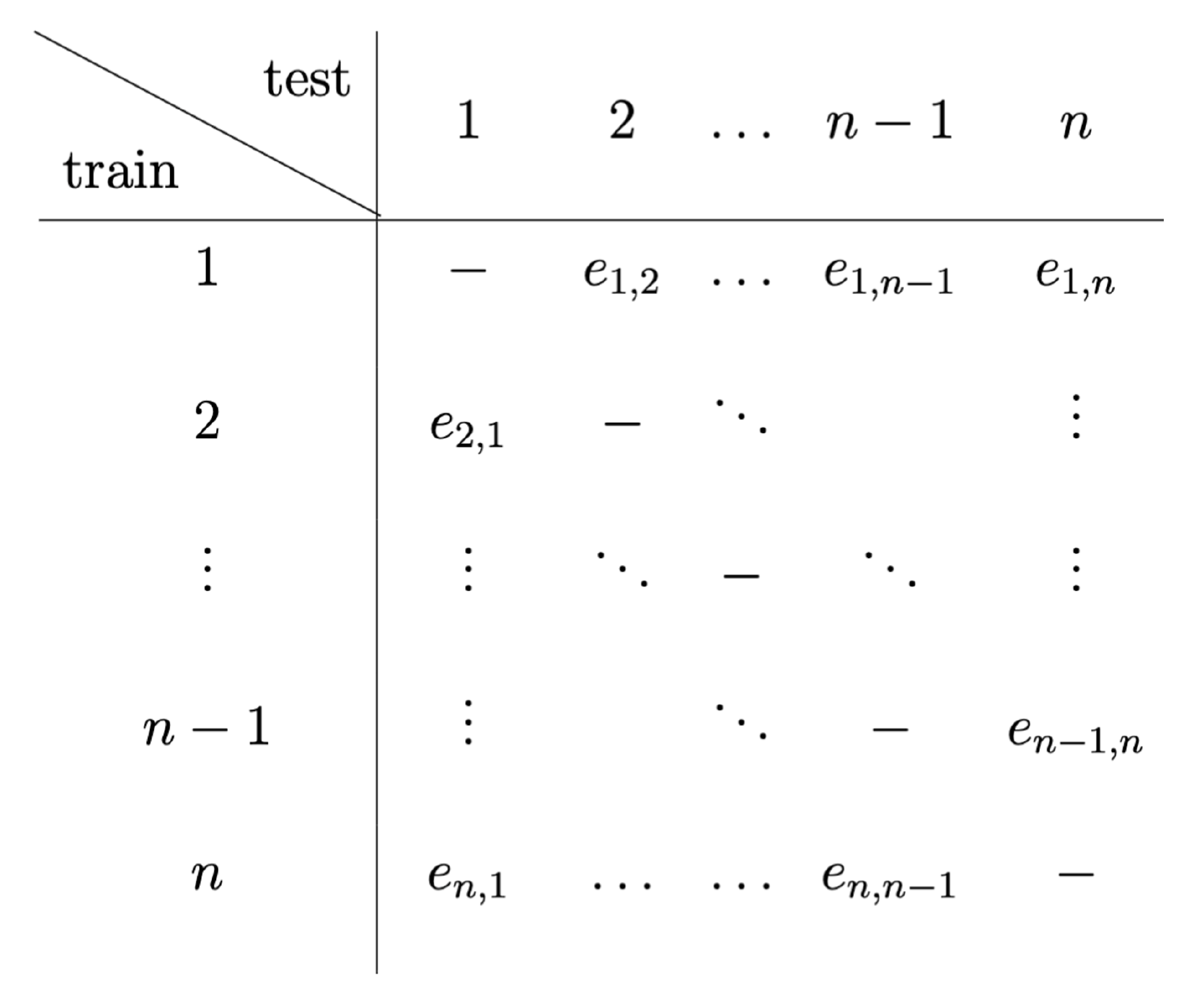}
    \caption{$e_{d,d'}$ is the transfer error from sample $S_d$ to $S_{d'}$.}
    \label{fig:matrix}
\end{figure}
\begin{definition}[Upper and Lower Quantiles] 
For any distribution $P$ let 
$\overline{Q}_{\tau}(P)=\inf\{b:P((-\infty, b]) \ge \tau \}$ 
and 
$\underline{Q}_{\tau}(P)=\sup\{b: P([b, \infty)) \ge 1-\tau\}$
denote the upper and lower $\tau$th quantiles, respectively.
\end{definition}
These  quantiles coincide for continuously distributed variables with connected support.
\begin{definition}[Quantiles of $F_\bM$] For any $\tau \in (0,1)$, let 
$\overline{e}_\tau^\bM \equiv \overline{Q}_{\tau}(F_\bM)$ and
$\underline{e}_\tau^\bM \equiv \underline{Q}_{1-\tau}(F_{\bM})$ be the $\tau$th upper quantile and $(1-\tau)$th lower quantile of the empirical distribution of transfer errors in the pooled sample.  \end{definition}
Our proposed forecast interval for the transfer error on the target sample is $[\underline{e}_\tau^\bM,\overline{e}_\tau^\bM]$.

%\section{Theoretical Results} \label{sec:MainResults}

\subsection{ Results} \label{sec:ResultsIID}

We first prove that $[\underline{e}_\tau^\bM,\overline{e}_\tau^\bM]$ is indeed a valid forecast interval.

\begin{proposition} \label{prop:pooled, percentiles}  For any $\tau \in (0,1)$,
  \begin{equation} \label{eq:OneSidedCI}
    \P\lb e_{\bT , n+1}\le \bar{e}_{\tau}^\bM\rb\ge \tau\left(\frac{n-r}{n+1}\right),
  \end{equation}
  and
  \begin{equation*}
    \P\lb e_{\bT , n+1}\in \left[\underline{e}_{\tau}^\bM, \bar{e}_{\tau}^\bM\right]\rb\ge (2\tau - 1)\left(\frac{n-r}{n+1}\right) .
  \end{equation*}  
\end{proposition}

\noindent Thus $\left(-\infty, \overline{e}_\tau^\bM\right]$ is a level-$\left(\frac{\tau(n-r)}{n+1}\right)$  one-sided forecast interval for  $e_{\bT,n+1}$,  
and  $\left[\underline{e}_{\tau}^\bM, \bar{e}_{\tau}^\bM\right]$ is a level-$\left((2\tau - 1)\left(\frac{n-r}{n+1}\right)\right)$ forecast interval for $e_{\bT,n+1}$.

Parameter $\tau$ influences the width of the forecast interval, where larger choices of $\tau$ lead to wider forecast intervals with higher confidence guarantees. Parameter $r$ determines how many samples in the meta-data are used for training versus testing. As discussed in Section \ref{sec:TransferErrors}, $r$ is determined by the the research procedure under evaluation.\footnote{We expect that in general, larger choices of $r$ will lead to lower but wider forecast intervals, since the model is estimated on a larger quantity of data, but there are fewer samples with which to evaluate the performance of the estimated model.}

The number of samples $n$ and the sizes of these samples $(m_d)_{d=1}^n$ enter into our result in different ways: Increasing the number of observed domains $n$, holding fixed the distribution over sample sizes within each domain,  does not change the distribution of $e_{\bT, n+1}$ but instead allows this distribution to be estimated more precisely.
In contrast,  increasing the number of observations per domain changes the distribution of $e_{\bT, n+1}$ and corresponds to the measurement of a different quantity. For example, in the limit of infinitely many  observations per sample, the error $e_{\bT,n+1}$ measures how well the best predictor from the model class in the training domains transfers across domains, while if the number of observations is small,  $e_{\bT,n+1}$ measures how well an imperfectly estimated model transfers.

 The next result shows that the guarantees in Proposition \ref{prop:pooled, percentiles} are tight to $O(1/n)$. We use $\bbT_{s,t}$ to denote the set of all vectors of length $s$ that consist of distinct elements from $\{1,\dots,t\}$.

\begin{claim}\label{thm:tight_coverage}
  Assume that 
$\left(e^\bM_{\T,d}: (\T,d) \in \bbT_{r+1, n+1}\right)$ almost surely  has no ties. Then
  \begin{equation*}
    \P\lb e_{\bT , n+1}\le \bar{e}^\bM_{\tau}\rb\le \tau \left(\frac{n-r}{n+1}\right) + \frac{r+1}{n+1} + \frac{(n-r)!}{(n+1)!}.
  \end{equation*}
  and
  \begin{equation*}
    \P\lb e_{\bT , n+1}\in \left[\underline{e}^\bM_{\tau}, \bar{e}^\bM_{\tau}\right]\rb\le (2\tau-1)\left(\frac{n-r}{n+1}\right) + \frac{r+1}{n+1} + \frac{(n-r)!}{(n+1)!}.
  \end{equation*}  
\end{claim}

To gain intuition for the intervals in Proposition \ref{prop:pooled, percentiles}, fix a realization of the unordered set  $\{S_1, \dots, S_n, S_{n+1}\}$. Because all samples are exchangeable by assumption, the realization of $e_{\bT,n+1}$ (conditional on $\{S_d\}_{d=1}^{n+1}$) is a uniform draw from
\begin{equation} \label{eq:PooledSample}
\left\{e_{\T,d}^{\bM}: (\T,d) \in \bbT_{r+1,n+1}\right\}.
\end{equation}
If we let $e^*_\tau$ denote the upper $\tau$-th quantile of this empirical distribution, then by definition 
\begin{equation} \label{eq:FullSample}
\mathbb{P}\left(e_{\bT,n+1} \leq e^*_\tau \mid \{S_d\}_{d=1}^{n+1}\right) \geq \tau.
\end{equation}
In the case $r=1$ where precisely one sample is used for training, the set of pooled errors (\ref{eq:PooledSample}) is the shaded cells in Figure \ref{fig:Matrix} (either yellow or blue), and the inequality in (\ref{eq:FullSample}) says that the probability that the value of a randomly drawn  cell falls below the $\tau$th upper quantile of  cells is at least $\tau$.

\begin{figure}[h]
    \centering
    \renewcommand{\arraystretch}{1.5}
    \begin{tabular}{c|M{10mm}M{10mm}M{10mm}M{10mm}M{10mm}M{10mm}}
        \hbox{\diagbox{\text{train}}{\text{test}}} & 1 & 2 & \dots & n-1 &  n & n+1 \\
        \hline 
        1 &  {-} & \cellcolor{myyellow}{$e_{1,2}$} & \cellcolor{myyellow}{$\dots$} & \cellcolor{myyellow}{$e_{1,n-1}$} & \cellcolor{myyellow}{$e_{1,n}$} & \cellcolor{myblue}{\textcolor{white}{$e_{1,n+1}$}} \\
        2 & \cellcolor{myyellow}{$e_{2,1}$} & {-}  & \cellcolor{myyellow}{$\ddots$} & \cellcolor{myyellow}{}&  \cellcolor{myyellow}{$\vdots$} & \cellcolor{myblue}{\textcolor{white}{$\vdots$}} \\
        $\vdots$ & \cellcolor{myyellow}{$\vdots$} & \cellcolor{myyellow}{$\ddots$} & {-}  & \cellcolor{myyellow}{$\ddots$} & \cellcolor{myyellow}{$\vdots$} & \cellcolor{myblue}{\textcolor{white}{$\vdots$}} \\
        $n-1$ & \cellcolor{myyellow}{$\vdots$} & \cellcolor{myyellow}{} & \cellcolor{myyellow}{$\ddots$} & {-} &  \cellcolor{myyellow}{$e_{n-1,n}$}  & \cellcolor{myblue}{\textcolor{white}{$\vdots$}} \\
        $n$ & \cellcolor{myyellow}{$e_{n,1}$} & \cellcolor{myyellow}{$\dots$} & \cellcolor{myyellow}{$\dots$}& \cellcolor{myyellow}{$e_{n,n-1}$} &  {-}  & \cellcolor{myblue}{\textcolor{white}{$e_{n,n+1}$}} \\
        $n+1$ & \cellcolor{myblue}{\textcolor{white}{$e_{n+1,1}$}} & \cellcolor{myblue}{\textcolor{white}{$\dots$}} & \cellcolor{myblue}{\textcolor{white}{$\dots$}} & \cellcolor{myblue}{\textcolor{white}{$\dots$}} & \cellcolor{myblue}{\textcolor{white}{\textcolor{white}{$e_{n+1,n}$}}} &  {-} \\
    \end{tabular}
    \caption{Transfer errors when training on one domain (row) and testing on another (column).}
    \label{fig:Matrix}
\end{figure}

The analyst does not observe the target sample $S_{n+1}$,  and so does not know   $e^*_\tau$. We instead use $\overline{e}_\tau^\bM$, the $\tau$th upper quantile of the pooled sample of errors when transferring across samples in $\bM$, to construct the forecast intervals. In Figure \ref{fig:Matrix}, the probability that
$e_{\bT,n+1} \leq \overline{e}_\tau^\bM$ is the probability that the value of a randomly drawn shaded cell (yellow or blue) falls below the $\tau$th quantile of the yellow cells. By a straightforward counting argument, 
\[\mathbb{P}\left(e_{\bT,n+1} \leq \overline{e}_\tau^\bM \mid \{S_i\}_{i=1}^{n+1}\right) \geq  \tau {n \choose r+1}/{n+1 \choose r+1} = \tau \left(\frac{n-r}{n+1}\right).\]
Applying the law of iterated expectations (with respect to the sample $\{S_i\}_{i=1}^{n+1}$) yields the one-sided forecast interval in (\ref{eq:OneSidedCI}). The proof for the two-sided forecast interval follows a similar logic but is more involved, see Appendix \ref{app:ProofThm1} for details.

\section{Beyond Exchangeability} \label{sec:RelaxIID}

Our results so far assume that the distributions governing the different samples $S_d$ are themselves independent and identically distributed. This assumption is not always appropriate. For example, suppose variation in domains corresponds to variation over locations, and the samples in the metadata (but not the target sample) are from experiments run at locations chosen by experimenters. If there is selection bias over where experiments are run---for example, if  the observed sites are chosen based on characteristics correlated with effect sizes (as \citet{Allcott} found in the Opower energy conservation experiments)---then it may be that the the target sample has fundamentally different properties from anything that is observed in the metadata. We thus now relax Assumption \ref{assp:IID} to allow the distribution governing the training samples and the distribution governing the target sample to be drawn from different meta-distributions.
Specifically, suppose that the analyst's metadata consists of samples $S_1, \dots, S_n \sim_{iid} \mu$ as in our main model, but $S_{n+1}$ is independently drawn from some other density $\nu$.   Let
\[\omega(S) = \frac{\nu(S)}{\mu(S)} \]
denote their likelihood ratio. We initially assume this likelihood ratio is known by the analyst (although $\nu$ and $\mu$ need not be), and subsequently consider weakenings of this assumption. As before, $e_{\bT,n+1}$ is the transfer error when training on $r$ samples drawn uniformly at random from $\{S_1,\dots,S_n\}$, and testing on $S_{n+1}$.

The following subsections consider successively weaker assumptions about what the analyst knows about the likelihood ratio $\omega$. Sections \ref{sec:KnowOmega} and \ref{sec:BoundOmega} respectively generalize our previous results when the analyst either knows $\omega$ or can bound it.  Section \ref{sec:PartialOrder} provides two ways of comparing the generalizability of models in environments where the analyst knows nothing about $\omega$. 
\subsection{The analyst knows the likelihood ratio $\omega$} \label{sec:KnowOmega}

We again construct a forecast interval for $e_{\bT,n+1}$ using the pooled sample of transfer errors across samples in the metadata, that is, $\left\{ e^{\bM}_{\T,d} : (\T,d) \in \bbT_{r+1,n}\right\}$. Different from the previous section, we no longer assign uniform weights to each $e^{\bM}_{T,d}$.  Intuitively, under our previous i.i.d.\ assumption, each sample in the metadata was equally representative of the training and target distributions, but in this relaxed model, whether a sample $S_d$ is more representative of the training or testing distribution depends on its relative likelihood under $\nu$ and $\mu$.  

A crucial quantity is the following:

\begin{definition} For every domain $d \in \{1, \dots,n\}$, define
  \begin{align}\label{eq:omegak'}
    W_d = \frac{(n-r-1)!}{(n-1)!}\frac{\omega(S_{d})}{\sum_{j=1}^{n}\omega(S_{j})}.
  \end{align}
  \end{definition}
  
  To interpret this quantity, consider an alternative data-generating process for the metadata where for some permutation $\pi:\{1, \dots, n\}\rightarrow \{1,\dots,n\}$, the samples $S_{\pi(1)},\dots,S_{\pi(n-1)} \sim_{iid} \mu$ while $S_{\pi(n)} \sim \nu$. Fix a realization of the metadata $(S_1,\dots,S_n)$, and suppose the analyst does not observe the permutation $\pi$. Let $\Pi$ denote the set of all permutations on $\{1, \dots, n\}$, and for any vector of sample indices $(t_1,\dots,t_r,d)$ let 
  \[\Pi_{(t_1,\dots,t_r,d)} = \{\pi \in \Pi : (\pi(1), \dots, \pi(r)) = (t_1,\dots, t_r) \mbox{ and } \pi(n) = d\}\]
  denote the permutations that specify  $(t_1,\dots,t_r)$ for training and $d$ as the target. Then conditional on a realization of the metadata $(S_1,\dots,S_n)$, the probability that $(S_{t_i})_{i=1}^r$ are the training samples and $S_d$ is the test sample is\footnote{This is a special case of weighted exchangeability; see \citet{tibshirani2019conformal}. The 
   results in this subsection continues to hold if the domains are not independent but satisfy the weighted exchangeability condition, which is more general but harder to interpret.}
 \begin{align*}
     \frac{\sum_{\pi \in \Pi_{(t_1,\dots,t_r,d)}} \left(\nu(S_{\pi(n)}) \cdot \prod_{j=1}^{n-1} \mu(S_{\pi(j)})\right)}{\sum_{\pi \in \Pi} \left(\nu(S_{\pi(n)}) \cdot \prod_{j=1}^{n-1} \mu(S_{\pi(j)})\right) } & = \frac{\sum_{\pi \in \Pi_{(t_1,\dots,t_r,d)}} \omega(S_{\pi(n)})}{\sum_{\pi \in \Pi}\omega(S_{\pi(n)})} \\
     & = \frac{(n-r-1)! \cdot  \omega(S_d)}{(n-1)! \cdot \sum_{j=1}^n \omega(S_j)} = W_d.
 \end{align*}
 This quantity depends only on the identity of the target sample $d$, and not on the identity of the training samples $t_1, \dots, t_r$. Finally, let
      \[F^\omega_{\bM} = \sum_{(\T,d)\in \bbT_{r+1,n}}W_d\cdot \delta_{e^{\bM}_{\T,d}}\]
      be the weighted empirical distribution of transfer errors, where each sample $d$ is weighted according to $W_d$.    
When the two meta-distributions $\mu$ and $\nu$ are identical as in our main model, then
$W_d \equiv (n-r-1)! / n!$
for every domain $d$, so the distribution $F^\omega_{\bM}$ is simply $F_{\bM}$ as defined in (\ref{def:FM}).

  \begin{definition}[Quantiles of $F_{\bM}^\omega$]
     For any likelihood ratio $\omega(\cdot)$ and quantile $\tau \in (0,1)$, define
    $\bar{e}_{\tau}^{\bM,\omega} = \overline{Q}_{\tau}\lb F^\omega_{\bM} \rb$
    and
    $\underline{e}_{\tau}^{\bM,\omega} = \underline{Q}_{1-\tau}\lb F_{\bM}^\omega \rb$ to be, respectively, the $\tau$th upper quantile and $(1-\tau)$th lower quantile of the weighted distribution of transfer errors in the pooled sample.
  \end{definition}

\begin{theorem}\label{thm:coverage_bias}
For any $\tau \in (0,1)$,
    \[\P\lb e_{\bT , n+1}\le \bar{e}_{\tau}^{\bM,\omega}\rb\ge \tau \cdot \frac{n-r}{n}\E\left[\frac{\sum_{j=1}^{n}\omega(S_{j})}{\sum_{j=1}^{n+1}\omega(S_{j})}\right],\]
    and
    \[\P\lb e_{\bT , n+1}\in \left[\underline{e}_{\tau}^{\bM,\omega}, \bar{e}_{\tau}^{\bM,\omega}\right]\rb\ge (2\tau - 1) \cdot \frac{n-r}{n}\E\left[\frac{\sum_{j=1}^{n}\omega(S_{j})}{\sum_{j=1}^{n+1}\omega(S_{j})}\right].\]
    Furthermore, if
  $\left(e^\bM_{\T,d}: (\T,d) \in \bbT_{r+1, n+1}\right)$ almost surely has no ties, then
    \[\P\lb e_{\bT , n+1}\le \bar{e}_{\tau}^{\bM,\omega}\rb\le  1 - \E\left[ \lb (1 - \tau)\frac{n-r}{n} - \frac{(n-r)!}{n!}\frac{\max_{k\le n}\omega(S_k)}{\sum_{j=1}^{n}\omega(S_j)}\rb\frac{\sum_{j=1}^{n}\omega(S_j)}{\sum_{j=1}^{n+1}\omega(S_j)}\right],\]
    and
    \[\P\lb e_{\bT , n+1}\in \left[\underline{e}_{\tau}^{\bM,\omega}, \bar{e}_{\tau}^{\bM,\omega}\right]\rb\le  1 - \E\left[\lb 2(1 - \tau)\frac{n-r}{n} - \frac{(n-r)!}{n!}\frac{\max_{k\le n}\omega(S_k)}{\sum_{j=1}^{n}\omega(S_j)}\rb\frac{\sum_{j=1}^{n}\omega(S_j)}{\sum_{j=1}^{n+1}\omega(S_j)}\right].\]
  \end{theorem}

\noindent This result strictly generalizes Proposition \ref{prop:pooled, percentiles} and Claim \ref{thm:tight_coverage}, since when $w(\cdot)$ is the identity then $\overline{e}_\tau^{\bM,\omega} = \overline{e}_\tau^\bM$ and $\underline{e}_\tau^{\bM,\omega} = \underline{e}_\tau^\bM$, and the bounds in this theorem reduce to those given in Proposition \ref{prop:pooled, percentiles}.

\subsection{The analyst does not know $\omega$ but can bound it} \label{sec:BoundOmega}
We next extend our results when the analyst does not know the likelihood ratio function $\omega$ precisely,  but---as in the literature on sensitivity analysis \citep{rosenbaum2005sensitivity}---knows that it admits an upper and lower bound.

\begin{definition}[Bounded Likelihood-Ratios]  For any $\Gamma\geq 1$, let $\mathcal{W}(\Gamma)$ be the class of density ratios that satisfy
$
    \omega(S)\in [\Gamma^{-1}, \Gamma]$ for all samples $S$.
\end{definition}

 Define the following worst case bounds for $\underline{e}_{\tau}^{\bM,\omega}$ and $\overline{e}_{\tau}^{\bM,\omega}$:
\begin{equation} \label{eq:WorstCase}
\overline{e}_{\tau}^{\bM}(\Gamma) = \sup_{\omega\in \mathcal{W}(\Gamma)} \bar{e}_{\tau}^{\bM,\omega}, \quad \underline{e}_{\tau}^{\bM}(\Gamma) = \inf_{\omega\in \mathcal{W}(\Gamma)} \bar{e}_{\tau}^{\bM,\omega}
\end{equation}
As shown in Appendix \ref{app:Algorithms}, these quantities can be computed from data in $O(n^{r+1})$ time. 
  \begin{corollary} \label{cor:coverage_bias} 
Suppose $\omega \in \mathcal{W}(\Gamma)$. Then
    \[\P\lb e_{\bT , n+1}\le \overline{e}_{\tau}^{\bM}(\Gamma)\rb\ge \tau \left(\frac{n-r}{n + \Gamma^2}\right),\]
    and
    \[\P\lb e_{\bT , n+1}\in [\underline{e}_{\tau}^{\bM}(\Gamma), \overline{e}_{\tau}^{\bM}(\Gamma)]\rb\ge (2\tau - 1)\left(\frac{n-r}{n + \Gamma^2}\right).\]
  \end{corollary}
  
\subsection{The analyst knows nothing about $\omega$} \label{sec:PartialOrder}

Finally, we provide two ways for comparing the transferability of two models $i=1,2$ when the analyst cannot bound $\omega$. We do not provide formal results about these orders, but show that they have bite in our subsequent application (see Section \ref{sec:Application}).

Let $\overline{e}_{i,\tau}^{\bM}(\Gamma)$ and $\underline{e}_{i,\tau}^{\bM}(\Gamma)$ again denote the worst case bounds for model $i$, as defined in (\ref{eq:WorstCase}).
 
\begin{definition}[Worst-Case Dominance] Say that model $1$ worst-case-upper-dominates model $2$ at the $\tau$-th quantile if 
\[\overline{e}_{1,\tau}^{\bM}(\Gamma)  \leq \overline{e}_{2,\tau}^{\bM}(\Gamma)  \quad \forall \Gamma \in [1,\infty).\]
\end{definition}

\noindent That is, model 1 worst-case-upper-dominates model 2 at the $\tau$-th quantile if for every $\Gamma$, the worst-case upper bound for model 1 exceeds the  worst-case for upper bound for model 2. 

We can strengthen this  comparison by requiring the upper bound of the forecast interval for model 1 to be smaller than the upper bound of the forecast interval for model 2 pointwise for each $\omega \in \mathcal{W}(\Gamma)$, rather than simply comparing worst-case upper bounds. 

\begin{definition}[Everywhere Dominance] Say that model 1 everywhere-upper-dominates model 2 at the $\tau$-th quantile if
\[\overline{e}_{1,\tau}^{\bM, \omega} \leq \overline{e}_{2,\tau}^{\bM, \omega} \quad \forall \Gamma >1 \forall \omega \in \mathcal{W}(\Gamma).\]

\end{definition}

Many decision rules will not be comparable under either of these definitions, but we show they are empirically relevant in our application. The even stronger requirement that  $\overline{e}_{1,\tau}^{\bM,\omega} \leq \underline{e}_{2,\tau}^{\bM,\omega}$, i.e., that the upper bound of model 1's forecast interval is smaller than the lower bound of model 2's forecast interval, is likely too stringent to be useful in practice.\footnote{This stronger order has bite only when the transfer error for model 1 across ``the most dissimilar" training and testing domains is lower than the transfer error for model 2 for ``the most similar" training and testing domains.}

\section{Application} \label{sec:Application}

To illustrate our methods, we evaluate the transferability of predictions  of certainty equivalents for binary lotteries across domains that have different subject pools and also differ in other ways. We focus on this application for several reasons: First, there are many public data sources that we can use to construct our metadata.  Second, the associated economic models have been extensively examined from the perspective of predictive performance  \citep{HarlessCamerer1994,HeyOrme1994,bruhin2010risk,bernheim2020empirical}, and recent work evaluates how well these models predict relative to black box algorithms \citep{PeysakhovichNaecker2017,Plonskyetal,FKLM}. Finally, as \citet{Einavetal2012} points out, given how models of risk preferences are often used in practice, it is also important to evaluate how well they transfer across domains. We thus view this application as a natural setting to illustrate our methods.

Section \ref{sec:Data} describes our metadata, and  Section \ref{sec:DecisionRules} describes the decision rules  we consider. Section \ref{sec:WithinDomain} conducts ``within-domain" out-of-sample tests, where the training and test data are drawn from the same domain. Section \ref{sec:Transfer} compares transfer performance across domains by constructing forecast intervals for three different definitions of  transfer error. 

\subsection{Data} \label{sec:Data}

Our metadata consists of samples of certainty equivalents from 44 subject pools, which we treat as the domains. These data are drawn from 14 papers in experimental economics, with twelve papers contributing one sample each, one paper contributing two, and a final paper (a study of risk preferences across countries) contributing 30 samples. Our samples range in size from 72 observations to 8906 observations, with an average of 2752.7 observations per sample.\footnote{Online Appendix \ref{app:DataSources} describes our data sources in more detail.} Besides the difference in subject pools, these samples may differ in  other details, such as whether the lotteries were restricted to the gain domain. We convert all prizes to dollars using purchasing power parity exchange rates (from \citealt{OECD})  in the year of the paper's publication 

Within each sample, observations take the form $(z_1,z_2,p;y)$, where $z_1$ and $z_2$ denote the possible prizes of the lottery (and we adopt the convention that $\vert z_1 \vert > \vert z_2 \vert$), $p$ is the probability of $z_1$, and $y$ is the reported certainty equivalent by a given subject.  Thus our feature space is $\mathcal{X} = \mathbb{R}\times \mathbb{R} \times [0,1]$,   the outcome space is $\mathcal{Y}  =\mathbb{R}$, and a prediction rule is any mapping from binary lotteries into predictions of the reported certainty equivalent. 
 We use squared-error loss $\ell(y,y') = (y-y')^2$ to evaluate the error of the prediction, but for ease of interpretation we report results in terms of root-mean-squared error, which puts the errors in the same units as the prizes.\footnote{This transformation is possible because  none of the results in this paper change if we redefine $e(\sigma,S) = g\left(\frac{1}{\# S} \sum_{(x,y) \in S} \ell(\sigma(x),y)\right)$ for any function $g$. Root-mean-squared error corresponds to setting $g(x) = \sqrt{x}$.}  
Since different subjects report different certainty equivalents for the same lottery, the best achievable error is generally bounded away from zero.

\subsection{Models and black boxes} \label{sec:DecisionRules} 

We consider two parametric economic models of certainty equivalents and two off-the-shelf black box algorithms.

\smallskip

\emph{Economic models.} First we consider an expected utility agent with a CRRA utility function parameterized by $\eta \geq 0$ (henceforth EU). For $\eta \neq 1$, define
\[v_\eta(z) = \left\{ \begin{array}{cc}
   \frac{z^{1-\eta}-1}{1-\eta} & \mbox{ if } z \geq 0 \\
   - \frac{(-z)^{1-\eta}-1}{1-\eta} & \mbox{ if } z< 0  
    \end{array}\right.\]
and for   $\eta=1$, set $v_\eta(z) = \ln(z)$ for positive prizes and $v_\eta(z) = -\ln(-z)$ for negative prizes. For each $\eta \geq 0$, define the prediction rule $\sigma_{\eta}$ to be
\[\sigma_{\eta}(z_1,z_2,p) = v_\eta^{-1}\big( p\cdot v_\eta(z_1) + (1-p) \cdot v_\eta(z_2)\big).\]
That is, the prediction rule $\sigma_\eta$ maps each lottery into the predicted certainty equivalent for an EU agent with utility function $v_\eta$.

 Next we consider the set of prediction rules $\Sigma_{CPT}$ derived from the parametric form of Cumulative Prospect Theory (CPT)  first proposed by  \citet{goldstein1987expression} and
\citet{lattimore1992influence}. Fixing values for the model's parameters $(\alpha,\beta,\delta,\gamma)$, each lottery $(z_1,z_2,p)$ is assigned a utility
\begin{equation*} 
w(p)v(z_1) + (1-w(p)) v(z_2)
\end{equation*}
where
\begin{equation} \label{eq:utility}
v(z)= \left\{ \begin{array}{cc}
z^{\alpha} & \mbox{ if $z\geq0$} \\
-(-z)^\beta & \mbox{ if $z<0$}
\end{array} \right.
\end{equation}
is a value function for money, and
\begin{equation} \label{eq:w}
w(p)= \frac{\delta p^\gamma}{\delta p^\gamma + (1-p)^\gamma}
\end{equation}
is a probability weighting function. 

For each $\alpha,\beta,\gamma,\delta$, the prediction rule $\sigma_{(\alpha,\beta,\gamma,\delta)}$ is defined as
\[\sigma_{(\alpha,\beta,\gamma,\delta)}(z_1,z_2,p) = v^{-1}\big( w(p)v(z_1) + (1-w(p)) v(z_2)\big).\]
That is, the prediction rule maps each lottery into the predicted certainty equivalent under CPT with parameters $(\alpha,\beta,\gamma,\delta)$. Following the literature, we impose the restriction that the parameters belong to the set $\Theta = \{ (\alpha, \beta,\gamma,\delta) : \alpha,\beta,\gamma \in [0,1], \delta \geq 0\}$. 

 We also evaluate restricted specifications of CPT that have appeared elsewhere in the literature:  CPT with free parameters $\alpha$ and $\beta$ (setting $\delta=\gamma=1$) describes an expected utility decision-maker whose utility function is as given in (\ref{eq:utility}); CPT with free parameters $\alpha$, $\beta$ and $\gamma$ (setting $\delta=1$) is the specification used in \cite{karmarkar1979}; and CPT with free parameters $\delta$ and $\gamma$ (setting $\alpha=\beta=1$) describes a risk-neutral CPT agent whose utility function over money is $u(z)=z$ but who exhibits nonlinear probability weighting. Additionally, we include CPT with the single free parameter $\gamma$ (setting $\alpha=\beta=\delta=1$), which \cite{FudenbergGaoLiang} found to be an especially effective one-parameter specification.

\smallskip

\emph{Black Box Algorithms.} We consider two popular machine learning algorithms. First, we train a \emph{random forest} (RF), which is an ensemble learning method consisting of a collection of decision trees.\footnote{A decision tree recursively partitions the input space, and learns a constant prediction for each partition element. The random forest algorithm collects the output of the individual decision trees, and returns their average as the prediction. Each decision tree is trained with a sample (of equal size to our training data) drawn with replacement from the actual training data. At each decision node, the tree splits the training samples into two groups using a True/False question about the value of some feature, where the split is chosen to greedily minimize mean squared error.} Second, we train a \emph{kernelized ridge regression} model (KR), which  modifies OLS to weight observations at nearby covariate vectors more heavily, and additionally places a penalty term on the size of the coefficients. Specifically, we use the radial basis function  kernel $\kappa(x, \tilde{x}) = e^{-\gamma  \|x-\tilde{x}\|_2^2}$ to assess the similarity between covariate vectors $x$ and $\tilde{x}$. Given training data $\{(x_i,y_i)\}_{i=1}^N$, the estimated weight vector is
$\vec{w} = (\mathbb{K} + \lambda I_N)^{-1} \vec{y}$, where $\mathbb{K}$ is the $N\times N$ matrix whose $(i,j)$-th entry is $\kappa(x_i,x_j)$, $I_N$ is the $N \times N$ identity matrix, and $\vec{y}=(y_1, \dots, y_N)'$ is the vector of observed outcomes in the training data.  
The estimated prediction rule is 
$\sigma(x) = \sum_{i=1}^N w_i \kappa(x,x_i)$.

There are at least two approaches for cross-validating hyper-parameters such as the size of the trees in the random forest algorithm. First, when there are  multiple training domains one can cross-validate across them; we use this in Appendix \ref{app:Altk}. Second, one can  cross-validate  across observations within the training domains.  Since we are interested in cross-domain performance, rather than within-domain performance, it is not guaranteed that this 
will improve performance, and indeed we find that choosing the hyper-parameters via within-domain cross-validation leads to worse transfer performance than using   default values. Thus in our main analysis with a single training domain, we set all hyper-parameters  to default values.\footnote{Specifically, we set $\lambda=1$ and $\gamma = 1/(\#\mbox{covariates})=1/3$ in the kernel regression algorithm. See  \citet{scikit-learn} and Chapter 14 of \citet{Murphy} for further reference. For the random forest model, we set the maximum depth to none, so the tree is extended until outcomes are homogeneous within each leaf.}

\smallskip

\emph{Discussion.} There is no established definition of what constitutes an economic model versus a black box algorithm, but one way of distinguishing between the two approaches is whether the prediction method is tailored to a general application or a general-purpose method of prediction. EU and CPT model the risk preferences of economic agents; we would not expect these models to predict well if we changed our problem to image classification. In contrast, random forest algorithms and kernel regression have been successfully applied across a wide array of prediction problems. In this sense, EU and CPT are economic models, while RF and KR are not. Our approach and results can, however, equally be applied to evaluate prediction methods that are a hybrid of the two approaches. For example, \citet{Plonskyetal} and \citet{Ke} consider black box algorithms whose inputs are  based on prior economic theory. We leave investigation of the transfer performance of such methods to future work.

We note finally that although black box algorithms are traditionally perceived as more flexible than economic models, whether this is in fact the case is something that has to be determined case-by-case. In particular, \citet{FudenbergGaoLiang} shows that although CPT uses only four parameters, it imposes very few restrictions on mappings from binary lotteries to certainty equivalents. 

\subsection{Within-domain performance} \label{sec:WithinDomain}

We first evaluate how these models perform when trained and evaluated on data from the same subject pool.  We  compute the  tenfold cross-validated out-of-sample error  for  each decision rule in each of the 44 domains.\footnote{We split the sample into ten subsets at random, choose nine of the ten subsets for training, and evaluate the estimated model's error on the final subset. The tenfold cross-validated error is the average of the out-of-sample errors on the ten possible choices of test set.}   The two black box methods (random forest and kernel regression) each achieve lower  cross-validated  error than EU and  CPT in 38 of the 44 domains, although the improvement is not large. To obtain a simple summary statistic for the comparison between the economic models and black boxes, we normalize each economic model's  error (in each domain) by the random forest error.  Table \ref{tab:WithinDomain} averages this ratio across domains and shows that on average, the cross-validated errors of the economic models are slightly larger than the random forest error. That is, the CPT error is on average 1.06 times the random forest error, and the EU error is on average 1.21 times the random forest error.\footnote{The numbers in Table \ref{tab:WithinDomain} are very similar if we normalized by the kernel regression error instead.}

\begin{table}[H]
\centering{}%
\begin{tabular}{lcccc}
\hline 
 Model & Normalized Error\\
\hline 
\hline 

EU  & 1.21  \\[2mm]
CPT variants \\
\quad $\gamma$                       & 1.12  \\
\quad $\alpha,\beta$                & 1.22 \\
\quad $\delta,\gamma$                & 1.08  \\
\quad $\alpha,\beta,\gamma$          & 1.07 \\
\quad $\alpha,\beta,\delta,\gamma$   & 1.06 \\
 \hline 
\end{tabular}
\caption{Average ratio of out-of-sample errors relative to random forest.} \label{tab:WithinDomain}
\end{table}

These results suggest that the different prediction methods we consider are comparable for within-domain prediction, with the black boxes performing slightly better. But the results do not distinguish whether the economic models and black boxes achieve similar out-of-sample errors by selecting approximately the same prediction rules, or if the rules they select lead to substantially different predictions out-of-domain. We also cannot determine whether the slightly better within-domain performance of the black box algorithms is achieved by learning generalizable structure that the economic models miss, or if the gains of the black boxes are confined to the domains on which they were trained. We next separate these explanations by evaluating the transfer performance of the models.

\subsection{Transfer error} \label{sec:Transfer}

We use the results in Section \ref{sec:ResultsIID} to construct forecast intervals for the two specifications of transfer error defined in (\ref{eq:TransferError}) and(\ref{def:TransferDeterioration}), which we will subsequently call \emph{raw transfer error}
 and \emph{transfer shortfall} respectively. We also  consider another normalization of the raw transfer error with respect to a proxy for the best achievable error on the target sample. Let $m \in \mathcal{M}$ index a set of models that each prescribe rules $f^m$ for mapping data to prediction rules. Then  \emph{transfer shortfall} 
\begin{equation} \label{def:NormalizedError}
\frac{ e(f_{S_{\bold{T}}},S_{n+1})}{ \min_{m \in \mathcal{M}} e\left(f^m_{S_{n+1}},S_{n+1}\right)}
\end{equation}
reveals how much lower the accuracy of the transferred model $f_{S_\bold{T}}$ is compared to the best in-sample accuracy using a model from $\mathcal{M}$.\footnote{This quantity (subtracted from 1) is similar to the ``completeness" measure introduced in \cite{FKLM}, without the use of a baseline model to set a maximal reasonable error, and  adapted for the transfer setting by training and testing on samples drawn from different domains.}  One advantage of this specification relative to  (\ref{eq:TransferError})  is that the raw error is very sensitive to the  predictability of $y$  given $x$ in the target sample, which may differ across domains but is not directly related to the model's transferability.

 In our meta-data there are $n=44$ domains, and we choose $r=1$ of these to use as the training domain which corresponds to the question, ``If the researcher draws one domain at random, and then tries to generalize to another domain, how well will they do?"   Figure \ref{fig:CI_77}  displays two-sided forecast intervals for  transfer performance, transfer deterioration, and transfer shortfall (where $R$ includes all decision rules shown in the figure). These forecast intervals use $\tau=0.95$, so the upper bound of the forecast interval is the 95th percentile of the pooled transfer errors (across choices of the training and test domains), and the lower bound of the forecast interval is the 5th percentile of the pooled transfer errors. (See Table \ref{tab:CI} in Appendix \ref{app:Exact} for the exact numbers.) Applying Proposition \ref{prop:pooled, percentiles}, these are 86\% forecast intervals. Choosing larger $\tau$ results in wider forecast intervals that have higher coverage levels, and we report some of these alternative forecast intervals in Online Appendix \ref{app:AlternativeCI}, including a 96\% forecast interval.

\begin{figure}[h]\centering
\subfloat[]{\label{fig:CI_77_a}\includegraphics[width=.45\linewidth]{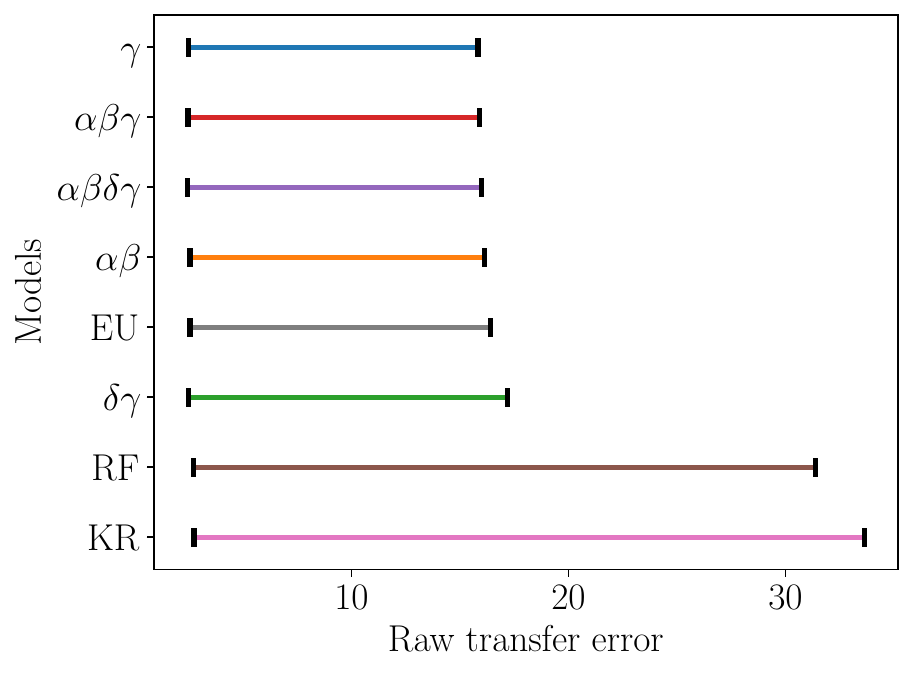}}\\
\subfloat[]{\label{fig:CI_77_b}\includegraphics[width=.45\linewidth]{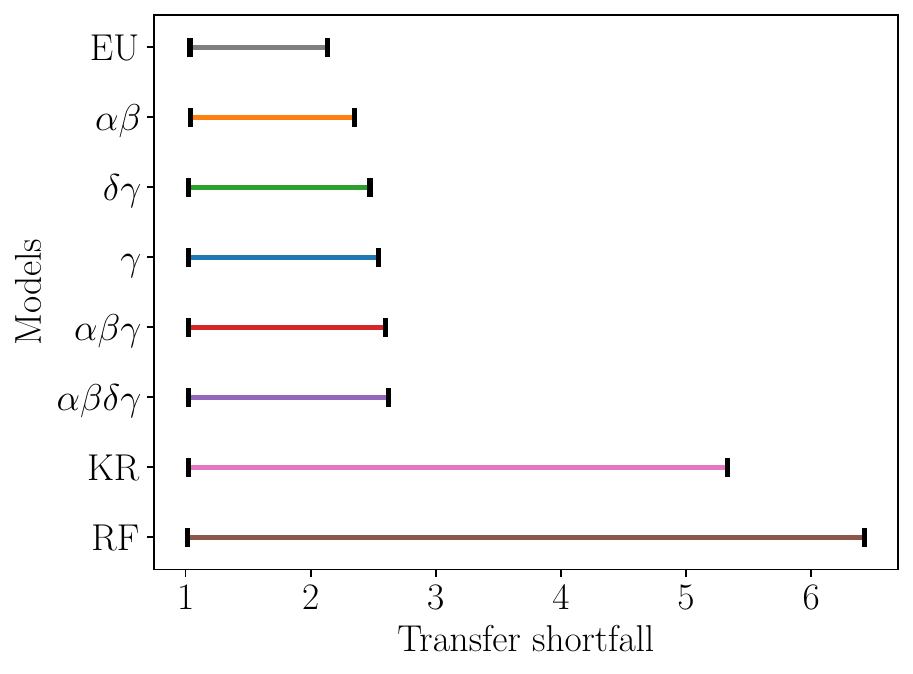}}\hfill
\subfloat[]{\label{fig:CI_77_c}\includegraphics[width=.45\linewidth]{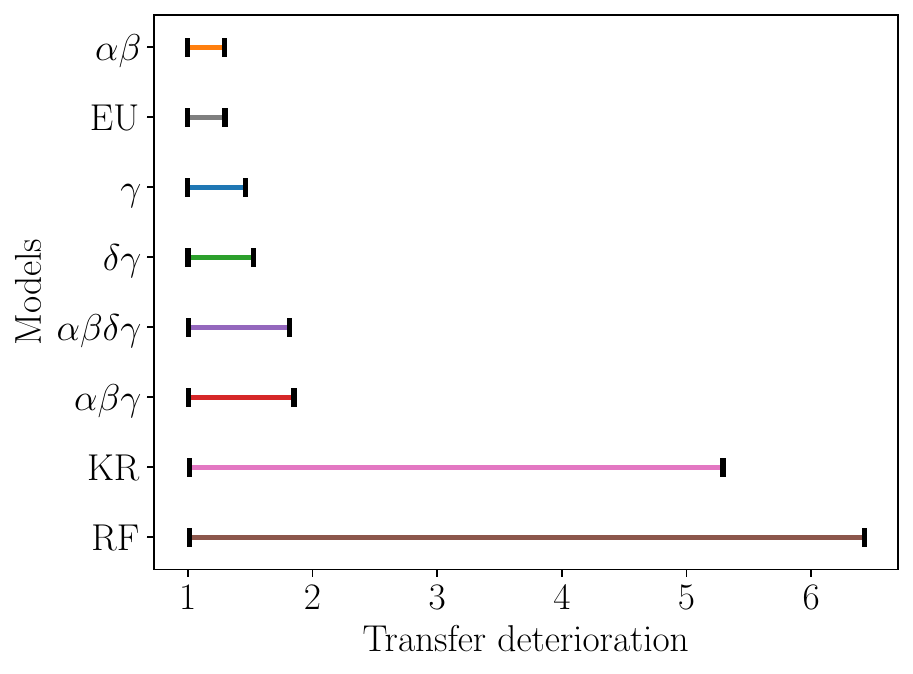}}\par 
\caption{86\% (n=44, $\tau=0.95$) forecast intervals for (a) raw transfer error, (b) transfer shortfall (with $\mathcal{R}$ consisting of the decision rules shown in the figure), and (c)  transfer deterioration.}
\label{fig:CI_77}
\end{figure}

Our main takeaway from Figure \ref{fig:CI_77} is that although the prediction methods we consider are very similar from the perspective of within-domain prediction, they have very different out-of-domain implications. Panel (a) of Figure \ref{fig:CI_77} shows that  the black box forecast intervals for raw transfer error have upper bounds that are roughly twice those of the economic models. Panel (b) shows that  the contrast between the economic models and the black boxes is even larger for transfer shortfall, which removes the common variation across models that emerges from variation in the predictability of the different target samples. Thus, although the economic models and the black box models select prediction rules that are close for the purposes of prediction in the training domain, they sometimes have very different performances in the test domain, and the prediction rules selected by the economic models generalize substantially better. Panel (c) of Figure \ref{fig:CI_77}, which reports transfer deterioration, shows that it is less important to re-estimate the economic models on new target domains than to retrain the black-box algorithms.

All of the forecast intervals overlap for each of the three measures. This is not surprising, as variation in the transfer errors due to the random selection of training and target domains cannot be eliminated even with data from many domains. We expect the black box intervals and the economic model intervals to overlap so long as the economic model errors on ``upper tail" training and target domain pairs exceed the black box errors on ``lower tail" training and target domain pairs. Section \ref{sec:FurtherResults} provides confidence intervals for different population quantities, including quantiles of the transfer error distribution and the expected transfer error, whose width we do expect to vanish as the number of domains grow large. There, we find similar conclusions with regards to the relative performance of the black box algorithms and economic models. 

The appendix provides several robustness checks and complementary analyses. Online Appendix \ref{app:AlternativeCI} plots the $\tau$-th percentile of  pooled transfer errors as $\tau$ varies, demonstrating that forecast intervals constructed using other choices of $\tau$ (besides $\tau=0.95$) would look  similar to those shown in the main text. Online Appendix \ref{sec:Compare} provides 86\% forecast intervals for the ratio of the raw CPT  transfer error to the raw random forest transfer error, and finds that the random forest error is sometimes much higher than the CPT error, but is rarely much lower. Online Appendix \ref{app:Altk} considers an alternative choice for the number of training domains, setting $r=3$ instead of $r=1$. While the results are  similar, the contrast between the economic models and black boxes is not as large, suggesting that the relative performance of the black boxes improves given a larger number of training domains.  Online Appendix \ref{sec:alt-def} provides forecast intervals when each of the 14 papers is treated as a different domain; once again the black box methods transfer worse than the economic models do.

 We next use our theoretical results from Section \ref{sec:RelaxIID} to study the consequences of relaxing the i.i.d.\  assumption in our comparison of CPT$(\alpha,\beta,\delta,\gamma)$ and RF. Since the main differences observed above concerned the upper bounds of our forecast intervals,  we limit attention to $\tau \ge 0.5,$ and compare the methods in terms of worst-case and everywhere upper-dominance with respect to all three measures of  transfer performance. These results are summarized in Table \ref{tab:dominance}.
\begin{table}[H]
\begin{tabular}{c|ccc}
\hline
Type &  raw transfer error &   transfer shortfall & transfer deterioration\\
\hline\hline
Worst-case dominance & $\tau \ge 0.5$ & $\tau \ge 0.5$ & $\tau \ge 0.5$\\
Everywhere dominance & $\tau \ge 0.954$ & $\tau \ge 0.866$ & $\tau \ge 0.647$\\
\hline
\end{tabular}
\caption{Comparison between CPT and RF in terms of worst-case and everywhere upper-dominance. Each cell gives the range of $\tau$ at which CPT dominates RF.}\label{tab:dominance}
\end{table}
Table \ref{tab:dominance} shows that  CPT worst-case-upper-dominates RF at all quantiles $\tau\ge 0.5$ and for all three transfer error measures. Hence, our finding that the upper tail of transfer errors is larger for RF than for CPT is  robust to relaxing the assumption that the training and test domains are drawn from the same distribution, provided that we are comfortable comparing the upper bound for one method to the upper bound for the other.  In Appendix \ref{app:worstcase_dominance}, we provide a more detailed view of worst-case-upper-dominance by plotting $\bar{e}_{\tau}^{\bM}(\Gamma)$ as functions of $\tau$ and $\Gamma$, respectively.

We can also consider the more demanding everywhere-upper-dominance criterion, which asks what happens if we relax our i.i.d.\ sampling assumption in a way which is as favorable to RF (and as unfavorable to CPT) as possible.  We find a substantial degree of robustness even under this highly demanding criterion: CPT everywhere-upper-dominates RF in raw transfer error for all $\tau\ge0.954,$ everywhere-dominates in transfer shortfall for $\tau\ge0.866,$ and everywhere dominates in transfer deterioration for $\tau\ge0.647.$

\subsection{Do  black boxes transfer poorly because they are too flexible?} \label{sec:Overfit}

One tempting  explanation of why the black box algorithms transfer less well is that they may overfit to idiosyncratic details of the training samples that do not generalize across subject pools. For example, suppose some subject pools tend to value lotteries depending on the specific digits they contain.\footnote{For example, \citet{Fortin} find that in neighborhoods with a higher than average percentage of Chinese residents, homes with address numbers ending in “4” are sold at a 2.2\% discount and those ending in “8” are sold at a 2.5\% premium. } This regularity could not be captured by the economic models, because they do not include parameters for individual digits, but could be learned by a random forest algorithm. If so, the random forest would have better within-domain prediction for those subject pools, but worse transfer performance (assuming this regularity does not generalize across subject pools).

While the flexibility of black box algorithms is likely an important determinant of their transfer performance, a second analysis shows that this cannot be a complete explanation of our result. One of the papers we use is based on  samples of certainty equivalents from 30 countries \citep{l2019all}. Of the 30 samples from this paper, 29 samples report certainty equivalents for the same 28 lotteries, and the remaining sample reports certainty equivalents for 24 of those lotteries. We repeat our analysis using these 30 samples as the metadata, and find that the forecast intervals for raw transfer error are indistinguishable across the prediction methods (Panel (a) of Figure \ref{fig:30_domains}). There is some separation between the forecast intervals for the remaining two measures, but in both cases the CPT and random forest forecast intervals are substantially more similar than in the original data. If overfitting were the main explanation of our previous results, we would expect the black box algorithms to overfit here as well.

\begin{figure}[h]\centering
\subfloat[]{\label{fig:30_domains_a}\includegraphics[width=.45\linewidth]{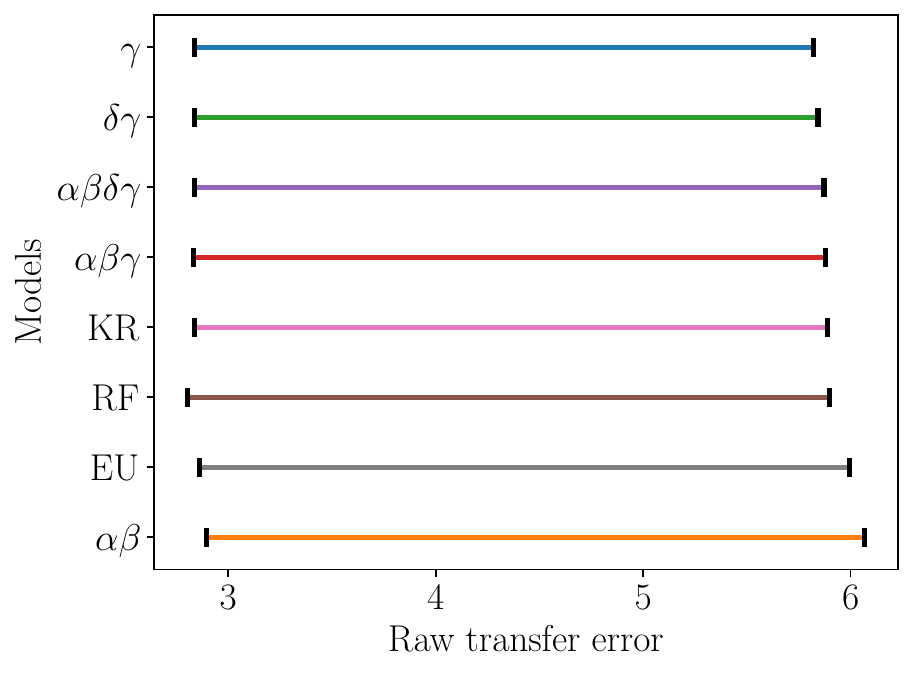}}\\
\subfloat[]{\label{fig:30_domains_b}\includegraphics[width=.45\linewidth]{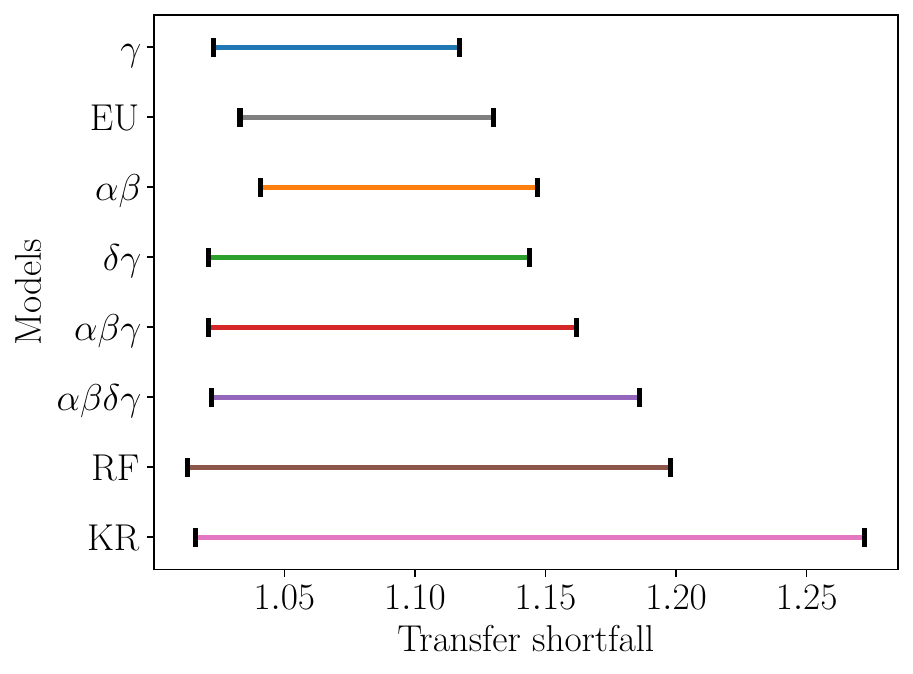}}\hfill
\subfloat[]{\label{fig:30_domains_c}\includegraphics[width=.45\linewidth]{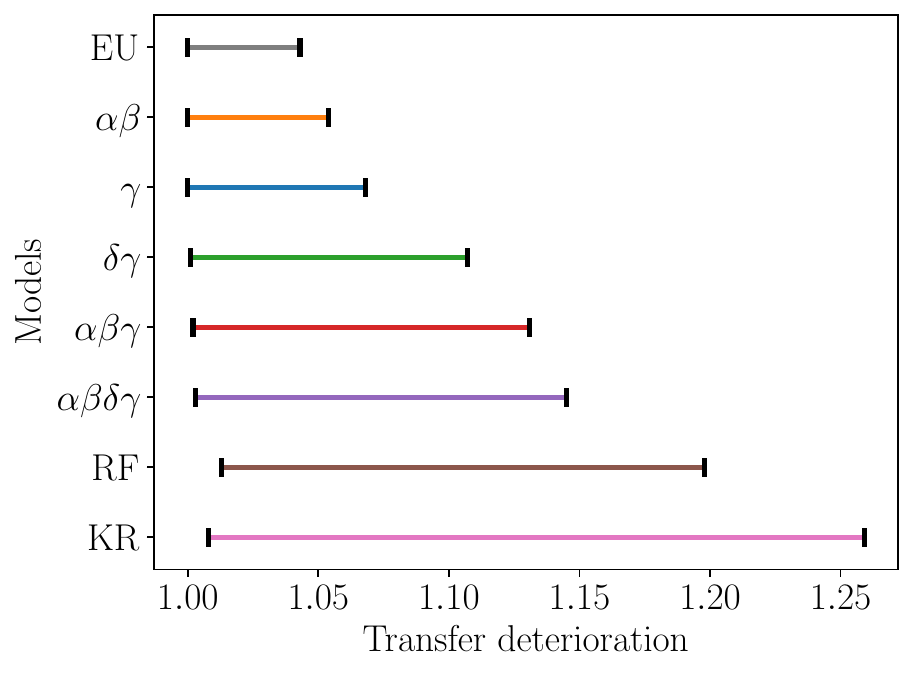}}\par 
\caption{ This figure reports 84\% (n=30, $\tau$=0.95) forecast intervals when we transfer across subject pools in the  \citet{l2019all} data.  }
\label{fig:30_domains}
\end{figure}

In contrast, we find that the economic models again outperform the black box algorithms when we consider a different definition of domains for the \citet{l2019all} data. Specifically, we aggregate all observations for the 24 lotteries that are shared in all 30 samples, and split these observations into 24 samples, where each sample includes all reported certainty equivalents (across subject pools) for a given lottery. For this new definition of domains, our transfer measures evaluate how well a model estimated on data from certain lotteries predict certainty equivalents for other lotteries. Figure \ref{fig:lottery_transfer_r1} reports 83-level confidence intervals, and we find that the economic models transfer substantially better than the black box algorithms. 
 In fact, isolating the difference across domains to be differences across lotteries exaggerates the relative value of economic models even relative to our original Figure \ref{fig:CI_77} (which uses a definition of domains that combines several sources of variation). For consistency, Figure \ref{fig:lottery_transfer_r1} reports confidence intervals for $r=1$ (corresponding to training on one lottery and predicting on another), but we show in Appendix \ref{app:r=3or5} that the qualitative features of this figure hold also for $r=3$ and $r=5$.

\begin{figure}[h]\centering
\subfloat[]{\label{fig:lottery_transfer_r1_a}\includegraphics[width=.45\linewidth]{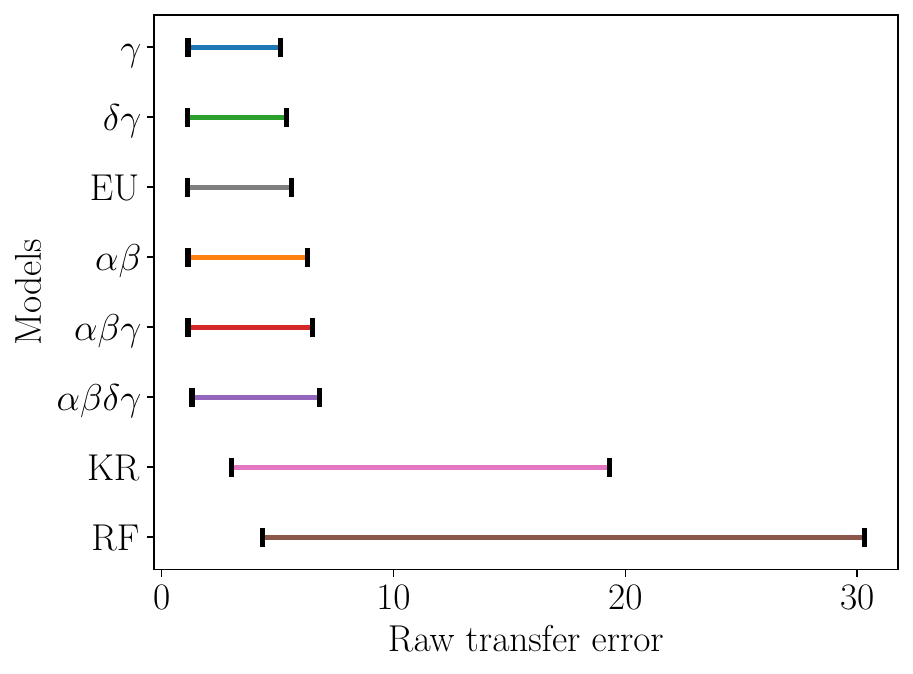}}\\
\subfloat[]{\label{fig:lottery_transfer_r1_b}\includegraphics[width=.45\linewidth]{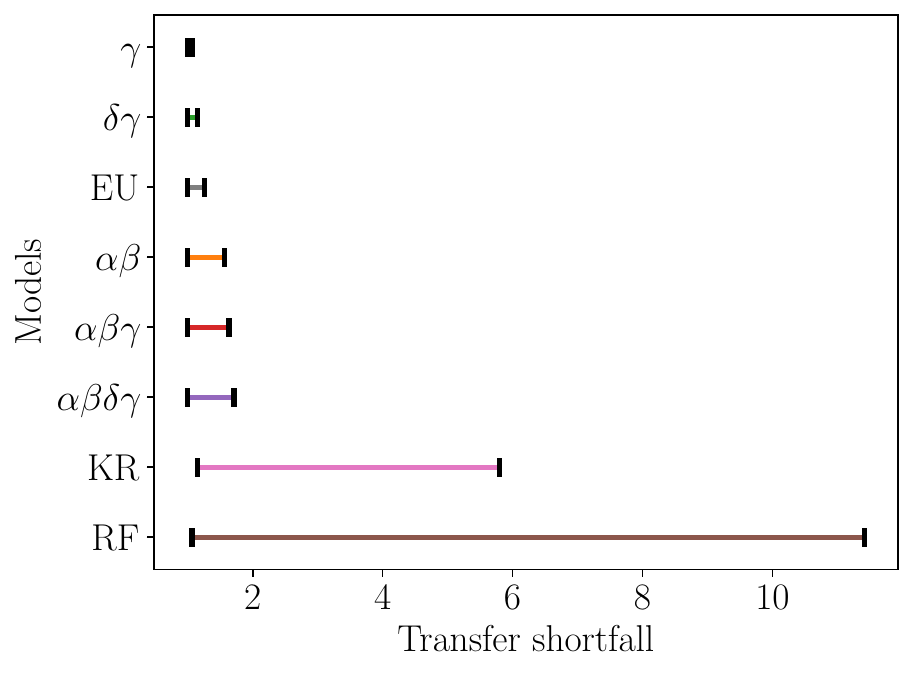}}\hfill
\subfloat[]{\label{fig:lottery_transfer_r1_c}\includegraphics[width=.45\linewidth]{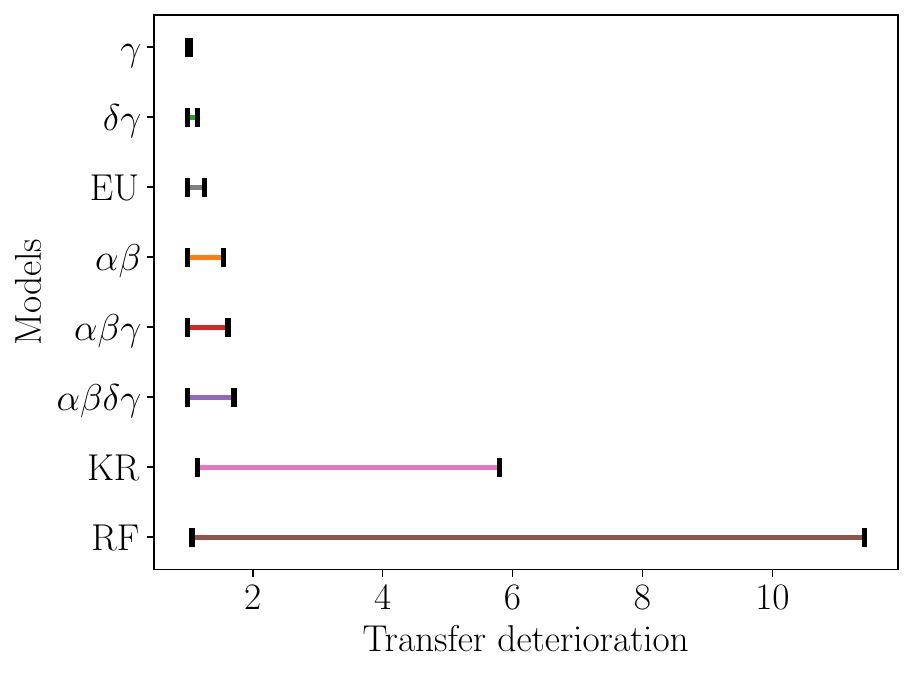}}\par 
\caption{This figure reports 83\% (n=24, $\tau$=0.95) forecast intervals when we transfer across lotteries in the \citet{l2019all} data.  }
\label{fig:lottery_transfer_r1}
\end{figure}

Taken together, our empirical results suggest that the crucial difference between economic models and black box algorithms isn't that one is more flexible and hence more inclined to overfit, but rather that economic models perform better in certain kinds of transfer tasks.\footnote{In fact, the flexibility gap between the black boxes and economic models is not large:  many conditional mean functions (for binary lotteries) can be well approximated by CPT for some choice of parameters values $\alpha,\beta,\delta,\gamma$ \citep{FudenbergGaoLiang}.} The next section discusses more formally one potential explanation for the difference in the relative performance of economic models and black box algorithms in these two transfer tasks. 

\subsection{Two kinds of transfer problems} \label{sec:CovariateShift}

 Our framework allows the distribution $P$ governing the training sample and the distribution $P'$ governing the test sample to differ.  At one extreme, $P$ and $P'$ may share a common marginal distribution on the feature space $\mathcal{X}$, but have very different  conditional distributions $P_{Y\mid X}$ and $P'_{Y \mid X}$ (known as \emph{model shift}). In our application, this would mean that the distribution over lotteries is the same, but the conditional distribution of reported certainty equivalents is different across domains. At another extreme, the conditional distributions $P_{Y \mid X}$ and $P'_{Y \mid X}$ might be the same, but  the marginal distributions over the feature space could differ across domains, e.g., if different kinds of lotteries are used in different domains (known as \emph{covariate shift}).

Our findings in Figure \ref{fig:30_domains} suggest that black boxes do as well as economic models at transfer prediction when the marginal distribution over features $P_X$ is held constant across samples.  Intuitively, when the relevant feature vectors are the same in every sample,  a black box algorithm can perform  well by simply memorizing a prediction for each of these feature vectors. In contrast, when the set of lotteries varies across samples, then  good transfer prediction  necessarily involves  extrapolation, and an algorithm that hasn't identified the right  structure for relating behavior across lotteries will fail to generalize. Since economic models of risk preferences are intended to relate an individual's preferences across lotteries that permits extrapolation of this form, our empirical results suggest that they do so effectively.

For a simple, stylized, example of this contrast, consider three domains with degenerate distributions over observations. In domain 1, the distribution  is degenerate at the lottery $(z_1,z_2,p)=(10,0,1/2)$ and certainty equivalent $y=3$. In domain 2, the distribution is degenerate at the lottery $(z_1,z_2,p)=(10,0,1/2)$ and  certainty equivalent $y=4$. In domain 3, the distribution is degenerate at a new lottery $(z_1,z_2,p)=(20,10,1/10)$ and certainty equivalent $y=11$. Suppose EU and a decision tree are both trained on a sample from domain 1. The CRRA parameter $\eta \approx 0.64$ perfectly fits the observation $(10,0,1/2;3)$, as does the trivial decision tree that predicts $y=3$ for all lotteries. The estimated EU model and decision tree are equivalent for predicting observations in domain 2: both predict $y=3$ and achieve a mean-squared error of $1$. But their errors are very different on domain 3: the EU prediction for the new lottery is  approximately 10.8 with a mean-squared error of approximately 0.05, while the decision tree's prediction is 3 with a mean-squared error of 64.

\subsection{Predicting the relative transfer performance of black boxes and economic models} \label{sec:PredictTE}

The preceding sections suggest that the relative transfer performance of black boxes and economic models is determined primarily by shifts in which lotteries are sampled, rather than shifts in behavior conditional on those lotteries. To further test this conjecture, we examine how well we can predict the ratio of the raw random forest transfer error to the raw CPT transfer error given information about the training and test lotteries but not about the distribution of certainty equivalents in either sample. If the relative performance of these methods depended importantly on behavioral shifts in the two domains---i.e., a change in the distribution of certainty equivalents for the same lotteries---then we would expect prediction of the relative performance based on lottery information alone to be poor. We find instead that lottery information has substantial predictive power for this ratio. 

For each sample $S = \{(z_{1,i},z_{2,i},p_i;y_i)\}_{i=1}^m$, we consider the following features: the mean, standard deviation, max, and min value of  $z_1$ among the  lotteries in $S$;
     the mean, standard deviation, max, and min value of  $z_2$ among the  lotteries in $S$; the mean, standard deviation, max, and min value  of  $p$ among the  lotteries in $S$;  the mean, standard deviation, max, and min value  of  $1-p$ among the  lotteries in $S$; the mean, standard deviation, max, and min of  $pz_1 + (1-p) z_2$ among the  lotteries in $S$; the size of $S$; and an indicator variable for whether $z_1,z_2 \geq 0$ for all lotteries in $S$.

We consider three possible feature sets: (a) \emph{Training Only,} which includes all features derived from the training sample $\bM_\T$; (b) \emph{Test Only,} which includes all features derived from the test sample $S_d$, (c) \emph{Both,} which includes all features derived from the training sample $\bM_\T$ and the test sample $S_d$. We evaluate two  prediction methods: OLS and a random forest algorithm. Table \ref{tab:Predict_Transfer_Error} reports tenfold cross-validated errors for each of these feature sets and prediction methods. As a benchmark, we also consider the best possible constant prediction. 

\begin{table}[H]
\centering{}%
\begin{tabular}{lcccc}
\hline 
 & Train Only & Test Only & Both \\
\hline 
\hline 
Constant & 2.57 & 2.57 & 2.57 \\
OLS & 1.00 & 2.61 & 0.94 \\
RF & 0.98 & 2.52 & 0.76 \\
 \hline 
\end{tabular}
\caption{Cross-Validated MSE} \label{tab:Predict_Transfer_Error}
\end{table}

\begin{figure}[h]
\begin{center}
\includegraphics[scale=0.45]{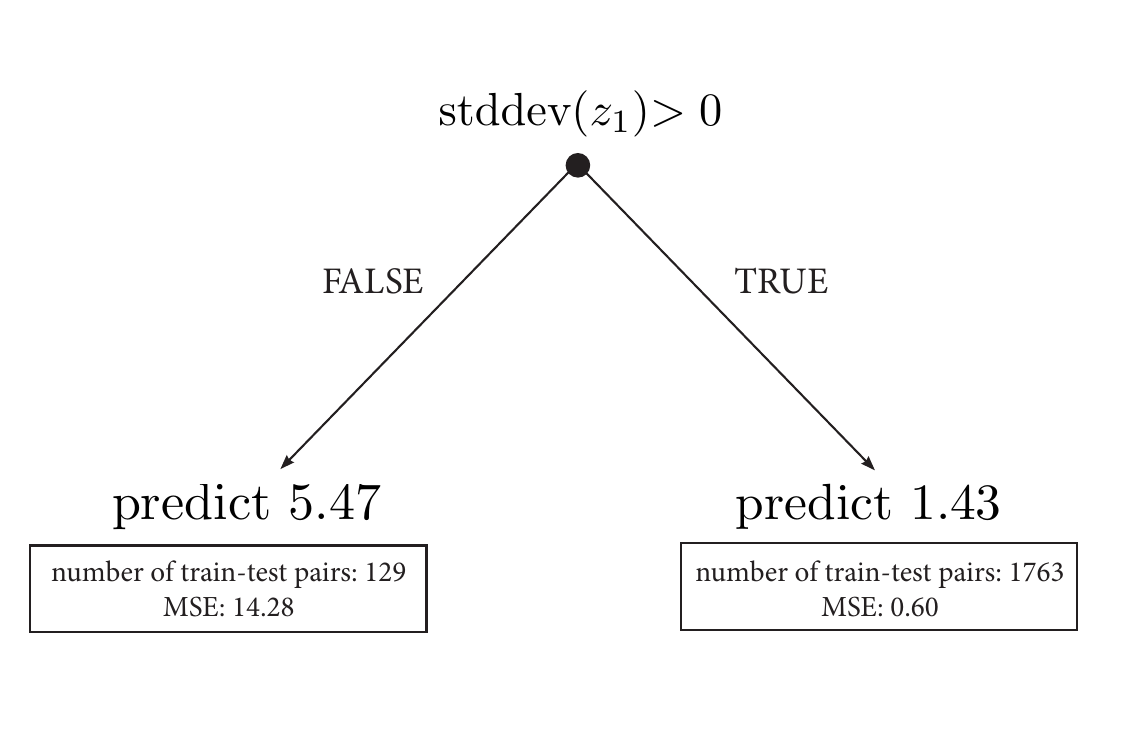}
\end{center}
\caption{Best 1-split decision tree based on training and test features.}
\label{fig:DT-1}
\end{figure}

 The best constant prediction achieves a mean-squared error of 2.57, which can be more than halved using features of the training set alone. Using features of both the training and test sets, the random forest algorithm reduces error to 30\% of the constant model. Crucially, the random forest algorithm is permitted to learn nonlinear combinations of the input features, and thus discover relationships between the training and test lotteries that are relevant to the relative performance of the black box and the economic model. 
 
 The random forest algorithm is too opaque to deliver insight into how it achieves these better predictions, but we can obtain some understanding by examining the best 1-split decision tree, shown in Figure \ref{fig:DT-1} below. This decision tree achieves a cross-validated MSE of 1.75, reducing the error of the constant model by 32\%. It partitions the set of (train,test) domain pairs into two groups depending on whether the standard deviation of $z_1$ (the larger prize) in the training set of lotteries exceeds zero. There are three domains in which the prizes $(z_1,z_2)$ are held constant across all training lotteries (although the probabilities vary). In the 129 transfer prediction tasks where one of these three domains is used for training, the decision tree predicts the ratio of the random forest error to the CPT error to be 5.47. For all other transfer prediction tasks, the decision tree predicts 1.43.

This finding reinforces our intuition that the relative performance of the black boxes and economic models is driven in part by whether the training sample covers the relevant part of the feature space. When the training observations concentrate on an unrepresentative part of the feature space (such as all lotteries that share a common pair of prize outcomes), then the black boxes transfer much more poorly than economic models.

Our  results also clarify a contrast between transfer performance and classical out-of-sample performance. In out-of-sample testing, the marginal distribution on $\mathcal{X}$ is the same for the training and test samples, so the set of training lotteries is likely to be representative of the  set of test lotteries as long as the training sample is sufficiently large. When  test and training samples are governed by  distributions with different marginals on $\mathcal{X}$, the set of training lotteries can be unrepresentative of the set of test lotteries regardless of the number of training observations. Training on observations pooled across many domains alleviates the potential unrepresentativeness of the training data, but the number of domains needed will depend on properties of the distribution: An  environment where each domain puts weight on exactly one lottery that  is itself sampled i.i.d.\ may be difficult for black-box algorithms,\footnote{In this case, the number of domains black boxes need to achieve good transfer performance is likely comparable to the number of observations they  need for good out-of-sample performance, which can be quite large.} while an environment where the marginal distribution is degenerate on the same lottery in all domains may be easier. There is no analog in out-of-sample testing for the role played by variation in the marginal distribution on $\mathcal{X}$ across domains. Moving beyond our specific application, we expect this variation to be an important determinant of the relative transfer performance of black box algorithms and economic models in general.

\section{Conclusion}
Our measures of transfer error quantify how well a model's performance on one domain extrapolates to other domains. We applied these measures to show that the predictions of expected utility theory and cumulative prospect theory outperform those of black box models on out-of-domain tests, even though the black boxes generally have lower out-of-sample prediction errors within a given domain. The relatively worse transfer performance of the black boxes seems to be because the black box algorithms have not identified structure that is commonly shared across domains, and thus cannot effectively extrapolate behavior from one set of features to another. Our finding that the economic models transfer better supports  the intuition that economic models capture regularities that are general across a variety of domains.

One may view our theoretical results, and the statistical assumptions which justify them, from two perspectives.  Taken literally, we consider a researcher who has access to data from a small set of domains, and an analyst who is interested in how well the researcher's procedure extrapolates.  To provide guarantees we restrict the ways domains relate to each other; the simplest and strongest of these is that these is that the domains are exchangeable. We also assume that the analyst has access to data from a larger set of domains than the researcher does. One might wonder why the researcher uses  a small set of domains for training, when the larger meta dataset is available to the analyst. In particular, the researcher might prefer use the larger cross-domain dataset for training, leaving no ``fresh" domains for performance evaluation.  
Although one could potentially weaken our data requirements to cover such cases by restricting the researcher's estimator,  this would rule out the black box prediction methods that are a focus of our analysis.

Alternatively, one can view our results as a tool for evaluating and comparing the performance of different methods for cross-domain prediction.  The starting point for our analysis of a given method is the matrix collecting that method's cross-domain prediction errors for a set of observed domains.  These matrices are rich objects, and it is not obvious how to compare them, but our theoretical analysis points to a natural set of summary statistics: quantiles of the observed error distribution.  Focusing on these quantiles, or equivalently on our forecast intervals, provides a theoretically-motivated way to make comparisons across methods.

%\bibliographystyle{plainnat}

%\singlespacing

%\clearpage
% Appendix
\appendix
%\onehalfspacing

\section{~Proofs}

\subsection{Notation}

Throughout let $\mathcal{N} \equiv \{1,\dots,n\}$. The set $\bbT_{r,n}$ consists of all vectors of length $r$ with  distinct values in $\mathcal{N}$.  For any $(d_1, \ldots, d_{r+1})\in \bbT_{r+1, n+1}$, let 
$f(d_1, \ldots, d_{r+1}) = e_{(d_1,\dots,d_r),d_{r+1}}$ denote the transfer error from training samples $S_{d_1}, \ldots, S_{d_r}$ to test sample $S_{d_{r+1}}$.

\subsection{Proofs of Proposition \ref{prop:pooled, percentiles} and Claim \ref{thm:tight_coverage}}

These are special cases of Theorem \ref{thm:coverage_bias} with $\Gamma = 1$. To avoid repetition, we will only prove Theorem \ref{thm:coverage_bias}. 

\subsection{Proof of Theorem \ref{thm:coverage_bias} and Corollary \ref{cor:coverage_bias}} \label{app:ProofThm1}

We start by proving a simple lemma.

\begin{lemma}\label{lem:mixture}
  Let $H = (1 - \pi)F + \pi G$ be a mixture of two distributions $F$ and $G$. Further let $Z_{H}$ be a draw from $H$ and $Z_{F}$ a draw from $F$. Then, for any $0\le \tau_1 < \tau_2 \le 1$,
  \[(1 - \pi)(\tau_2 - \tau_1)\le \P\lb Z_{H}\in [\underline{Q}_{\tau_1}(F), \overline{Q}_{\tau_2}(F)]\rb \le (1 - \pi)\P\lb Z_{F}\in [\underline{Q}_{\tau_1}(F), \overline{Q}_{\tau_2}(F)] \rb + \pi\]
\end{lemma}
\begin{proof}
  Let $Z_G$ be a draw from $G$, and let $W$ be a binary random variable, independent of $Z_F$ and $Z_G$, with $\E[W] = \pi$. Then $Z_{H}\stackrel{d}{=} (1-W)Z_F + WZ_G$ and
  \begin{align*}
    &\P\lb Z_{H}\in [\underline{Q}_{\tau_1}(F), \overline{Q}_{\tau_2}(F)]\rb\\
    & = (1 - \pi)\P\lb Z_{F}\in [\underline{Q}_{\tau_1}(F), \overline{Q}_{\tau_2}(F)]\rb + \pi \P\lb Z_{G}\in [\underline{Q}_{\tau_1}(F), \overline{Q}_{\tau_2}(F)]\rb\\
    & \in (1 - \pi)\P\lb Z_{F}\in [\underline{Q}_{\tau_1}(F), \overline{Q}_{\tau_2}(F)]\rb + [0, \pi].
  \end{align*}
  By definition of upper and lower quantiles, $\P\lb Z_{F}\in [\underline{Q}_{\tau_1}(F), \overline{Q}_{\tau_2}(F)]\rb\ge \tau_2 - \tau_1.$
  The result then follows.
\end{proof}
\begin{proof}[Proof of Theorem \ref{thm:coverage_bias}]
Throughout the proof we condition on the unordered samples $\{S_1, \ldots, S_{n+1}\}$ and denote by $\{S_{(1)}, \ldots, S_{(n+1)}\}$ any typical realization. Let $\cF$ denote the sigma-field generated by the unordered set $\{S_{(1)}, \ldots, S_{(n+1)}\}$.  Under the assumed data-generating process, 
  \begin{equation*}
    e_{(d_1, \ldots, d_r), n+1}\mid \cF \stackrel{d}{=} f(\bfpiw(d_1), \ldots, \bfpiw(d_r), \bfpiw(n+1)), \quad \forall (d_1, \ldots, d_r)\in \bbT_{r, n}.
  \end{equation*}
where $\bfpiw$ is a random permutation on $\{1, \ldots, n+1\}$ distributed according to 
  \begin{equation}
    \label{eq:dist_piw}
    \P(\bfpiw = \pi) = \frac{1}{n!}\frac{w(S_{\pi(n+1)})}{\sum_{j=1}^{n+1}\omega(S_{j})}.
  \end{equation}
  On the other hand, $\bT  \stackrel{d}{=} (\bfpin(1), \ldots, \bfpin(r))$ where $\bfpin$ denotes a uniform random permutation on $\{1, \ldots, n\}$, so 
  $ e_{\bT , n+1}\mid \cF \stackrel{d}{=} f(\bfpiw\circ \bfpin (1), \ldots, \bfpiw\circ \bfpin(r), \bfpiw(n+1)).$
  By \eqref{eq:dist_piw}, we have
  \begin{equation}
    \label{eq:dist_f_w}
    e_{\bT , n+1}\mid \cF \stackrel{d}{=} f(\bfpiw (1), \ldots, \bfpiw(r), \bfpiw(n+1)).
  \end{equation}
  For any $(d_1, \ldots, d_{r}, k)\in \bbT_{r+1, n+1}$, let
  \[W_k'= \P((\bfpiw(1), \ldots, \bfpiw(r), \bfpiw(n+1)) = (d_1, \ldots, d_{r}, k)) = \frac{(n-r)!}{n!}\frac{w(S_{k})}{\sum_{j=1}^{n+1}\omega(S_{j})}.\]
  Thus,
  \begin{equation}\label{eq:eT_cond_dist}
    e_{\bT , n+1}\mid \cF\sim \sum_{(d_1, \ldots, d_{r}, k)\in \bbT_{r+1, n+1}}W_k'\cdot \delta_{f(d_1, \ldots, d_r, k)}.
    \end{equation}
  %people who make it this far wont need the reminder 
  Because $\P(Z\le \bar{Q}_{\tau}(F))\ge \tau$, 
  \begin{equation}
    \label{eq:oracle_bound}
    \P\lb e_{\bT , n+1}\le \bar{Q}_{\tau}\lb\sum_{(d_1, \ldots, d_{r}, k)\in \bbT_{r+1, n+1}}W_k'\cdot \delta_{f(d_1, \ldots, d_r, k)}\rb\mid \cF\rb\ge \tau.
  \end{equation}
  Let 
  \begin{equation} \label{eq:Omega}
  \Omega_{n+1} = \sum_{(d_1, \ldots, d_{r}, k)\in \bbT_{r+1, n+1}\setminus \bbT_{r+1, n}}W_k'.
  \end{equation}
  \\ 
  By definition, the element $n+1$ must belong to every tuple $(d_1, \ldots, d_{r}, k)\in \bbT_{r+1, n+1}\setminus \bbT_{r+1, n}$. Thus, $W_{n+1}'$ shows up $|\bbT_{r+1, n+1}\setminus \bbT_{r+1, n}|/(r+1)$ times in the sum (\ref{eq:Omega}). By symmetry, each of the other $W_k$'s shows up $|\bbT_{r+1, n+1}\setminus \bbT_{r+1, n}|r/(r+1)n$ times. Since 
  \[|\bbT_{r+1, n+1}\setminus \bbT_{r+1, n}| = |\bbT_{r+1, n+1}| - |\bbT_{r+1, n}| = \frac{(n+1)!}{(n-r)!} - \frac{n!}{(n-r-1)!} = (r+1)\frac{n!}{(n-r)!},\]
  we obtain that
  \begin{align}
    \Omega_{n+1} & = \frac{n!}{(n-r)!}W_{n+1}' + \frac{r(n-1)!}{(n-r)!}\sum_{k=1}^{n}W_k'\nonumber\\
                 & = \frac{r}{n} + \frac{(n-1)!}{(n - r - 1)!}W_{n+1}'= \frac{r}{n} + \frac{n-r}{n}\frac{w(S_{n+1})}{\sum_{j=1}^{n+1}\omega(S_{j})}.\label{eq:sum_omega}
  \end{align}
  where the second to last equality uses $\sum_{k=1}^{n+1}W_{k}' = \frac{(n - r)!}{n!}.$ 
  By \eqref{eq:sum_omega} and \eqref{eq:omegak'}, 
  \[W_k = \frac{W_k'}{1 - \Omega_{n+1}}.\]
  Thus, the distribution in \eqref{eq:eT_cond_dist} can be written as the following mixture distribution
  \begin{align*}
    &\sum_{(d_1, \ldots, d_{r}, k)\in \bbT_{r+1, n+1}}W_k'\cdot \delta_{f(d_1, \ldots, d_r, k)} = (1-\Omega_{n+1}) F_{\bM}^{\omega} + \Omega_{n+1}\cdot G
  \end{align*}
  where $G = \sum_{(d_1, \ldots, d_{r}, k)\in \bbT_{r+1, n+1}\setminus \bbT_{r+1, n}}(W_k' / \Omega_{n+1})\cdot \delta_{f(d_1, \ldots, d_r, k)}$. By Lemma \ref{lem:mixture} with $\tau_1 = 0, \tau_2 = \tau,$ and $F = F_{\bM}^{\omega}$, we have
  \[\P\lb e_{\bT , n+1}\le \bar{Q}_{\tau}(F_{\bM}^\omega)\mid \cF\rb\ge \tau(1 - \Omega_{n+1}).\]
  Moreover, when $f(d_1, \ldots, d_r, k)$ are mutually distinct,
  \[\P\lb Z\le \bar{Q}_{\tau}(F_{\bM}^\omega)\rb\le 1 - \tau + \max_{k}W_{k},\]
  where $Z$ is the draw from the distribution inside $\bar{Q}_{\tau}$ and 
   \[\max_{k} W_k = \frac{(n-r-1)!}{(n-1)!}\frac{\max_{k\le n}\omega(S_k)}{\sum_{j=1}^{n}\omega(S_j)}.\]
  Thus, Lemma \ref{lem:mixture} implies
  \begin{align*}
    &\P\lb e_{\bT , n+1}\le \bar{Q}_{\tau}(F_{\bM}^\omega)\mid \cF\rb\\
    & \le \lb\tau + \frac{(n-r-1)!}{(n-1)!}\frac{\max_{k\le n}\omega(S_k)}{\sum_{j=1}^{n}\omega(S_j)}\rb(1 - \Omega_{n+1}) + \Omega_{n+1}\\
    & = 1 - \lb 1 - \tau - \frac{(n-r-1)!}{(n-1)!}\frac{\max_{k\le n}\omega(S_k)}{\sum_{j=1}^{n}\omega(S_j)}\rb(1 - \Omega_{n+1}).
  \end{align*}
  The result then follows by the iterated law of expectation and \eqref{eq:sum_omega}, which implies 
  \[1 - \Omega_{n+1} = \frac{n-r}{n}\frac{\sum_{j=1}^{n}\omega(S_j)}{\sum_{j=1}^{n+1}\omega(S_j)}.\]
  The two-sided guarantee can be similarly obtained by Lemma \ref{lem:mixture} with $\tau_1 = 1 - \tau$ and $\tau_2 = \tau$.

  To prove Corollary \ref{cor:coverage_bias}, we note that 
  \[\frac{\omega(S_{n+1})}{\sum_{j=1}^{n+1}\omega(S_j)}\le \frac{\Gamma}{n\Gamma^{-1} + \Gamma} = \frac{1}{n\Gamma^{-2} + 1}.\]
  Thus, 
  \[1 - \Omega_{n+1} = \frac{n-r}{n}\lb 1 - \frac{\omega(S_{n+1})}{\sum_{j=1}^{n+1}\omega(S_j)}\rb \ge \frac{n-r}{n}\frac{n\Gamma^{-2}}{n\Gamma^{-2} + 1} = \frac{n - r}{n + \Gamma^2}.\]
\end{proof}

\newpage
\setcounter{page}{1}
\setcounter{section}{15}
\setcounter{subsection}{0}

\clearpage
\noindent \begin{center}
{\large{}Online appendix to the paper}
\par\end{center}{\large \par}

\noindent \begin{center}
{\LARGE{}The Transfer Performance of Economic Models}
\par\end{center}{\LARGE \par}

\medskip{}

\begin{center}
Isaiah Andrews \quad Drew Fudenberg \quad Lihua Lei \quad Annie Liang \quad Chaofeng Wu
\par\end{center}{\large \par}

\noindent \begin{center}
\today
\par\end{center}

\section{~Other Transfer Problems} \label{app:OtherTransfer}

Although we have focused on specifications of transfer error that evaluate how well a model transfers from one domain to another, our results apply for the substantially broader class of random variables given in Definition \ref{def:TransferError}. We discuss below other interesting specifications of $e_{\mathcal{T},d^*}$ and what they might measure.

\subsection{Parameter Transfer} 

\label{sec:ParamaterTransfer} When a model has interpretable parameters, we may  be interested in whether the parameter values estimated on the training data will be a good proxy for the best-fitting parameters in the target sample.

\begin{example}[Effectiveness of a Job Training Program] An economist has estimated the effectiveness of a job training program  using a data set from one location (as in \cite{Hotzetal2005}). How similar would the estimate be if the program were implemented at another location?
\end{example}

\begin{example}[Loss Aversion] An economist observes on a data set of choice over lotteries that ``losses loom larger than gains," specifically that the loss aversion parameter in Prospect Theory has a value larger than 1. If the economist were to elicit choices over a different set of lotteries, would this qualitative conclusion continue to hold?
\end{example}

Consider any model that can be defined as a set $\mathcal{F}_\Theta = \{f_\theta\}_{\theta \in \Theta}$ of prediction rules $f_\theta: \mathcal{X} \rightarrow \mathcal{Y}$, which depend continuously on a parameter $\theta$ in a compact parameter space $\Theta$.  Given any training data $S_{\bold{T}}$, let 
$\hat{\theta}(S_{\bold{T}}) = \arginf_{\theta \in \Theta} \sum_{d \in \bold{T}} \frac{\vert S_d \vert}{\sum_{d \in \bold{T}} \vert S_d \vert} \sum_{d \in \bold{T}} e(f_\theta,S_d)$ be the parameter value that minimizes a weighted sum of the errors across the samples in the training data, and let 
$f_{\hat{\theta}(S_{\bold{T}})}$ denote the corresponding prediction rule.\footnote{If there are ties, break them arbitrarily} To assess parameter variation, first fix a distance metric $d(\theta,\theta')$ (e.g., Euclidean distance) to assess how different two parameter vectors $\theta$ and $\theta'$ are. Then the transfer error
\[e_{\bold{T},n+1} =d\left(\hat{\theta}(S_{\bold{T}}),\hat{\theta}(S_{n+1})\right)\]
 measures how far the estimated parameters on the training data are from the best-fitting parameters on the target sample.

We can also assess how well a qualitative prediction that is based on the estimated parameters will transfer to the target sample (e.g., a prediction that some coefficient is positive). Let $A$ denote any event that can be described as a function of the parameter $\theta$. Then 
\[e_{\bold{T},n+1}=\left\{ \begin{array}{cl}
1 & \mbox{ if } \mathbbm{1}\left(\hat{\theta}(S_{\bold{T}}) \in A\right) = \mathbbm{1}\left(\hat{\theta}(S_{n+1}) \in A\right) \\
0 & \mbox{ otherwise}
\end{array}\right.\]
is a transfer error which tells us whether the prediction about $A$ based on the training samples also holds in the target sample.

\subsection{Other Estimation Procedures} In the examples above, a model is trained on $r$ training samples and used to predict properties of a target sample. Our results apply also for other training procedures. To avoid introducing extensive notation, we describe these procedures informally.

\begin{example}[Transfer Learning]
In \emph{transfer learning} problems in computer science (see e.g., \citet{PanQiang}),   some observations from the target sample are available in addition to the training samples $S_{\bold{T}}$. The model or algorithm is trained on these observations jointly, with some specification of how to weight the target sample observations relative to the other training data. The performance of a model estimated in this way is another  transfer error.
\end{example}

\begin{example}[Transfer of Specific Parameters] While some economic parameters are viewed as constant across domains,  other parameters may be viewed as domain-specific. For example,   spatial models of trade often have structural parameters (e.g., the elasticity of demand substitution between goods produced in different countries) whose values are set using estimates from another paper, and ``fundamentals'' (e.g., productivity in each country), which are re-estimated on each sample \citep[see for example][]{AlfaroUrenaetal}. The performance of a model that is estimated and evaluated in this way is a  transfer error.

\end{example}

\begin{example}[Using Cross-Validation to Tune Parameters] \label{ex:DomainCV}

Our framework can also accommodate training procedures in which cross-validation is used to tune select model parameters. For example, black box algorithms often have a complexity parameter (e.g., the penalization parameter in LASSO or the depth of decision trees in a random forest algorithm). One way of choosing the size of this parameter is based on out-of-sample fit \citep{hastie2009elements, chetverikov2021cross}. In our setting, this  means holding out one of the training samples to use for testing, training the algorithm on the remaining $r-1$ training samples, and evaluating fit on the remaining test sample. The chosen complexity parameter is the one that yields the lowest average error across the $r$ possible choices of the test sample. Fixing this value for the complexity parameter, the algorithm is then fit to the entire training data. The performance of such an algorithm on the target sample is a transfer error.
\end{example}

\begin{example}[Counterfactual Predictions]  One way that economic models are  used is to form predictions for outcomes under policy changes that have yet to be implemented. For instance, \cite{McFadden1974} predicted the demand impacts of the then-new BART rapid transit system in the San Franciso Bay Area, and  \cite{PathakShi2013} predicted demand for schools under changes to the Boston school choice system.  One can generalize our framework to cover the case where each sample $S_d$ is instead a pair of two observations, $S_d = (S_d^0, S_d^1)$. The pre-intervention samples $(S_1^0, \dots, S_{n+1}^0)$ are drawn i.i.d. from one distribution, while the post-intervention samples $(S_1^1,\dots,S_{n+1}^1)$ are drawn i.i.d.\ from another. In this more general setting, a transfer error is any function of the training pairs $\{(S_d^0,S_d^1)\}_{d\in \bold{T}}$, the target pair $(S_{n+1}^0,S_{n+1}^1)$, and potentially an independent noise variable.\footnote{Our theoretical results generalize completely for transfer errors defined in this way;  the main limitation is the difficulty of obtaining sufficiently many pre- and post-intervention pairs. We mention this potential application in the case that such data does eventually become available.}
\end{example}

\section{~Extensions and further results} \label{sec:FurtherResults}

Our main results focus on forecasting realized transfer errors, which is useful when we want to know the range of plausible errors in transferring a given model to a new domain.
We now complement those results with procedures for inference focused on population quantities:    Section \ref{sec:Quantiles} provides confidence intervals for quantiles of the transfer error distribution, and Section \ref{sec:Rademacher} provides a confidence interval for the expected transfer error. Since these quantities can be perfectly recovered given data from an infinite number of domains, we expect the lengths of these intervals to vanish as the number of observed domains grows large, unlike the forecast intervals from Section \ref{sec:ResultsIID}.

\subsection{Preliminary Lemma}

We start by establishing a bound that will be useful in the subsequent construction of confidence intervals. Let
\[U= \frac{(n-k)!}{n!} \sum_{(i_1, \dots, i_k)\in \bbT_{r+1,n}} \phi(Z_{i_1}, \dots, Z_{i_k})\]
be an arbitrary U-statistic of degree $k$  with a bounded (and potentially asymmetric) kernel $\phi$ that takes values in $[0, 1]$.

\begin{definition} For every $n,k \in \mathbb{Z}_+$ and $x,y \in \mathbb{R}$, define
\[B_{n, k}(x; y) \equiv \min\Bigg\{b^1_{n,k}(x; y),b^2_{n,k}(x; y),b^3_{n,k}(x; y) \Bigg\}\]
where
\begin{align*}
    b^1_{n,k}(x;y) & \equiv  \exp\left\{-\lceil n/k \rceil \left( x \wedge y \log\left(\frac{x \wedge y}{y}\right) + (1-x \wedge y) \log \left( \frac{1 - x \wedge y}{1 - y}\right)\right)\right\} \\
    b^2_{n,k}(x;y) & \equiv e\cdot \P\left({\rm Binom}(\lceil n/k \rceil; y) \leq \lceil \lceil n/k \rceil \cdot x\rceil\right) \\
    b^3_{n,k}(x; y) & \equiv  \min_{\lambda > 0}\frac{n\lambda}{k} \lb x - \frac{\lambda }{\lambda + k G(\lambda)} y\rb
\end{align*}

\end{definition}

\begin{lemma} \label{lemm:BoundU} 
If $\phi(Z_1, \ldots, Z_k)\in [0, 1]$ almost surely, then
$P(U \leq x) \leq B_{n,k}(x;\E(U))$ for every $x \in [0,1]$.
\end{lemma}

\subsection{Quantiles of transfer error} \label{sec:Quantiles}

Let $F$ denote the CDF of $e_{\bT,n+1}$, which we assume is continuous. This section builds a confidence interval for the
$\beta$-th quantile of $F$, denoted $q_{\beta}$.

For arbitrary $q \in \mathbb{R}$ and realized metadata $\bM = \{S_1, \dots, S_n\}$, define
\[\varphi(q,\bM) = \frac{(n - r - 1)!}{n!}\sum_{(d_1, \ldots, d_{r+1})\in \bbT_{r+1, n}}\I(e_{(d_1, \dots, d_r), d_{r+1}})\le q)\]
where  $\I(\cdot)$ is the indicator function, recalling that $e_{(d_1, \dots, d_r), d_{r+1}}$ denotes the observed transfer error from samples $(S_{d_1}, \dots, S_{d_r})$ to sample $S_{d_{r+1}}$. This is the fraction of observed transfer errors in the metadata (from $r$ training samples to one test sample) that are less than $q$. Then
$U_\beta \equiv \varphi(q_\beta,\bM)$ 
is a U-statistic where by definition,
$\E[U_\beta]= \beta$. Lemma \ref{lemm:BoundU} then implies
\begin{equation} \label{eq:UseB}
\P(U_\beta \leq x) \leq B_{n,r+1}(x, \beta) \quad \quad \P(U_\beta \geq x)=\P(1 - U_\beta \leq 1-x) \leq B_{n,r+1}(1-x, 1-\beta).
\end{equation}

\begin{definition}
For any quantile $\beta \in (0,1)$ and confidence level $1-\alpha \in (0, 1)$, let
$\hat{u}^{+}_{\beta}(\alpha) = \inf\{u: B_{n, r+1}(u; \beta)\ge \alpha\}$ and $ \hat{u}^{-}_{\beta}(\alpha) = \sup\{u: B_{n, r+1}(1 - u; 1 - \beta)\ge \alpha\}$. Further define $\hat{q}_{\beta}^{\text{L}}(\alpha) \equiv \min\left\{q: \varphi(q,\bM) \ge \hat{u}^{+}_{\beta}(\alpha)\right\}$ and $\hat{q}_{\beta}^{\text{U}}(\alpha) \equiv \max\left\{q: \varphi(q,\bM) \le \hat{u}^{-}_{r}(\alpha)\right\}.$
\end{definition}

Since $B_{n,r+1}(u;\cdot)$ is right-continuous in $u$, it follows from  (\ref{eq:UseB}) that $\P(U_\beta < \hat{u}^+_\beta(\alpha)) \leq \alpha$ and $ \P(U_\beta > \hat{u}^-_\beta(\alpha))\leq \alpha.$ Since $\varphi(q,\bM)$ is monotonically increasing in $q$, the event $\{U_\beta < \hat{u}^+_\beta(\alpha)\}$ is equivalent to
$\{q_\beta < \hat{q}_{\beta}^{\text{L}}(\alpha)\}$, while $\{U_\beta > \hat{u}^-_\beta(\alpha)\}$ is equivalent to
$\{q_\beta > \hat{q}_{\beta}^{\text{U}}(\alpha)\}.$
This yields:

\begin{proposition}\label{thm:quantile_transfer_error_CI}
For any quantile $\beta \in (0,1)$ and confidence level $1-\alpha \in (0, 1)$,
\\ 
$ P(q_{\beta}\le \hat{q}_{\beta}^{ \text{U}}(\alpha))\ge 1 - \alpha$ and
  $\P\lb q_{\beta}\in \left[\hat{q}_{ \beta}^{\text{L}}(\alpha/2), \hat{q}_{ \beta}^{\text{U}}(\alpha/2)\right]\rb\ge 1- \alpha.$
\end{proposition}

Figure \ref{fig:median_bounds} applies Proposition \ref{thm:quantile_transfer_error_CI} to construct two-sided 81\% confidence interval for the median raw transfer error, median  transfer shortfall, and median transfer deterioration. As  in Figure \ref{fig:CI_77}, these confidence intervals are substantially wider for the black box algorithms, and have higher upper bounds.

\begin{figure}[h]\centering
\subfloat[]{\includegraphics[width=.45\linewidth]{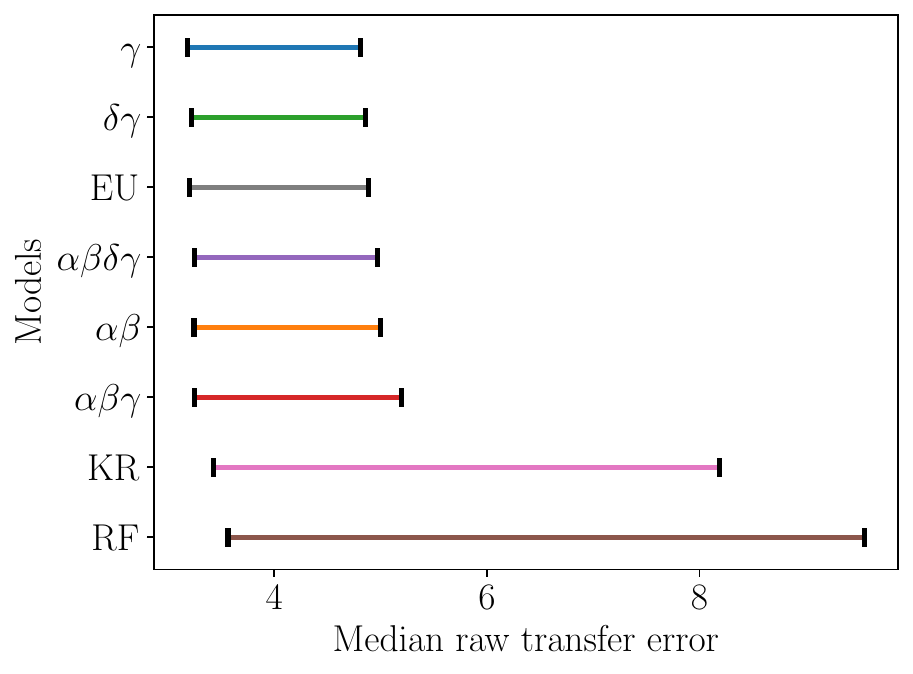}}\\
\subfloat[]{\includegraphics[width=.45\linewidth]{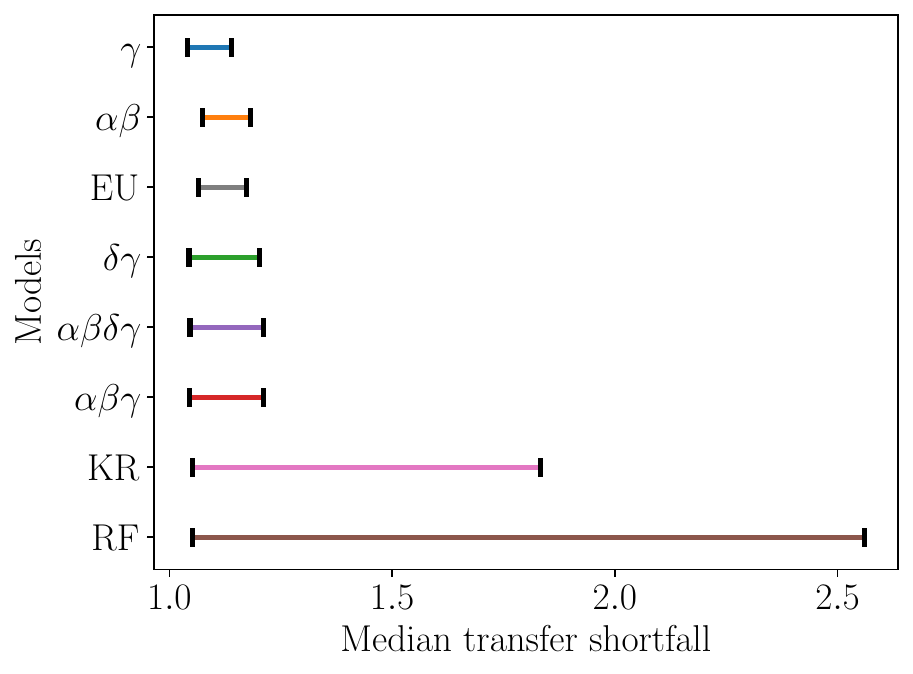}}\hfill
\subfloat[]{\includegraphics[width=.45\linewidth]{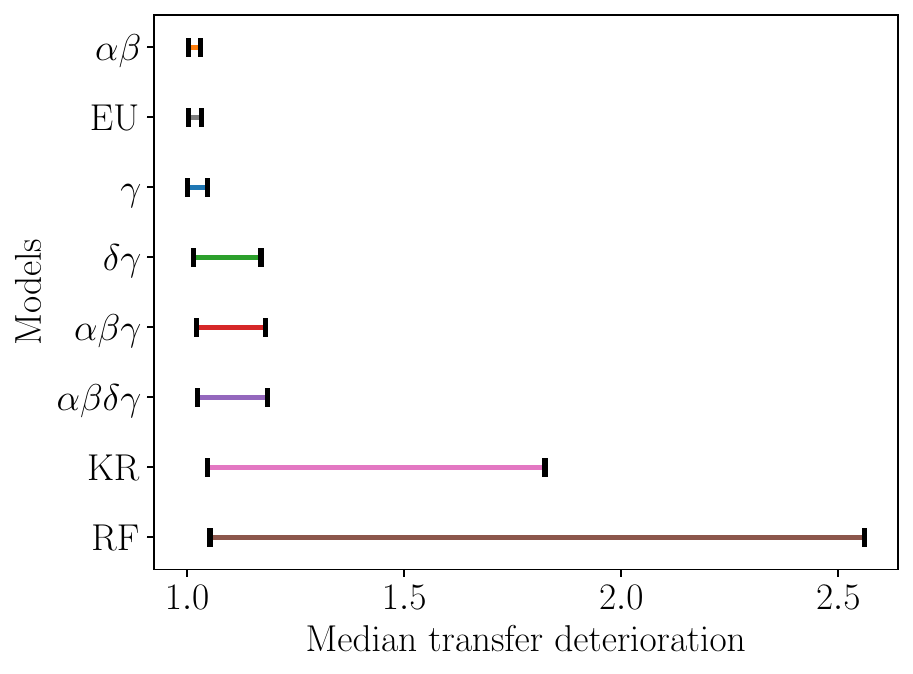}}\par 
\caption{81\% confidence intervals for the median of (a) raw transfer error, (b) transfer shortfall, and (c) transfer deterioration.}
\label{fig:median_bounds}
\end{figure}

\subsection{Expected transfer error} \label{sec:Rademacher}

This section constructs confidence intervals for the expected transfer error, $\mu \equiv \mathbb{E}(e_{\bT, n+1})$, under the assumption that transfer errors are uniformly bounded (in which case it is without loss to set $e_{\bT, n+1}\in [0, 1]$). Define the U-statistic
\begin{equation*}
  U = \frac{(n-r-1)!}{n!}\sum_{(d_1, \ldots, d_{r+1})\in \bbT_{r+1, n}}e_{(d_1, \dots, d_r), d_{r+1}}.
\end{equation*}
Because $\mathbb{E}[U] = \mu$, Lemma \ref{lemm:BoundU} implies that $\P(U \leq x ) \leq B_{n,r+1}(x,\mu)$ and $\P(U \geq x ) \leq B_{n,r+1}(1-x,1-\mu)$ for all $x \in \mathbb{R}$.

\begin{definition}
For any confidence guarantee $1-\alpha \in (0, 1)$, let
$\hat{\mu}^{+}(\alpha) = \sup\{\mu: B_{n, r+1}(U; \mu)\ge \alpha\}$ and $ \hat{\mu}^{-}(\alpha) = \inf\{\mu: B_{n, r+1}(1 - U; 1 - \mu)\ge \alpha\}$.
\end{definition}

It follows from   (\ref{eq:UseB}) that $\P(U < \hat{u}^+(\alpha)) \leq \alpha$ and $ \P(U > \hat{u}^-(\alpha))\leq \alpha$, which implies:

\begin{proposition} \label{prop:ExpectedError}
 If $e_{\T, d}\in [0, 1]$ almost surely, then
  $\P\lb\mu \le \hat{\mu}^{+}(\alpha)\rb\ge 1- \alpha$
  and
  \\
 $\P\lb\mu \in [\hat{\mu}^{-}(\alpha/2), \hat{\mu}^{+}(\alpha/2)]\rb\ge 1- \alpha.$
\end{proposition}

Figure \ref{fig:exp_bounds} applies this result to construct   two-sided 81\% confidence intervals for the transfer errors we considered in Section \ref{sec:Transfer}. Since transfer shortfall and transfer deterioration are not bounded, we report instead confidence intervals for the expectation of their inverses 
$\frac{ \min_{m \in \mathcal{M}} e\left(f^m_{S_{n+1}},S_{n+1}\right)}{ e(f_{S_{\bold{T}}},S_{n+1})}$ and $
\frac{e(f_{S_{n+1}},S_{n+1})}{ e(f_{S_\bold{T}},S_{n+1})};$
lower values for these measures correspond to worse transfer performance. We again  find that the confidence intervals for the black box algorithms are qualitatively worse than those for the economic models.

\begin{figure}[h]\centering
\subfloat[]
{\includegraphics[width=.45\linewidth]{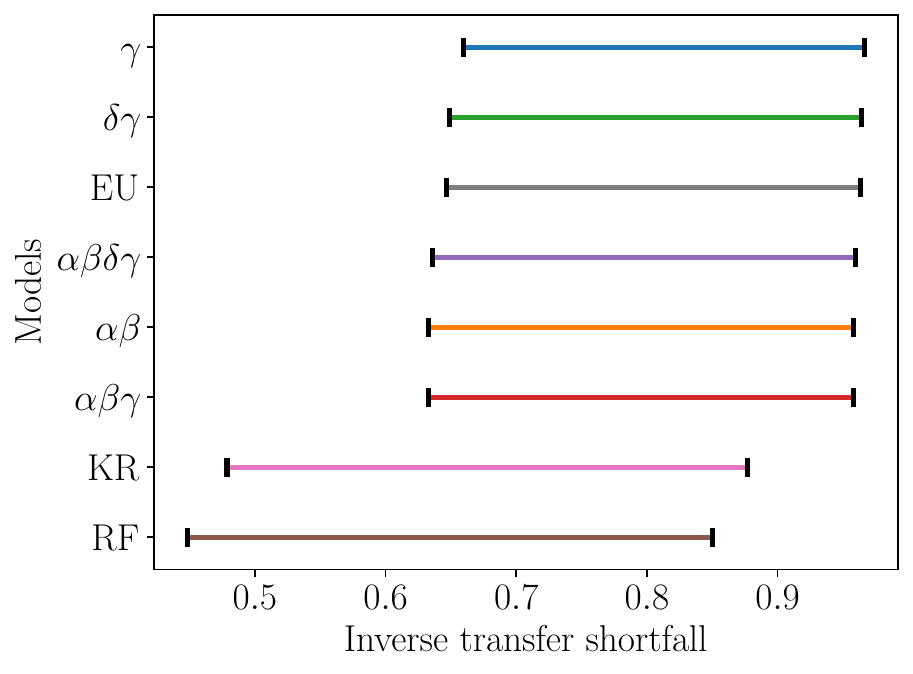}}\hfill
\subfloat[]{\includegraphics[width=.45\linewidth]{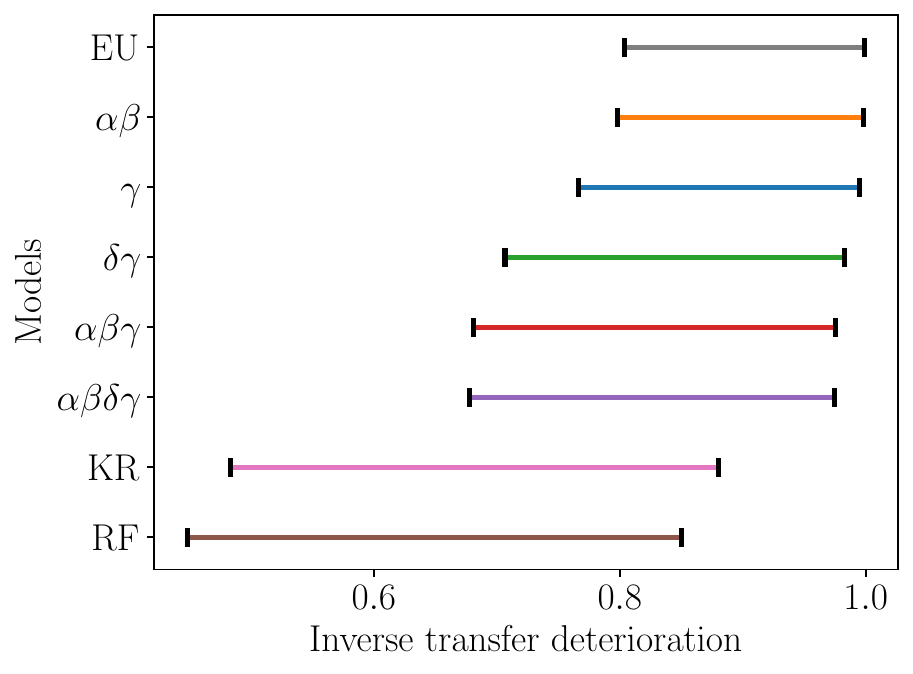}}
\caption{81\% forecast intervals for (a) expected inverse transfer shortfall, (b) expected inverse transfer deterioration.}
\label{fig:exp_bounds}
\end{figure}

\subsection{Proof of Lemma \ref{lemm:BoundU}}\label{app:QuantileExpected}
 \cite{hoeffding1963probability} shows that  $P(U \leq x) \leq b^1_{n,k}(x,\E(U))$, and \cite{bates2021distribution}  shows 
 that $P(U \leq x) \leq b^2_{n,k}(x,\E(U))$. We now show that if $x \in [0,1]$ then $P(U \leq x) \leq b^3_{n,k}(x,\E(U))$. To do this, we use a series of intermediate results to extend a result of \cite{bates2021distribution} on U-statistics of degree $2$ with bounded kernels to U-statistics with bounded kernels for any order $k\ge 2$. 
 
 Let $Z_1, \ldots, Z_n$ be i.i.d.\ random variables and $\phi: \R^{k}\rightarrow [0, 1]$ be a bounded function. Then a U-statistic of degree $k$ is defined as
 \begin{equation}
   \label{eq:Ustat}
   U = \frac{(n-k)!}{n!}\sum_{i_1, \ldots, i_k}\phi(Z_{i_1}, \ldots, Z_{i_k}),
 \end{equation}
 where $\sum_{i_1, \ldots, i_k}$ denotes the sum over all $k$-tuples in $\mathcal{N}$ with mutually distinct elements. The average of $Z_i$ is a special case of \eqref{eq:Ustat} with $k = 1$ and $\phi(z) = z$.
 
 Let $m = \lfloor n / k\rfloor$ and $\bfpin: \mathcal{N}\mapsto \mathcal{N}$ be a uniformly random permutation. For each permutation $\pi$, define 
 \[W_{\pi} = \frac{1}{m}\sum_{j=1}^{m}\phi\lb Z_{\pi((j-1)k + 1)}, \ldots, Z_{\pi(jk)}\rb.\]
 Note that the summands in $W_{\pi}$ are independent given $\pi$. Then $U = \E_{\bfpin}[W_{\bfpin}],$ 
where $\E_{\bfpin}$ denotes the expectation with respect to $\bfpin$ when  conditioning on $Z_1, \ldots, Z_n$. By Jensen's inequality, for any convex function $\psi$,
$\E[\psi(U)] = \E[\psi(\E_{\bfpin}[W_{\bfpin}])]\le \E[\E_{\bfpin}\psi(W_{\bfpin})] = \E_{\bfpin}[\E \psi(W_{\bfpin})].$ Since $W_{\pi}$ has identical distributions for all $\pi$,
\begin{equation}
  \label{eq:hoeffding_representation}
  \E[\psi(U)] \le \E[\psi(W_{\mathbf{id}})]
\end{equation}
where $\mathbf{id}$ is the permutation that maps each element to itself.

Recalling that Hoeffding's inequality is derived from the moment-generating function $\psi(z) = e^{\lambda z}$ \citep{hoeffding1963probability}, and the Bentkus inequality is derived from the piecewise linear function $\psi(z) = (z - t)_{+}$ \citep{bentkus2004hoeffding}, the following tail inequalities for U-statistics are a direct consequence of \eqref{eq:hoeffding_representation}. 
 \begin{proposition}\label{prop:Ustat_hoeffding_bentkus}
Let $U$ be a U-statistic of order $k$ with a bounded kernel $\phi\in [0, 1]$ in the form of \eqref{eq:Ustat} and $m = \lfloor n / k\rfloor$. Then
   \begin{enumerate}     
   \item (Hoeffding inequality for U-statistics, Section 5 of \citealt{hoeffding1963probability})
     \[\P(U \le x)\le \exp\left\{-m h_1\lb x\wedge \E[U]; \E[U]\rb\right\},\]
where
     \[ h_1(y; \mu) = y\log\lb\frac{y}{\mu}\rb + (1 - y)\log\lb\frac{1 - y}{1 - \mu}\rb.\]
   \item (Bentkus inequality for U-statistics, modified from \citealt{bentkus2004hoeffding})
     \[\P\lb U\le x\rb\le e\P\lb\Bin\lb m; \E[U]\rb \le \lceil mx\rceil\rb.\]
   \end{enumerate}
 \end{proposition}

 Other  concentration inequalities can be derived from  the leave-one-out property. Write $U(Z_1, \ldots, Z_n)$ for $U$ and let
 $U_{i} = \inf_{z_i}U(Z_1, \ldots, Z_{i-1}, z_{i}, Z_{i+1}, \ldots, Z_{n}).$ Note that $U_{i}$ is independent of $Z_{i}$. Since $\phi(\cdot)\ge 0$, we have
$0\le U - U_{i} \le \frac{(n-k)!}{n!}\sum_{j=1}^{k}\sum_{i_1, \ldots, i_k, i_j = i}\phi(Z_{i_1}, \ldots, Z_{i_k})$ so 
$\frac{n}{k}(U - U_{i})\le 1$ and
\begin{align*}
  \sum_{i=1}^{n}(U - U_{i})^2
  &\le \frac{((n - k)!)^2}{(n!)^2}\sum_{i=1}^{n}\lb\sum_{j=1}^{k}\sum_{i_1, \ldots, i_k, i_j = i}\phi(Z_{i_1}, \ldots, Z_{i_k})\rb^2\\
  &\stackrel{(i)}{\le} \frac{k(n - k)!}{n\cdot n!}\sum_{j=1}^{k}\sum_{i=1}^{n}\sum_{i_1, \ldots, i_k, i_j = i}\phi(Z_{i_1}, \ldots, Z_{i_k})^2\\
  &\stackrel{(ii)}{\le} \frac{k(n - k)!}{n\cdot n!}\sum_{j=1}^{k}\sum_{i=1}^{n}\sum_{i_1, \ldots, i_k, i_j = i}\phi(Z_{i_1}, \ldots, Z_{i_k})\\  
  & = \frac{k^2}{n}U,
\end{align*}
where (i) applies the Cauchy-Schwarz inequality and (ii) uses the fact that $\phi(\cdot)\le 1$. If we let $W = (n / k)U$ and $W_i = (n / k)U_i$, then $
  W - W_{i}\le 1, \quad \sum_{i=1}^{n}(W - W_{i})^2\le kW.$ This implies that $W$ as a function of $Z_{1}, \ldots, Z_{n}$ satisfies the assumptions for the claim (34) in Theorem 13 of \cite{maurer2006concentration} with constant $a = k$.\footnote{Theorem 13 of \cite{maurer2006concentration} states a weaker result that $\log \E[e^{\lambda(\E[W] - W)}]\le \frac{k\E[W]}{2}\lambda^2$. The stronger version stated here can be found in the second last display in the proof of Theorem 13 of \cite{maurer2006concentration}.} 

\begin{proposition}[Theorem 13, \citealt{maurer2006concentration}]\label{prop:Ustat_maurer}
  Let $G(\lambda) = (e^{\lambda} - \lambda - 1)/\lambda$. Then for any $\lambda > 0$,
  \[\log \E[e^{\lambda(\E[W] - W)}]\le \frac{k\lambda G(\lambda)}{\lambda + k G(\lambda)} \E [W].\]
  This further implies that for any $x \in (0, \E[U])$,
  \[\P\lb U \le x\rb\le \exp\left\{\min_{\lambda > 0}\frac{n\lambda}{k} \lb x - \frac{\lambda }{\lambda + k G(\lambda)} \E [U]\rb\right\}.\]
\end{proposition}

Putting Proposition \ref{prop:Ustat_hoeffding_bentkus} and Proposition \ref{prop:Ustat_maurer} together yields Lemma \ref{lemm:BoundU}.

\section{~Supplementary Material to Section \ref{sec:RelaxIID}}\label{app:Algorithms}

  \subsection{A more general distribution shift model}\label{subsec:general_dist_shift}
In this subsection we discuss a more general distribution shift model that allows $S_1, \ldots, S_{n}$ to have non-identical distributions. Let $p$ denote the joint density of $S_1, \ldots, S_{n+1}$ (with respect to a dominating measure). Let $\bfpip$ be a random permutation on $\{1, \ldots, n+1\}$ such that, for any realization $\{s_1, \ldots, s_{n+1}\}$ of $\{S_1, \ldots, S_{n+1}\}$,
\begin{align*}
  &\P\lb \bfpip(1) = d_1, \ldots, \bfpip(n+1) = d_{n+1} \mid \{S_{(1)}, \ldots, S_{(n+1)}\} = \{s_1, \ldots, s_{n+1}\}\rb\\
  & = \frac{p(s_{d_1}, \ldots, s_{d_{n+1}})}{\sum_{(d'_1, \ldots, d'_{n+1})\in \bbT_{n+1, n+1}} p(s_{d'_1}, \ldots, s_{d'_{n+1}})}.
\end{align*}
Then, 
\begin{align*}
  &(s_{\bfpip(1)}, \ldots, s_{\bfpip(n)})\mid \{S_{(1)}, \ldots, S_{(n+1)}\} = \{s_1, \ldots, s_{n+1}\}\\
  & \stackrel{d}{=}  (S_1, \ldots, S_{n+1}) \mid \{S_{(1)}, \ldots, S_{(n+1)}\} = \{s_1, \ldots, s_{n+1}\}.
\end{align*}
Again let $\cF$ denote the sigma-field generated by the unordered set $\{S_{(1)}, \ldots, S_{(n+1)}\}$.

\begin{definition}
  For any $\Gamma \ge 1$, let $\mathcal{P}(\Gamma; r)$ be the class of distributions on $(S_1, \ldots, S_{n+1})$ with
  \[\frac{(n+1)!}{(n-r)!}\P\lb\bfpip(1) = d_1, \ldots, \bfpip(r) = d_r, \bfpip(n+1) = k\mid \cF\rb \in \left[\Gamma^{-1}, \Gamma\right] \,\, \text{almost surely},\]
  for any $(d_1, \ldots, d_r, k)\in \bbT_{n+1}$.
\end{definition}
Above, $(n-r)!/(n+1)!$ is the probability under a uniform permutation and thus the LHS can be interpreted as the density ratio between $\bfpip$ and a uniform permutation, which measures the deviation from exchangeability.

By \eqref{eq:dist_piw}, when $\nu\in \mathcal{W}(\Gamma)$, the joint density $p$ satisfies 
\begin{align*}
  &\frac{(n+1)!}{(n-r)!}\P(\bfpip(1) = d_1, \ldots, \bfpip(r) = d_r, \bfpip(n+1) = k)\\
  & = \frac{(n+1)\omega(S_k)}{\sum_{j=1}^{n+1}\omega(S_j)} \le \frac{(n+1)\Gamma}{n\Gamma^{-1} + \Gamma} = \frac{(n+1)\Gamma^2}{n + \Gamma^2}.
\end{align*}
Thus,
\[p \in \mathcal{P}\lb \frac{(n+1)\Gamma^2}{n + \Gamma^2}; r\rb\subset \mathcal{P}\lb\Gamma^2; r\rb. \]

Next, we derive forecast intervals akin to Corollary \ref{cor:coverage_bias}.

\begin{theorem}
  Suppose the joint distribution of $(S_1, \ldots, S_{n+1})$ lies in $\mathcal{P}(\Gamma; r)$. Then
  \[\P\lb e_{\bT , n+1}\le \bar{e}_{1 - (1 - \tau)/\Gamma}^{\bM}\rb\ge \tau\lb 1 - \frac{(r+1)\Gamma}{n+1}\rb,\]
  and
  \[\P\lb e_{\bT , n+1}\in [\bar{e}_{(1 - \tau)/\Gamma}^{\bM}, \bar{e}_{1 - (1 - \tau)/\Gamma}^{\bM}]\rb \ge (2\tau - 1)\lb 1 - \frac{(r+1)\Gamma}{n+1}\rb.\]
\end{theorem}
\begin{proof}
  For notational convenience, for any $(d_1, \ldots, d_r, k)\in \bbT_{r+1, n+1}$, let
  \[A_{d_1, \ldots, d_r, k} = \P\lb\bfpip(1) = d_1, \ldots, \bfpip(r) = d_r, \bfpip(n+1) = k\mid \cF\rb.\]
  Again, we condition on the unordered samples $S_1, \ldots, S_{n+1}$ and denote by $S_{(1)},\ldots, S_{(n+1)}$ a typical realization. By similar arguments used to show \eqref{eq:eT_cond_dist}, we have
  \begin{equation}\label{eq:general_dist_shift}
    e_{\bT , n+1}\mid \cF\sim \sum_{(d_1, \ldots, d_{r}, k)\in \bbT_{r+1, n+1}}A_{d_1, \ldots, d_r, k}\cdot \delta_{f(d_1, \ldots, d_r, k)}.
    \end{equation}
  Thus,
  \[\P\lb e_{\bT , n+1}\le \bar{Q}_\tau \lb \sum_{(d_1, \ldots, d_{r}, k)\in \bbT_{r+1, n+1}}A_{d_1, \ldots, d_r, k}\cdot \delta_{f(d_1, \ldots, d_r, k)}\rb \mid \cF\rb\ge \tau.\]
  Let
  \[\Omega_{n+1} = \sum_{(d_1, \ldots, d_{r}, k)\in \bbT_{r+1, n+1}\setminus \bbT_{r+1, n}}A_{d_1, \ldots, d_r, k}.\]
  Then we can rewrite the distribution in \eqref{eq:general_dist_shift} as a mixture distribution
  \[\sum_{(d_1, \ldots, d_{r}, k)\in \bbT_{r+1, n+1}}A_{d_1, \ldots, d_r, k}\cdot \delta_{f(d_1, \ldots, d_r, k)} = (1- \Omega_{n+1})\cdot F + \Omega_{n+1}\cdot G,\]
  where
  \[F = \sum_{(d_1, \ldots, d_{r}, k)\in \bbT_{r+1, n}}\frac{A_{d_1, \ldots, d_r, k}}{1-\Omega_{n+1}}\cdot \delta_{f(d_1, \ldots, d_r, k)}, \quad G = \sum_{(d_1, \ldots, d_{r}, k)\in \bbT_{r+1, n+1}\setminus \bbT_{r+1, n}}\frac{A_{d_1, \ldots, d_r, k}}{\Omega_{n+1}}\cdot \delta_{f(d_1, \ldots, d_r, k)}.\]
  By Lemma \ref{lem:mixture} with $\tau_1 = 0, \tau_2 = \tau$, we have that
  \[\P\lb e_{\bT , n+1}\le \bar{Q}_\tau \lb \sum_{(d_1, \ldots, d_{r}, k)\in \bbT_{r+1, n}}A_{d_1, \ldots, d_r, k}\cdot \delta_{f(d_1, \ldots, d_r, k)}\rb \mid \cF\rb\ge \tau(1 - \Omega_{n+1}).\]
  By definition,
  \[A_{d_1, \ldots, d_r, k}\in \frac{(n-r)!}{(n+1)!}\cdot \left[\Gamma^{-1}, \Gamma\right].\]
  The largest possible value for $\bar{Q}_{\tau}\lb\sum_{(d_1, \ldots, d_{r}, k)\in \bbT_{r+1, n}}A_{d_1, \ldots, d_r, k}\cdot \delta_{f(d_1, \ldots, d_r, k)}\rb$ is achieved when $A_{d_1, \ldots, d_r, k} = \Gamma (n-r)!/(n+1)!$ for the largest values of $f(d_1, \ldots, d_r, k)$. Thus,
  \[\bar{Q}_{\tau}\lb\sum_{(d_1, \ldots, d_{r}, k)\in \bbT_{r+1, n}}A_{d_1, \ldots, d_r, k}\cdot \delta_{f(d_1, \ldots, d_r, k)}\rb \le \bar{Q}_{\tau'}\lb \sum_{(d_1, \ldots, d_{r}, k)\in \bbT_{r+1, n}}\delta_{f(d_1, \ldots, d_r, k)}\rb = \overline{e}_{\tau'}^{\bM},\]
  where $\Gamma(1 - \tau') = 1 - \tau$. Clearly, $\tau' = 1 - (1 - \tau)/\Gamma$. Thus,
  \[\P\lb e_{\bT , n+1}\le \overline{e}_{\tau'}^{\bM} \mid \cF\rb\ge \tau(1 - \Omega_{n+1}).\]
Moreover,
\[\Omega_{n+1} \le \Gamma\frac{(n-r)!}{(n+1)!}\cdot \lb|\bbT_{r+1,n+1}| - |\bbT_{r+1,n}|\rb = \frac{(r+1)\Gamma}{n+1}.\]
Thus, the result for the one-sided interval is proved. The result for the two-sided interval can be proved similarly by Lemma \ref{lem:mixture} with $\tau_1 = 1 - \tau, \tau_2 = \tau$ and by noting that the smallest possible value for $\underline{Q}_{\tau}\lb\sum_{(d_1, \ldots, d_{r}, k)\in \bbT_{r+1, n}}A_{d_1, \ldots, d_r, k}\cdot \delta_{f(d_1, \ldots, d_r, k)}\rb$ is achieved when $A_{d_1, \ldots, d_r, k} = \Gamma (n-r)!/(n+1)!$ for the smallest values of $f(d_1, \ldots, d_r, k)$, and hence
\[\underline{Q}_{\tau}\lb\sum_{(d_1, \ldots, d_{r}, k)\in \bbT_{r+1, n}}A_{d_1, \ldots, d_r, k}\cdot \delta_{f(d_1, \ldots, d_r, k)}\rb \ge \underline{Q}_{1-\tau'}\lb \sum_{(d_1, \ldots, d_{r}, k)\in \bbT_{r+1, n}}\delta_{f(d_1, \ldots, d_r, k)}\rb = \underline{e}_{\tau'}^{\bM}.\]
\end{proof}

\subsection{Algorithm for evaluating worst-case-upper-dominance}

We provide an algorithm that computes $\bar{e}_{\tau}(\Gamma)$ with a single $\tau$ in $O(rn^{r+1} \log n)$ time and computes $\bar{e}_{\tau}(\Gamma)$ for all $\tau\in (0, 1)$ in $O(rn^{r+1} \log n + n^{r+2})$ time. 
 First,  sort the elements in $\{f(d_1, \ldots, d_{r+1}): (d_1, \ldots, d_{r+1})\in \bbT_{r+1, n}\}$ as
  \[f_{(1)}\le f_{(2)}\le \ldots \le f_{(|\bbT_{r+1, n}|)},\]
  where
  \[f_{(j)} = f(d^{(j)}), \quad d^{(j)} = (d_1^{(j)}, \ldots, d_{r+1}^{(j)})\in \bbT_{r+1, n}.\]
  Let $\psi^{(j)}\in \{0, 1\}^{n}$ with
  \[\psi_{i}^{(j)} = I\lb d_{r+1}^{(j)} = i\rb.\]
  Further define the cumulative sum of $\psi^{(j)}$ as
  \[\Psi^{(j)} = \sum_{\ell = 1}^{j}\psi^{(\ell)}.\]
  Let $\w = (\omega(S_1), \ldots, \omega(S_n))^{T}$ and $\one_{n} = (1, 1, \ldots, 1)^{T}$. By \eqref{eq:omegak'}, for each $j$,
  \[f_{(j)} \ge \bar{e}_{\tau}^{\bM, \omega}\Longleftrightarrow \frac{(n-r-1)!}{(n-1)!}\frac{\w^{T}\Psi^{(j)}}{\w^{T}\one_{n}}\ge \tau.\]
  Therefore,
  \[\bar{e}_{\tau}^{\bM, \omega} = f_{J^{\omega}}, \quad \text{where }J_{\tau}^{\omega} =  \min\left\{j: \frac{\w^{T}\Psi^{(j)}}{\w^{T}\one_{n}}\ge \tau \frac{(n - 1)!}{(n - r - 1)!}\right\}.\]
  By definition, the set of $\w$ generated by all $\omega\in \cW(\Gamma)$ is 
  $[\Gamma^{-1}, \Gamma]^{n}$. Thus,
  \begin{equation}
    \label{eq:etau_Gamma_def}
    \bar{e}_{\tau}^{\bM}(\Gamma) = f_{J_{\tau}(\Gamma)}, \quad \text{where } J_{\tau}(\Gamma) = \min\left\{j: \min_{\w\in [\Gamma^{-1}, \Gamma]^{n}}\frac{\w^{T}\Psi^{(j)}}{\w^{T}\one_{n}}\ge \tau \frac{(n - 1)!}{(n - r - 1)!}\right\}.
  \end{equation}
  Via some algebra, we can further simplify the expression of $\bar{e}_{\tau}^{\bM}(\Gamma)$.
  \begin{theorem}\label{thm:ebar_Gamma_fixtau}
    Let $\bar{\Psi}_k^{(j)}$ be the average of the $k$-smallest coordinates of $\Psi^{(j)}$ and 
    \[Q_j(\Gamma) = \frac{j}{n} + \min_{k\in \mathcal{N}}\frac{\bar{\Psi}_{k}^{(j)} - \frac{j}{n}}{1 + \frac{n}{k (\Gamma^2 - 1)}}.\]
    Then $Q_{j}(\Gamma)$ is strictly increasing in both $j$ and $\Gamma$. Moreover, $\bar{e}_{\tau}^{\bM}(\Gamma) = f_{(J_{\tau}(\Gamma))}$, where
    \[J_{\tau}(\Gamma)= \min\left\{j\ge \tau \frac{n!}{(n - r - 1)!}: Q_j(\Gamma)\ge \tau\frac{(n-1)!}{(n-r-1)!}\right\}.\]
  \end{theorem}
  \begin{proof}
    First, we prove that
    \begin{equation}
      \label{eq:boundary_condition}
      \min_{\w\in [\Gamma^{-1}, \Gamma]^{n}}\frac{\w^{T}\Psi^{(j)}}{\w^{T}\one_{n}} = \min_{\w\in \{\Gamma^{-1}, \Gamma\}^{n}}\frac{\w^{T}\Psi^{(j)}}{\w^{T}\one_{n}}.
    \end{equation}
    Let $g_j(\w) = \w^T \Psi^{(j)} / \w^T \one_{n}$. Then $g_{j}$ is continuous and bounded on the closed set $[\Gamma^{-1}, \Gamma]^{n}$ and thus the minimum can be achieved. Let
    \[\w^{(j)}(\Gamma) = \argmin_{\w:g_j(\w) = \min_{\w\in [\Gamma^{-1}, \Gamma]^{n}} g_j(\w)}\sum_{i=1}^{n}\min\left\{|\w_i - \Gamma|, |\w_i - \Gamma^{-1}|\right\}.\]
    Suppose there exists $i\in \mathcal{N}$ such that $\w_{i}^{(j)}(\Gamma) \in (\Gamma^{-1}, \Gamma)$. Then
    \[g_j(\w_i, \w_{-i}) = \frac{\Psi_i^{(j)} \w_i + \Psi_{-i}^{(j)T} \w_{-i}}{\w_i + \one_{n-1}^{T}\w_{-i}} = \Psi_i^{(j)} + \frac{\Psi_{-i}^{(j)T} \w_{-i} - \Psi_i^{(j)}\cdot \one_{n-1}^{T}\w_{-i}}{\w_i + \one_{n-1}^{T}\w_{-i}},\]
    where $\Psi_{-i}^{(j)}$ and $\w_{-i}$ are the leave-$i$-th-entry subvectors of $\Psi^{(j)}$ and $\w$. Clearly, $g_j$ is a monotone function of $\w_i$ for any given $\w_{-i}$. Since $\w^{(j)}(\Gamma)$ is a minimizer and $\w_i^{(j)}(\Gamma) \in (\Gamma^{-1}, \Gamma)$, we must have $\Psi_{-i}^{(j)T} \w_{-i} - \Psi_i^{(j)}\cdot \one_{n-1}^{T}\w_{-i}= 0$. Define $\td{\w}^{(j)}(\Gamma)$ with
    \[\td{\w}_i^{(j)}(\Gamma) = \Gamma, \quad \td{\w}_{-i}^{(j)}(\Gamma) = \w_{-i}^{(j)}(\Gamma).\]
    Then
    \[g_j(\td{\w}^{(j)}(\Gamma)) = g_j(\w^{(j)}(\Gamma)) = \min_{\w\in [\Gamma^{-1}, \Gamma]^{n}}g_j(\w),\]
    while 
    \[\sum_{i=1}^{n}\min\left\{|\td{\w}_i^{(j)}(\Gamma) - \Gamma|, |\td{\w}_i^{(j)}(\Gamma) - \Gamma^{-1}|\right\} < \sum_{i=1}^{n}\min\left\{|\w_i^{(j)}(\Gamma) - \Gamma|, |\w_i^{(j)}(\Gamma) - \Gamma^{-1}|\right\}.\]
    This contradicts the definition of $\w^{(j)}(\Gamma)$, so $\w^{(j)}(\Gamma) \in \{\Gamma^{-1}, \Gamma\}^{n},$ which completes the proof of \eqref{eq:boundary_condition}.

    For any $\w \in \{\Gamma^{-1}, \Gamma\}^{n}$ with $|\{i: \w_i = \Gamma\}| = k$, the Fr\'{e}chet-Hoeffding inequality implies that $\Gamma$'s are allocated to the $k$ smallest entries of $\Psi^{(j)}$. Thus,
    \[\min_{\w\in \{\Gamma^{-1}, \Gamma\}^{n}}\frac{\w^{T}\Psi^{(j)}}{\w^{T}\one_{n}} = \min_{k\in \mathcal{N}\cup \{0\}}\frac{\Gamma k \bar{\Psi}_k^{(j)} + \Gamma^{-1}(\one_{n}^{T}\Psi_i^{(j)} - k \bar{\Psi}_k^{(j)})}{\Gamma k + \Gamma^{-1}(n - k)}.\]
    By definition, $\one_{n}^{T}\Psi_i^{(j)} = j$. Then for each $k$, the above expression can be simplified as
    \begin{align*}
      &\frac{\Gamma k \bar{\Psi}_k^{(j)} + \Gamma^{-1}(\one_{n}^{T}\Psi_i^{(j)} - k \bar{\Psi}_k^{(j)})}{\Gamma k + \Gamma^{-1}(n - k)} = \frac{\Gamma k \bar{\Psi}_k^{(j)} + \Gamma^{-1}(j - k \bar{\Psi}_k^{(j)})}{\Gamma k + \Gamma^{-1}(n - k)}\\
      & = \frac{(\Gamma - \Gamma^{-1}) k \bar{\Psi}_k^{(j)} + \Gamma^{-1}j}{(\Gamma - \Gamma^{-1}) k + \Gamma^{-1}n} = \frac{j}{n} + \frac{(\Gamma - \Gamma^{-1}) k \lb\bar{\Psi}_k^{(j)} - \frac{j}{n}\rb}{(\Gamma - \Gamma^{-1}) k + \Gamma^{-1}n}\\
      & = \frac{j}{n} + \frac{\bar{\Psi}_k^{(j)} - \frac{j}{n}}{1 + \frac{n}{k(\Gamma^{2} - 1)}}.
    \end{align*}
    The above expression is $j / n$ for both $k = n$ and $k = 0$, so we can remove $0$ from the minimum, and thus 
    \[\min_{\w\in [\Gamma^{-1}, \Gamma]^{n}}\frac{\w^{T}\Psi^{(j)}}{\w^{T}\one_{n}} = Q_j(\Gamma).\]
    By \eqref{eq:etau_Gamma_def},
    \[\bar{e}_{\tau}^{\bM}(\Gamma) = \min\left\{j: Q_{j}(\Gamma)\ge \tau \frac{(n-1)!}{(n-r-1)!}\right\}.\]
    Finally, we can restrict to $j\ge \tau n! / (n - r - 1)!$ because
    $Q_{j}(\Gamma) \le \frac{j}{n}$ by taking $k = n$.
  \end{proof}

  Since $Q_{j}(\Gamma)$ is increasing in $j$, $J_{\tau}(\Gamma)$ can be found via binary search with iteration complexity $O(\log n^{r+1}) = O(r\log n)$. Each iteration costs at most $O(n)$ operations to sort the entries of $\Psi^{(j)}$ based on the ordered version of $\Psi^{(j-1)}$, since there is only entry updated, and $O(n)$ additional operations to compute $Q_{j}(\Gamma)$. Thus, the overall computational overhead after obtaining $(f_{(1)}, \ldots, f_{(|\bbT_{r+1, n}|)})$ is just $O(rn \log n)$, which is much smaller than the cost of sorting $f$-values $O(n^{r+1}\log n^{r+1}) = O(rn^{r+1}\log n)$.

  In some cases, we want to compute $\bar{e}_{\tau}^{\bM}(\Gamma)$ for all $\tau \in [0, 1]$ at once. The following result  links $\bar{e}_{\tau}^{\bM}(\Gamma)$ to an induced distribution on the $f$'s.
  \begin{corollary}\label{cor:ebar_Gamma_alltau}
    For any $\Gamma\ge 1$, let $\mu_{\Gamma}$ be a weighted measure with
    \[\mu_{\Gamma} = \sum_{j=1}^{|\bbT_{r+1, n}|}\frac{(n - r - 1)!}{(n - 1)!}(Q_{j}(\Gamma) - Q_{j-1}(\Gamma))\cdot \delta_{f_{(j)}},\]
    where $Q_{0}(\Gamma) = 0$. Then $\bar{e}_{\tau}^{\bM}(\Gamma)$ is the $\tau$-th quantile of $\mu_{\Gamma}$.
  \end{corollary}
  Since the ordering takes $O(rn^{r+1}\log n)$ time and computing each $Q_{j}(\Gamma)$ takes $O(n)$ time, the total computational cost to compute $\bar{e}_{\tau}^{\bM}(\Gamma)$ for all $\tau\in [0, 1]$ is $O(rn^{r+1} \log n + n^{r+2})$.
  
  \subsection{Algorithm for evaluating everywhere dominance}

  Let $f_{(j), 1}$ and $f_{(j), 2}$ be the $j$-th largest transfer errors for method 1 and 2, respectively. Similarly, the count vectors for two methods are denoted by $\Psi^{(j), 1}$ and $\Psi^{(j), 2}$. Then method 1 does NOT everywhere-upper-dominate method 2 at the $\tau$-th quantile if and only if there exists $j_1, j_2\in \{1, \ldots, |\bbT_{r+1, n}|\}$ and $W\in [0, \infty)^{n}$ such that
  \begin{equation}
    \label{eq:crude_everywhere_dominance}
    f_{(j_1), 1} > f_{(j_2), 2}, \quad \frac{(n - r - 1)!}{(n - 1)!}\frac{\w^T \Psi^{(j_1 - 1), 1}}{\w^{T}\one_{n}} < \tau \le \frac{(n - r - 1)!}{(n - 1)!}\frac{\w^T \Psi^{(j_2), 2}}{\w^{T}\one_{n}}.
  \end{equation}
  Above $\Psi^{(0), 1} = (0, 0, \ldots, 0)^{T}$.

  To avoid pairwise comparisons, which incur $O(n^{2(r+1)})$ computation, we can check \eqref{eq:crude_everywhere_dominance} by only focusing on $j_1 = m(j), j_2 = j$ where
  \[m(j) = \min\{j': f_{(j'), 1} > f_{(j), 2}\}.\]
  It is easy to see that \eqref{eq:crude_everywhere_dominance} holds for some pair $(j_1, j_2)\in \{1, \ldots, |\bbT_{r+1, n}|\}^2$ if and only if it holds for $(m(j), j)$ for some $j\in \{1, \ldots, |\bbT_{r+1, n}|\}$. For any given $j$, \eqref{eq:crude_everywhere_dominance} reduces to
  \begin{equation*}
    \label{eq:simplified_everywhere_dominance}
\frac{(n - r - 1)!}{(n - 1)!}\frac{\w^T \Psi^{(m(j) - 1), 1}}{\w^{T}\one_{n}} < \tau \le \frac{(n - r - 1)!}{(n - 1)!}\frac{\w^T \Psi^{(j), 2}}{\w^{T}\one_{n}}, \quad \w\in [0, \infty)^{n}.
\end{equation*}
This is equivalent to solving the following linear fractional programming problem and then checking if the objective is below $\tau$:
\[\min \frac{\w^T a^{(j)}}{\w^{T}\one_{n}}, \quad \text{s.t., }\frac{\w^T b^{(j)}}{\w^{T}\one_{n}}\ge \tau, \quad \w\in [0, \infty)^{n},\]
where
\[a^{(j)} = \Psi^{(m(j)), 1}\cdot \frac{(n - r - 1)! }{n - 1)!}, \quad b^{(j)} = \Psi^{(j), 2}\cdot \frac{(n - r - 1)! }{ (n - 1)!}.\]
We can apply the Charnes-Cooper transformation \citep{charnes1962programming} by introducing $v = \w / \w^{T}\one_{n}$ to transform it into a linear programming problem:
\begin{equation}
  \label{eq:charnes_cooper}
  \min v^{T}a^{(j)}, \quad \text{s.t.,}\,\,v^{T}b^{(j)} \ge \tau, v^{T}\one_{n} = 1, v\in [0, \infty)^{n}.
\end{equation}
Solving these $O(n^{r+1})$ LP problems can be  accelerated by the following two observations:
\begin{enumerate}
\item Using the same argument as in the last step of the proof of Theorem \ref{thm:ebar_Gamma_fixtau}, we can restrict
  \[j\ge \tau \frac{n!}{(n - r - 1)!}.\]
\item When $a_i^{(j)}\ge b_i^{(j)}$ for every $i\in \mathcal{N}$, then the objective of \eqref{eq:charnes_cooper} can never be below $\tau$.
\end{enumerate}

\section{~Supplementary material for Section \ref{sec:Application}}

\subsection{Description of data} \label{app:DataSources}

We briefly describe the individual samples in our meta-data. 
{\myfontsize
\begin{longtable}[H]{lcccccc}
\caption{}\\ \hline

 Source of Data & \# Obs & \# Subj & \# Lottery & Country & Gains Only \\  
\hline 
\hline
 \citet{abdellaoui2015experiments} & 801 & 89 & 3 & France & Y \\

 \citet{fan2019decisions} & 4750 & 125 & 19 & US & Y \\
    \citet{bouchouicha2017accommodating} & 3162 & 94 & 66 & UK & N \\
    \citet{sutter2013impatience} & 661 & 661 & 4 & Austria & Y \\
      \citet{etchart2011monetary} & 3036 & 46 & 20 & France & N \\
       \citet{fehr2010rationality} & 8560 & 153 & 56 & China & N \\
        \citet{lefebvre2010incentive} & 72 & 72 & 2 & France & Y \\
        \citet{halevy2007ellsberg} & 366 & 122 & 2 & Canada & Y \\
        \citet{anderhub2001interaction} & 183 & 61 & 1 & Israel & Y \\
        \citet{murad2016risk} & 2131 & 86 & 25 & UK & Y \\
        \citet{dean2019empirical} & 1032 & 179 & 3 & US & Y \\
         \citet{bernheim2020empirical} & 1071 & 153 & 7 & US & Y \\
          \citet{bruhin2010risk} & 8906 & 179 & 50 & Switzerland & N \\
               \citet{bruhin2010risk} & 4669 & 118 & 40 & Switzerland & N \\
          \citet{l2019all} & 1708 & 61 & 28 &	Australia & N \\
          \citet{l2019all} & 2548 & 95 & 28 & Belgium & N \\
           \citet{l2019all} & 2350 & 84 & 28 & Brazil & N \\
            \citet{l2019all} & 2240 & 80 & 28 & Cambodia & N \\
             \citet{l2019all} & 2687 & 96 & 28 & Chile & N \\
              \citet{l2019all} & 5711 & 204 & 28 & China & N \\
               \citet{l2019all} & 3072 & 128 & 24 & Colombia & N \\
                \citet{l2019all} & 2968 & 106 & 28 & Costa Rica & N \\
                 \citet{l2019all} & 2770 & 99 & 28 & Czech Republic & N \\
                  \citet{l2019all} & 3906 & 140 & 28 & Ethiopia & N \\
                   \citet{l2019all} & 2604 & 93 & 28 & France & N \\
                    \citet{l2019all} & 3639 & 130 & 28 & Germany & N \\
                     \citet{l2019all} & 2352 & 84 & 28 & Guatemala & N \\
             \citet{l2019all} & 2492 & 89 & 28 & India & N \\
             \citet{l2019all} & 2352 & 84 & 28 & Japan & N \\
              \citet{l2019all} & 2716 & 97 & 28 & Kyrgyzstan & N \\
 \citet{l2019all} & 1791 & 64 & 28 & Malaysia & N \\
  \citet{l2019all} & 3360 & 120 & 28 & Nicaragua & N \\
   \citet{l2019all} & 5638 & 202 & 28 & Nigeria & N \\
    \citet{l2019all} & 2660 & 95 & 28 & Peru & N \\
     \citet{l2019all} & 2491 & 89 & 28 & Poland & N \\
 \citet{l2019all} & 1959 & 70 & 28 & Russia & N \\
  \citet{l2019all} & 1819 & 65 & 28 & Saudi Arabia & N \\
   \citet{l2019all} & 1988 & 71 & 28 & South Africa & N \\
    \citet{l2019all} & 2240 & 80 & 28 & Spain & N \\
     \citet{l2019all} & 2212 & 79 & 28 & Thailand & N \\
      \citet{l2019all} & 2070 & 74 & 28 & Tunisia & N \\
       \citet{l2019all} & 2240 & 80 & 28 & UK & N \\
        \citet{l2019all} & 2701 & 97 & 28 & US & N \\
         \citet{l2019all} & 2436 & 87 & 28 & Vietnam & N \\
   \hline  
\end{longtable}
}

\subsection{Papers as domains}\label{sec:alt-def}

We now consider an alternative definition of domains, with each of the 14 papers representing a different domain. This changes the content of the i.i.d.\ assumption imposed in Section \ref{sec:ResultsIID}, where we now assume that samples are i.i.d.\ across papers, but may be dependent across subject pools within the same paper. We repeat our main analysis and report 78\% two-sided forecast intervals in Figure \ref{fig:PapersasDomains}. These intervals are qualitatively similar to those reported in Figure \ref{fig:CI_77}.

\begin{figure}%[h]
\centering
\subfloat[Raw transfer error]{\includegraphics[width=.45\linewidth]{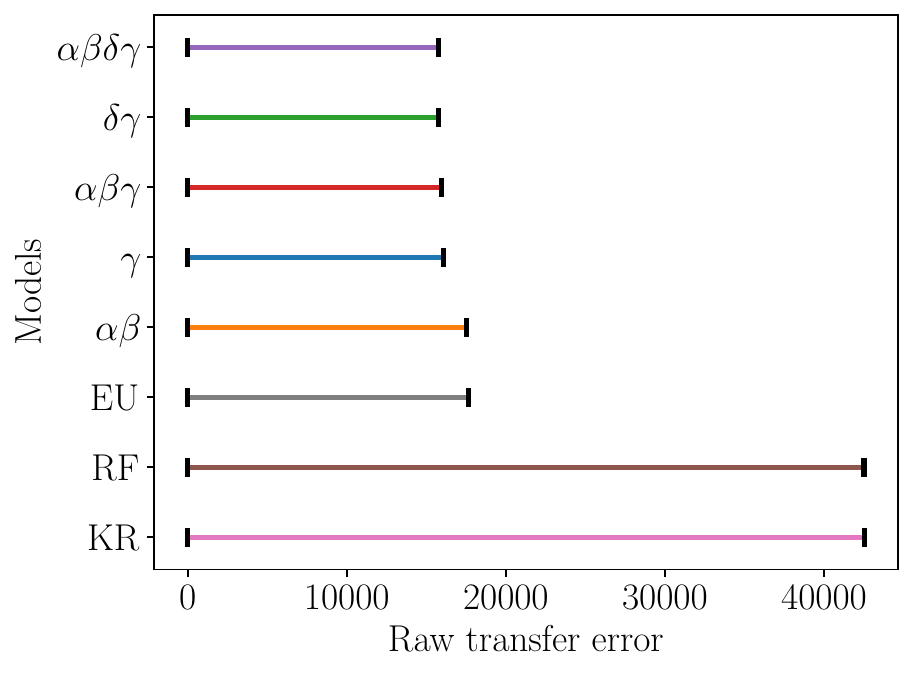}}\\
\subfloat[Transfer shortfall]{\includegraphics[width=.45\linewidth]{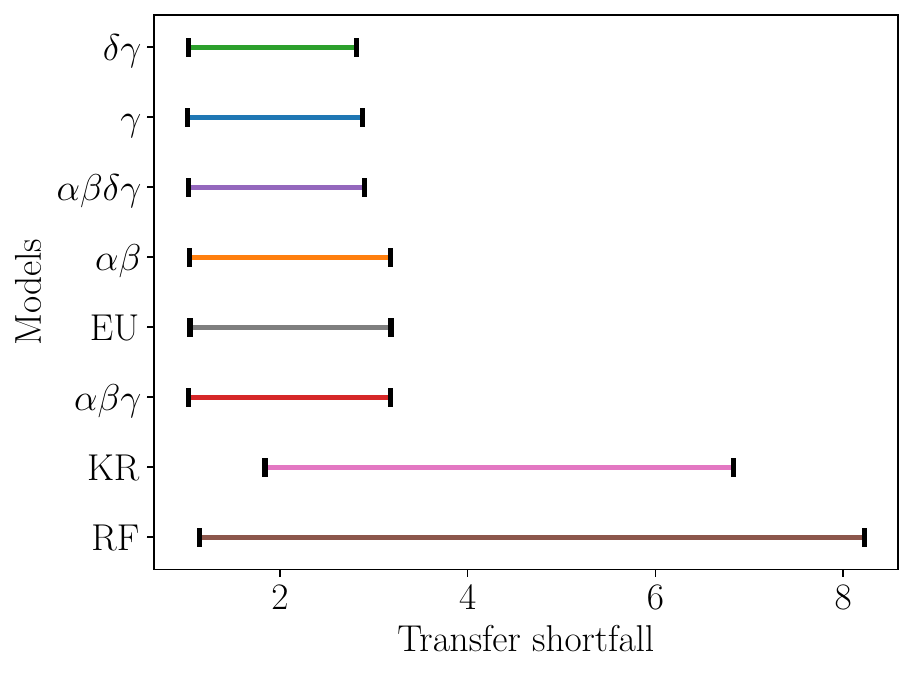}}\hfill
\subfloat[Transfer deterioration]{\includegraphics[width=.45\linewidth]{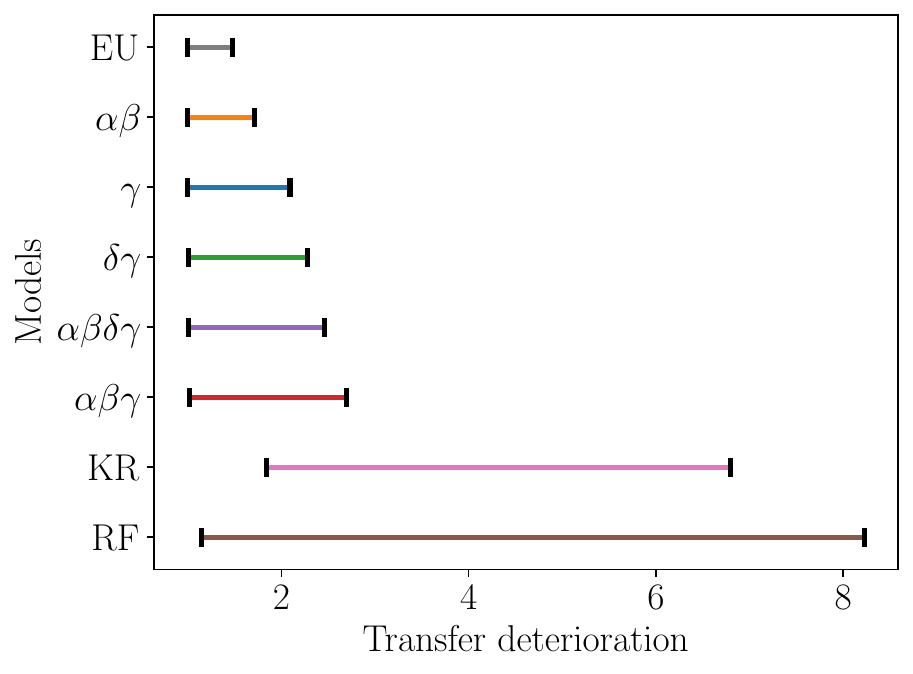}}\par 
\caption{78\% (n=14, $\tau=0.95$) forecast intervals for each of the three measures, treating each paper as a separate domain.}
\label{fig:PapersasDomains}
\end{figure}

\subsection{Supplementary tables and figures for main analysis} \label{app:Exact}

Table \ref{tab:CI} reports the forecast intervals that are depicted in Figure \ref{fig:CI_77}.

\begin{table}[H]
\centering{}%
\begin{tabular}{lccc}
\hline 
 Model & Raw Transfer Error & Transfer Shortfall & Transfer Deterioration \\
\hline 
\hline 
CPT variants \\
\quad $\gamma$                      & [2.50,15.83] & [1.03,2.54] & [1.00,1.47]\\
\quad $\alpha,\beta$                & [2.56,16.13] & [1.04,2.35] & [1.00,1.30] \\
\quad $\delta,\gamma$               & [2.48,17.19] & [1.02,2.47] & [1.00,1.53] \\
\quad $\alpha,\beta,\gamma$         & [2.47,15.91] & [1.02,2.60] & [1.00,1.85] \\
\quad $\alpha,\beta,\delta,\gamma$  & [2.46,15.99] & [1.02,2.62] & [1.00,1.82]\\[2mm]
EU models \\
\quad EU                        & [2.56,16.41] & [1.04,2.14] & [1.00,1.30] \\[2mm]
ML algorithms \\
\quad Random Forest                 & [2.71,31.39] & [1.02,6.42] & [1.02,6.42]\\
\quad Kernel Regression                   & [2.75,33.62] & [1.02,5.33] & [1.01,5.29] \\
\hline
\end{tabular}
\caption{86\% (n=44, $\tau=0.95$) forecast intervals} \label{tab:CI}
\end{table}

\subsection{Alternative forecast intervals} \label{app:AlternativeCI}

In this section, we report  alternative forecast intervals for our three measures. Table \ref{tab:CI_95} constructs 96\% two-sided forecast intervals (setting $\tau=1$),\footnote{The lower bounds of these intervals are the minimum transfer error (among the pooled transfer errors) and the upper bounds are the maximum transfer error.}  and Table \ref{tab:CI_91} reports  91\% one-sided forecast intervals (setting $\tau=0.95$). All of the forecast intervals are qualitatively similar to the 86\% two-sided  forecast intervals reported in the main text.

\begin{table}[h]
\centering{}%
\begin{tabular}{lrrrrrr}
\hline 
 Model & Raw Transfer Error & Transfer Shortfall & Transfer Deterioration \\
\hline 
\hline 
CPT main variants \\
\quad $\gamma$                      & [0.81,23104.96] & [1.01,7.31] & [1.00,7.22] \\
\quad $\alpha,\beta$                & [0.71,19999.41] & [1.00,5.28] & [1.00,5.27] \\
\quad $\delta,\gamma$               & [0.71,23052.76] & [1.00,7.25] & [1.00,7.18] \\
\quad $\alpha,\beta,\gamma$         & [0.71,28122.26] & [1.00,5.65] & [1.00,5.60] \\
\quad $\alpha,\beta,\delta,\gamma$  & [0.71,27959.10] & [1.00,6.01] & [1.00,5.95] \\[2mm]
EU models \\

\quad EU                        & [0.72,22787.99] & [1.00,4.44] & [1.00,1.75] \\[2mm]
ML algorithms \\
\quad Random Forest                 & [0.96,42520.49] & [1.01,33.17] & [1.01,33.17] \\

\quad Kernel Regression                 & [1.01,42519.23] & [1.01,6.835] & [1.00,6.79] \\
 \hline 
\end{tabular}
\caption{96\% (n=44, $\tau=1$) two-sided forecast intervals} \label{tab:CI_95}
\end{table}

\begin{table}[h]
\centering{}%
\begin{tabular}{lrrrrrr}
\hline 
 Model & Raw Transfer Error & Transfer Shortfall & Transfer Deterioration \\
\hline 
\hline 
CPT main variants \\
\quad $\gamma$                      & [0,15.83] & [1,2.54] & [1,1.47]\\
\quad $\alpha,\beta$                & [0,16.13] & [1,2.35] & [1,1.30] \\
\quad $\delta,\gamma$               & [0,17.19] & [1,2.47] & [1,1.53]\\
\quad $\alpha,\beta,\gamma$         & [0,15.91] & [1,2.60] & [1,1.85] \\
\quad $\alpha,\beta,\delta,\gamma$  & [0,15.99] & [1,2.62] & [1,1.82]\\[2mm]
EU models \\
\quad EU                        & [0,16.41] & [1,2.14] & [1,1.30] \\[2mm]
ML algorithms \\
\quad Random Forest                 & [0,31.39] & [1,6.42] & [1,6.42]\\
\quad Kernel Regression                   & [0,33.62] & [1,5.33] & [1,5.29] \\
 \hline 
\end{tabular}
\caption{91\% (n=44, $\tau=0.95$) one-sided forecast intervals} \label{tab:CI_91}
\end{table}

Finally, Figure \ref{fig:Full} plots  the $\tau$-th percentile of the pooled transfer errors as $\tau$ varies. The figure shows that the qualitative conclusions we have drawn about the relative performance of black boxes and economic models are not specific to any choice of $\tau$.\footnote{To improve readability, we remove extreme numbers by truncating $\tau \in [5,95]$, and show results only for the $\alpha\beta\gamma\delta$ specification of the CPT model.} In fact, in Panels (a) and (c), the black box curves lie everywhere above the CPT and EU curves, so both the lower and upper bounds of the black boxes' forecast intervals are higher than those of the economic models for every choice of $\tau$. 

\begin{figure}[h]\centering
\subfloat[Raw transfer error]{\label{fig:Full_a}\includegraphics[width=.45\linewidth]{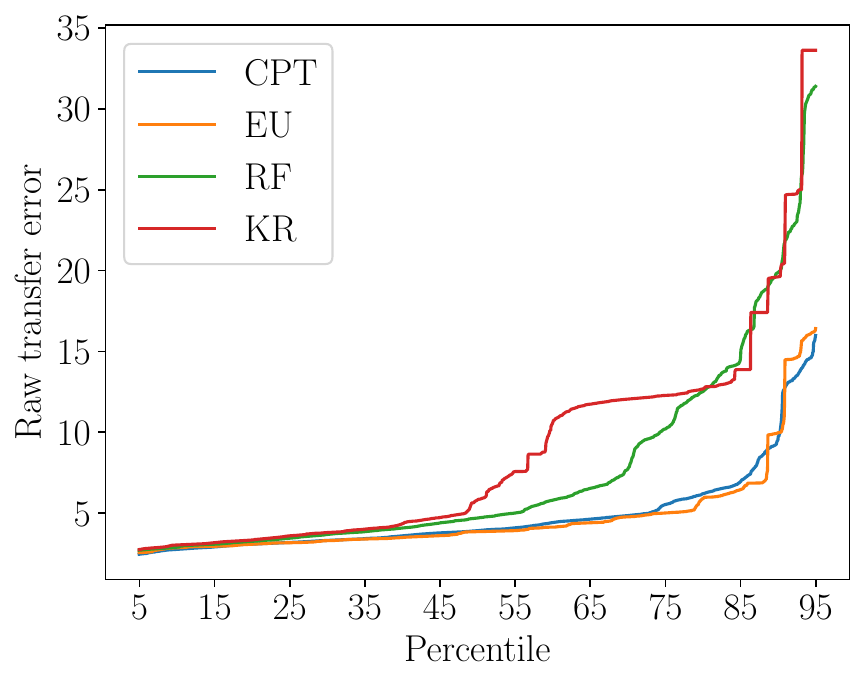}}\\
\subfloat[Transfer shortfall]{\label{fig:Full_b}\includegraphics[width=.45\linewidth]{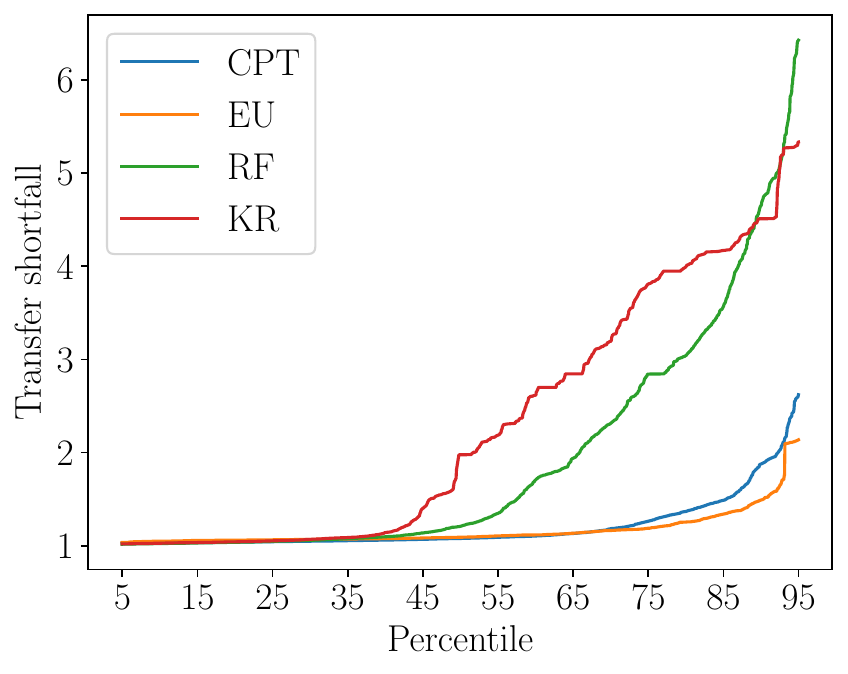}}\hfill
\subfloat[Transfer deterioration]{\label{fig:Full_c}\includegraphics[width=.45\linewidth]{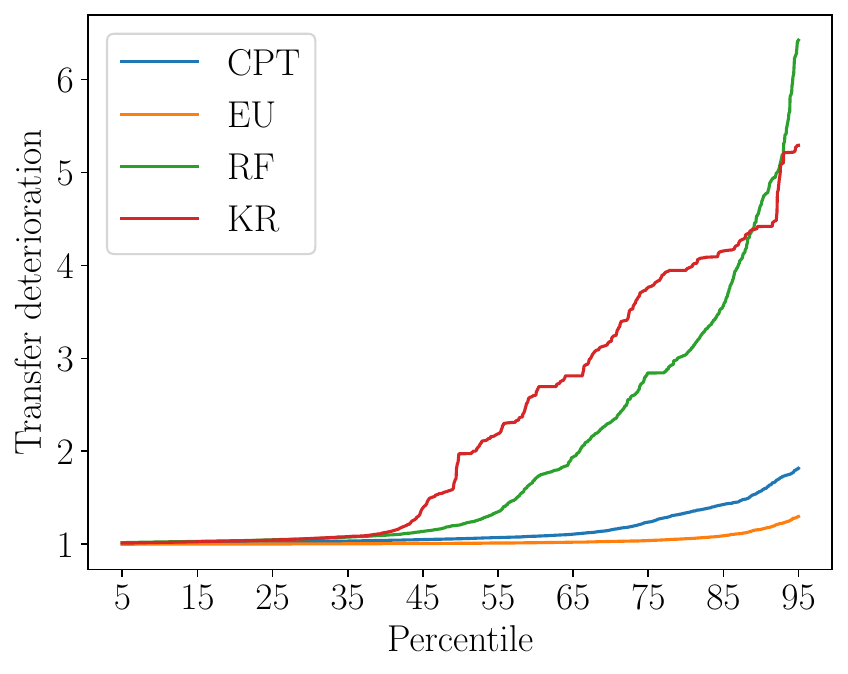}}\par 
\caption{Error percentiles from  5 to 95 (truncated for readability).}
\label{fig:Full}
\end{figure}

\subsection{Forecast intervals for the ratio of raw CPT and RF transfer errors} \label{sec:Compare}

Let $e_{\T,d}$
be the ratio of the raw random forest transfer error to the raw CPT transfer error (i.e., using the specification in (\ref{eq:TransferError})), henceforth the \emph{transfer error ratio}.

\begin{figure}[h]\centering
\subfloat[86\% (n=44, $\tau=0.95$) Forecast intervals]{\includegraphics[width=.45\linewidth]{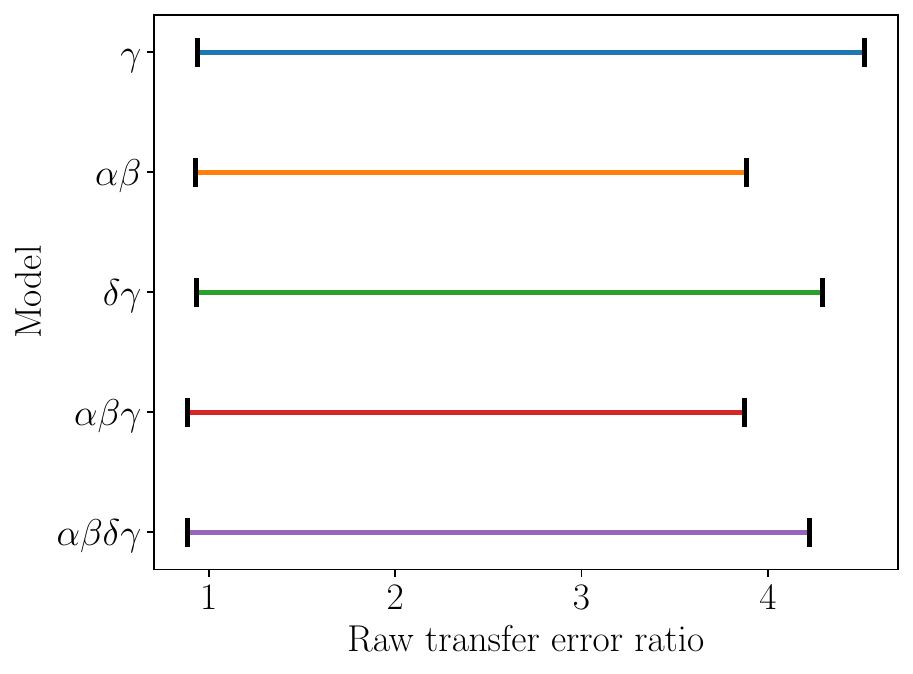}}\\
\subfloat[Density]{\includegraphics[width=.45\linewidth]{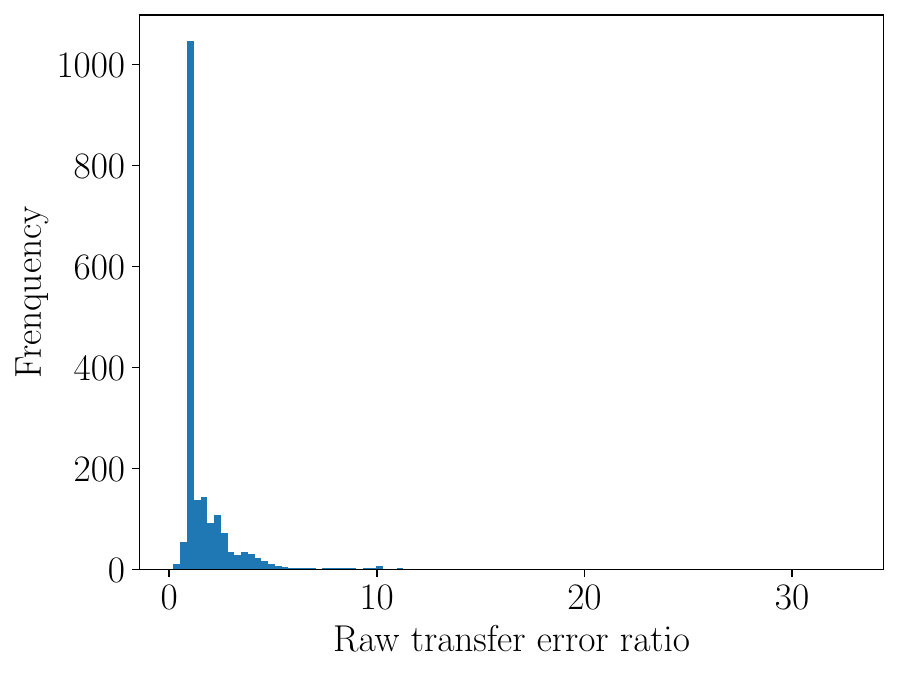}}\hfill
\subfloat[CDF]{\includegraphics[width=.45\linewidth]{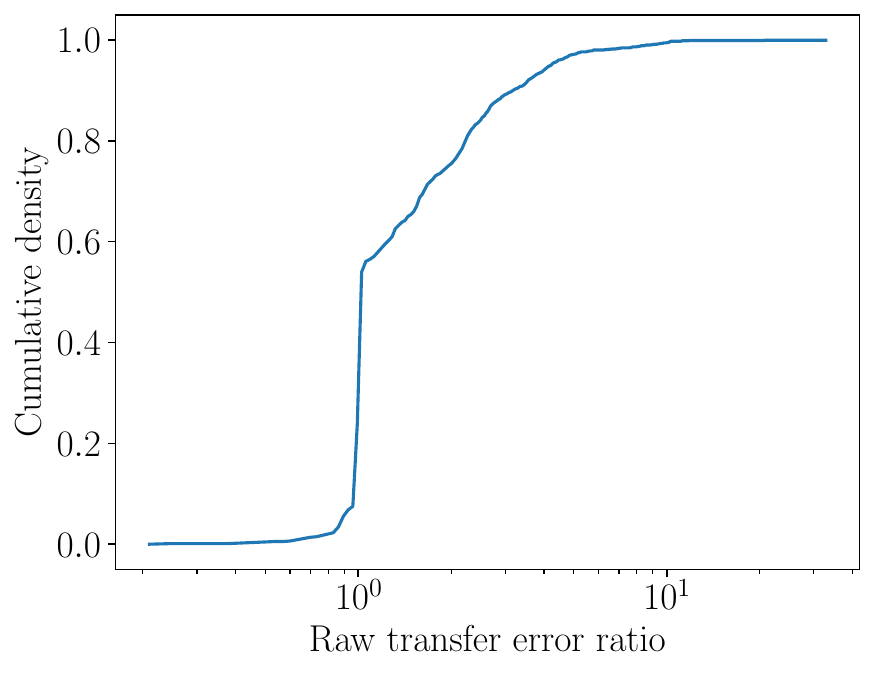}}\par 
\caption{Forecast intervals, density, and cdf for the ratio of the raw random forest transfer error to the raw CPT transfer error.}
\label{fig:Ratio}
\end{figure}

 Panel (a) of Figure \ref{fig:Ratio} reports 86\% two-sided forecast intervals for the raw transfer error ratio for each CPT specification.
  The lower bound for each CPT model is approximately $0.9$, while the upper bound is as large as 4.5.  Panel (b) of the figure is a histogram of raw transfer error ratios for the 4-parameter CPT model when the training domains $\T$ and the target domains $d$ are drawn uniformly at random from the set of domains in the meta-data.
This distribution has a large cluster of ratios around 1 (i.e., raw CPT transfer errors are similar to the raw random forest errors) and a long right tail of ratios achieving a max value of 32.8 (i.e., the random forest error can be up to 32 times as large as the CPT error). 
The cumulative distribution function of $e_{\T,d}$, reported in Panel (c) of Figure \ref{fig:Ratio}, shows that  the random forest algorithm outperforms CPT in approximately $35\%$ of $(\T,d)$ pairs, although CPT rarely has a much worse 
raw transfer error than the random forest and is sometimes much better. 

\subsection{Alternative Choice of $r$} \label{app:Altk}

Here we consider an alternative choice for the number of training domains, setting $r=3$ instead of $r=1$.

This corresponds to randomly choosing 3 of the 44 domains to be the training domains, finding the best prediction rule for this pooled data, and using the estimated prediction rule to predict the remaining 41 samples. For this analysis we use domain cross-validation to select tuning parameters for the black box algorithms, as described in Example \ref{ex:DomainCV}.

\begin{figure}[h]\centering
\subfloat[]{\label{fig:CI_76_3_domain_a}\includegraphics[width=.45\linewidth]{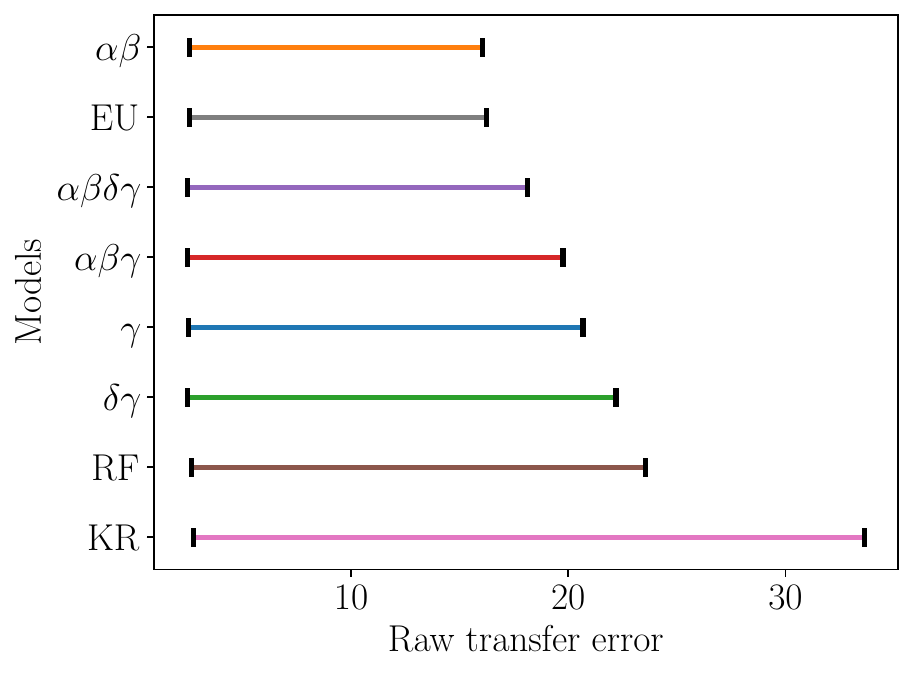}}\\
\subfloat[]{\label{fig:CI_76_3_domain_b}\includegraphics[width=.45\linewidth]{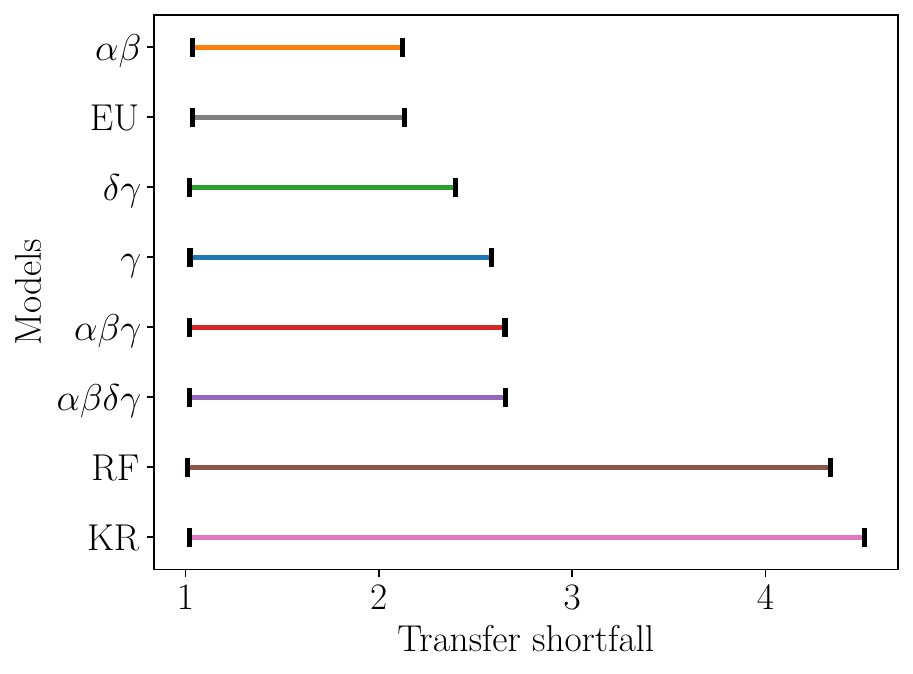}}\hfill
\subfloat[]{\label{fig:CI_76_3_domain_c}\includegraphics[width=.45\linewidth]{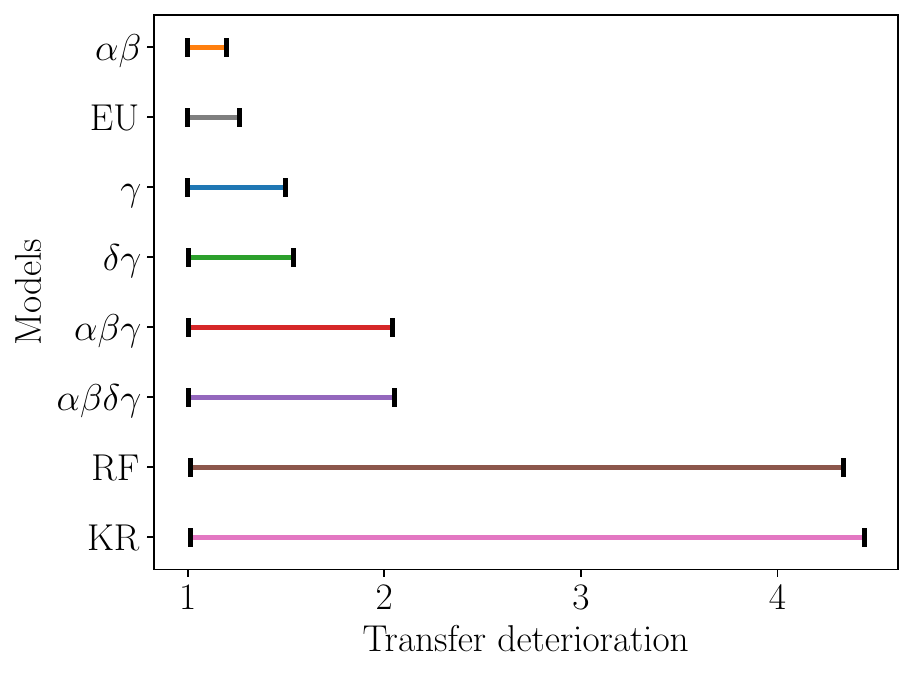}}\par 
\caption{82\% (n=44, $\tau=0.95$) forecast intervals for (a) raw transfer error, (b) transfer shortfall, and (c) transfer deterioration, with the choice of $r=3$.}
\label{fig:CI_76_3_domain}
%broken ref

\end{figure}

Figure \ref{fig:CI_76_3_domain} is the analog of Figure \ref{fig:CI_77}. Again we choose $\tau=0.95$, thus constructing forecast intervals whose lower bounds are the 5\% percentile of pooled transfer errors, and whose upper bounds are the 95\% percentile of pooled transfer errors. Applying Proposition \ref{prop:pooled, percentiles}, these are 82\% forecast intervals. The most notable change is that the random forest forecast interval shrinks considerably, which suggests that the raw transfer error of the random forest algorithm becomes less variable when it is trained on more domains. Otherwise, all of the qualitative statements in the main text for $r=1$ continue to hold. In particular, as with $r=1$, we find that the forecast intervals for all three of our measures have higher lower and upper bounds for the black box algorithms than for the CPT specifications.

\subsection{Supplementary Material to Section \ref{sec:Overfit}} \label{app:r=3or5}
Here we consider an alternative choice of for the number of training samples, setting $r=3$ and $r=5$ instead $r=1$. Recalling that each sample includes the observations associated with a unique lottery, 
this corresponds to randomly choosing three (or five) of the 24 lotteries for training, finding the best prediction rule for this pooled data, and using the estimated prediction rule to predict certainty equivalents for the remaining 21 (or 19) lotteries. We use domain cross-validation to select tuning parameters for the black box algorithms, as described in Example \ref{ex:DomainCV}.

Figure \ref{fig:lottery_transfer_r3} and Figure \ref{fig:lottery_transfer_r5} are the analog of Figure \ref{fig:lottery_transfer_r1}, with $r=3$ and $r=5$ respectively. We again choose $\tau=0.95$, thus constructing forecast intervals whose lower bounds are the 5\% percentile of pooled transfer errors, and whose upper bounds are the 95\% percentile of pooed transfer errors. Applying Proposition \ref{prop:pooled, percentiles}, these are 76\% for $r=3$ and 68\% for $r=5$ forecast intervals. The most notable change is that the forecast intervals shrink for all of the prediction methods, which suggests that the raw transfer error becomes less variable when it is trained on more lotteries. Otherwise, all of the qualitative statements in the main text for $r=1$ continue to hold, and in particular the economic models continue to transfer better than the black box algorithms do.

\begin{figure}[h]\centering
\subfloat[]{\label{fig:lottery_transfer_r3_a}\includegraphics[width=.45\linewidth]{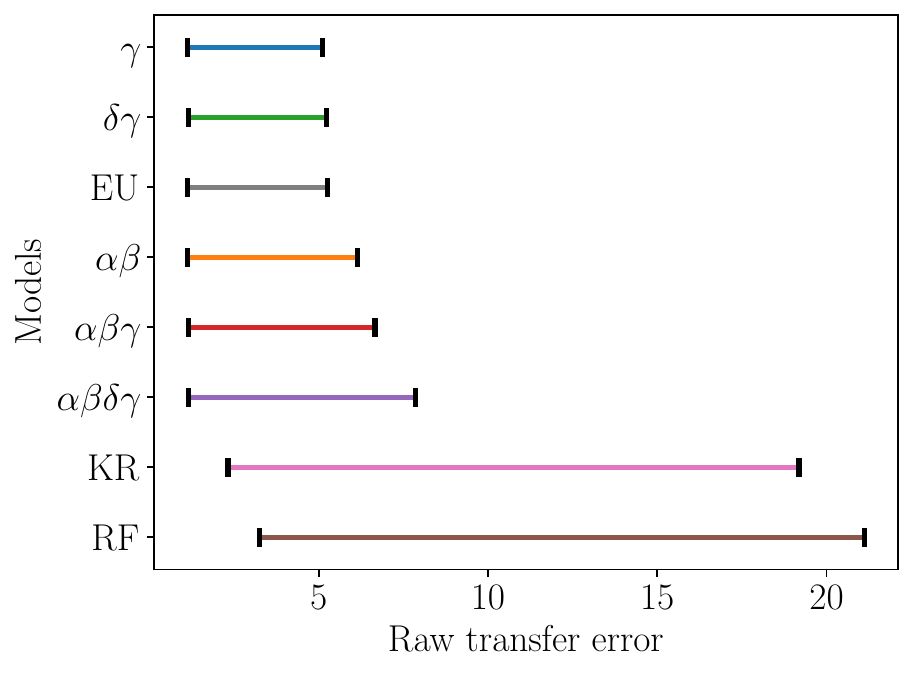}}\\
\subfloat[]{\label{lottery_transfer_r3_b}\includegraphics[width=.45\linewidth]{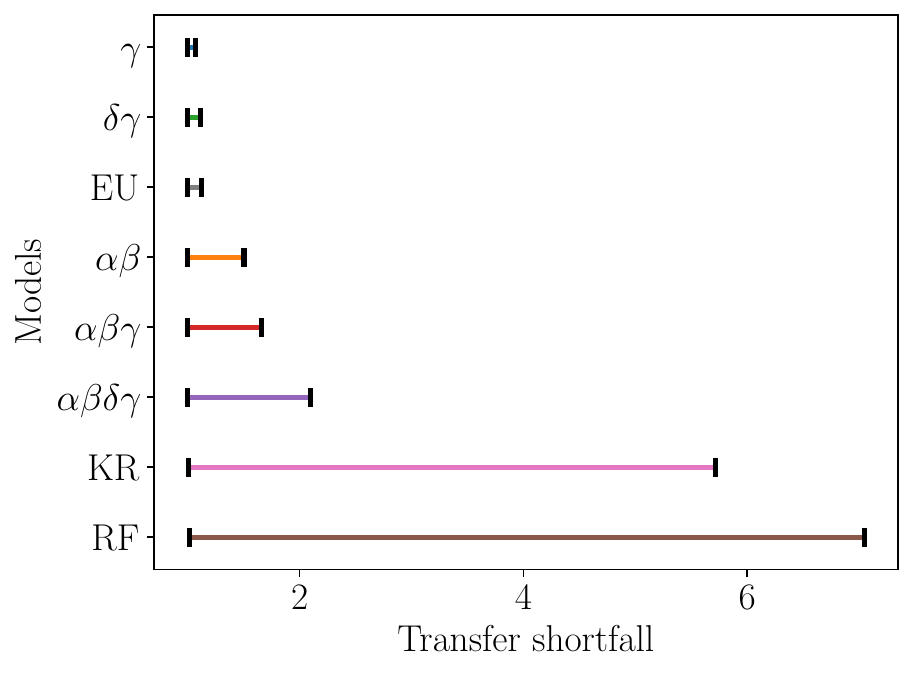}}\hfill
\subfloat[]{\label{fig:lottery_transfer_r3_c}\includegraphics[width=.45\linewidth]{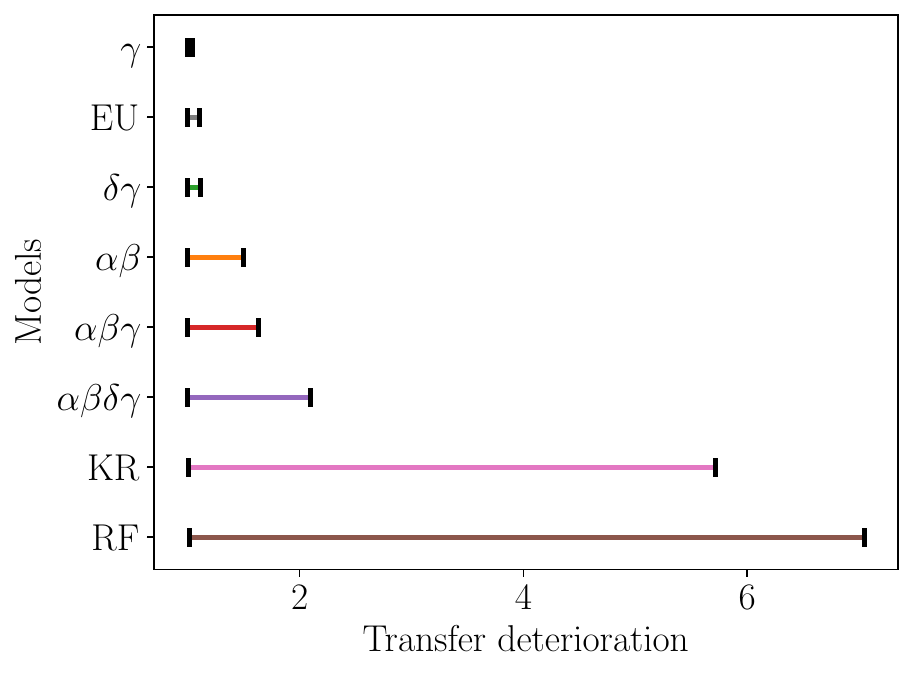}}\par 
\caption{76\% (n=24, $\tau$=0.95) forecast intervals using common lotteries in  \citet{l2019all}, with the choice of $r=3$.  }
\label{fig:lottery_transfer_r3}
\end{figure}

\begin{figure}[h]\centering
\subfloat[]{\label{fig:lottery_transfer_r5_a}\includegraphics[width=.45\linewidth]{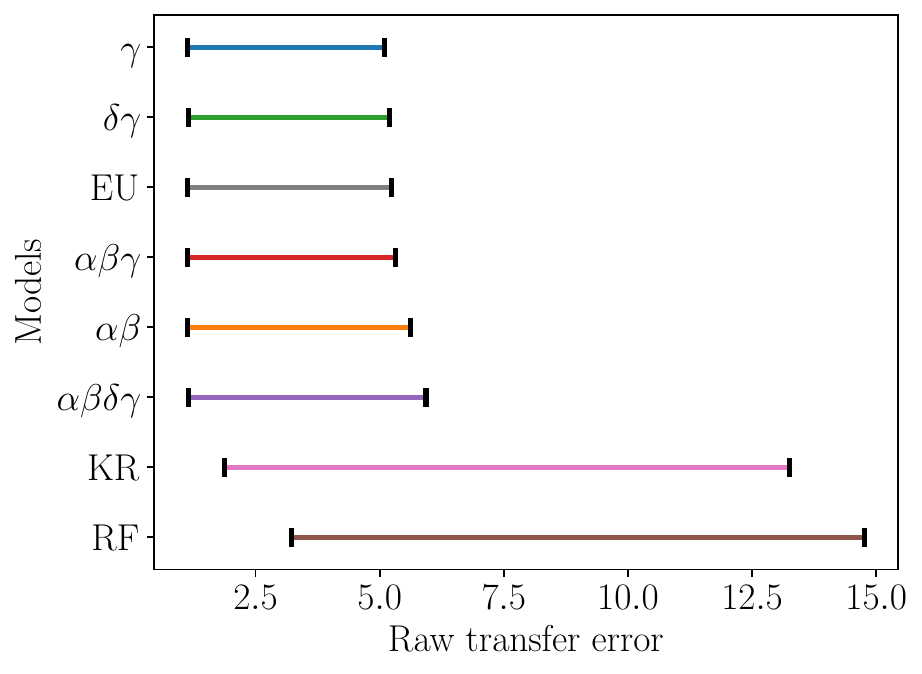}}\\
\subfloat[]{\label{fig:lottery_transfer_r5_b}\includegraphics[width=.45\linewidth]{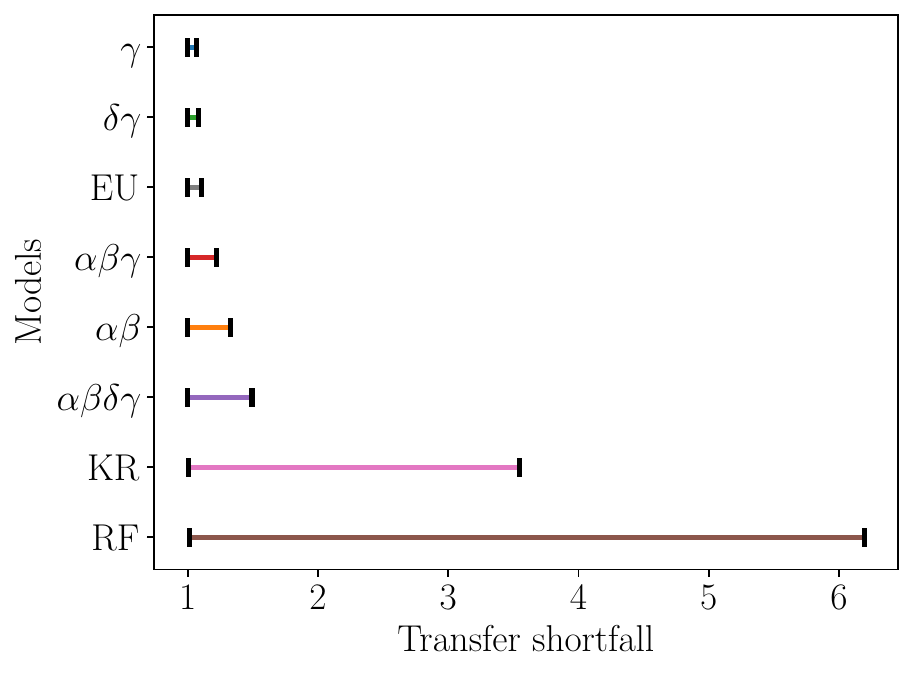}}\hfill
\subfloat[]{\label{fig:lottery_transfer_r5_c}\includegraphics[width=.45\linewidth]{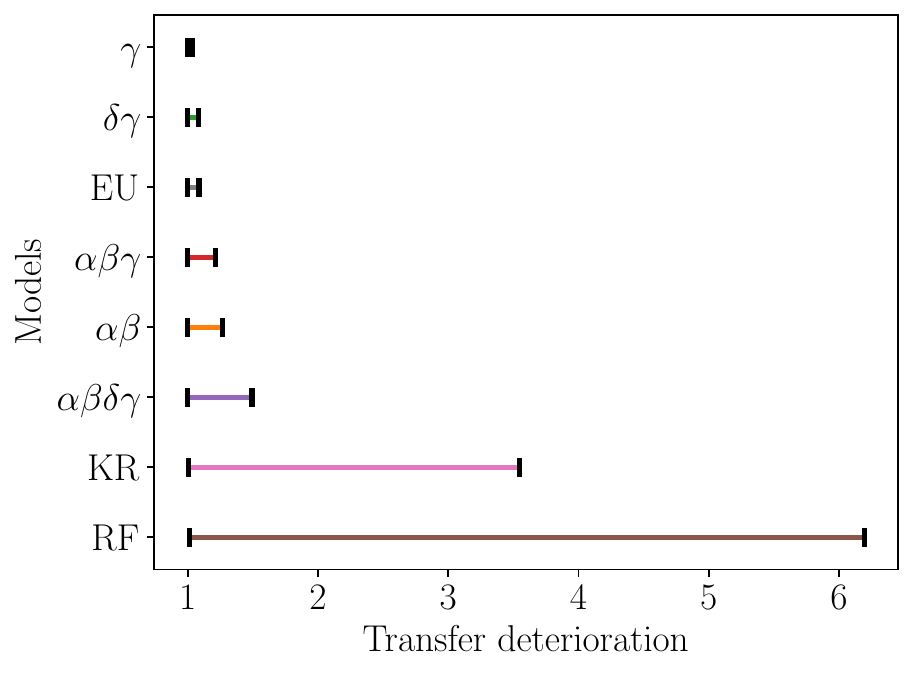}}\par 
\caption{68\% (n=24, $\tau$=0.95) forecast intervals using common lotteries in  \citet{l2019all}, with the choice of $r=5$.  }
\label{fig:lottery_transfer_r5}
\end{figure}

\subsection{More details on worst-case dominance}\label{app:worstcase_dominance}

Figures  \ref{fig:VaryGamma} and  \ref{fig:VaryTau} compare the worst case upper bound of the forecast intervals for CPT and RF for our three transfer measures as either $\gamma$ or $\tau$ varies. In each case the dominance relation is clear.

  \begin{figure}[h]
    \centering
    \subfloat{
    \includegraphics[scale=0.7]{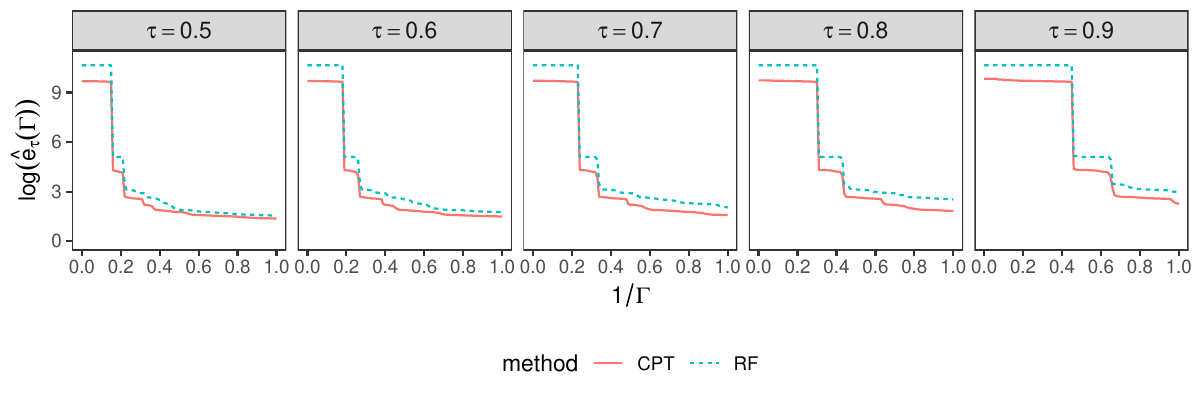}
    }\\
    \subfloat{
    \includegraphics[scale=0.7]{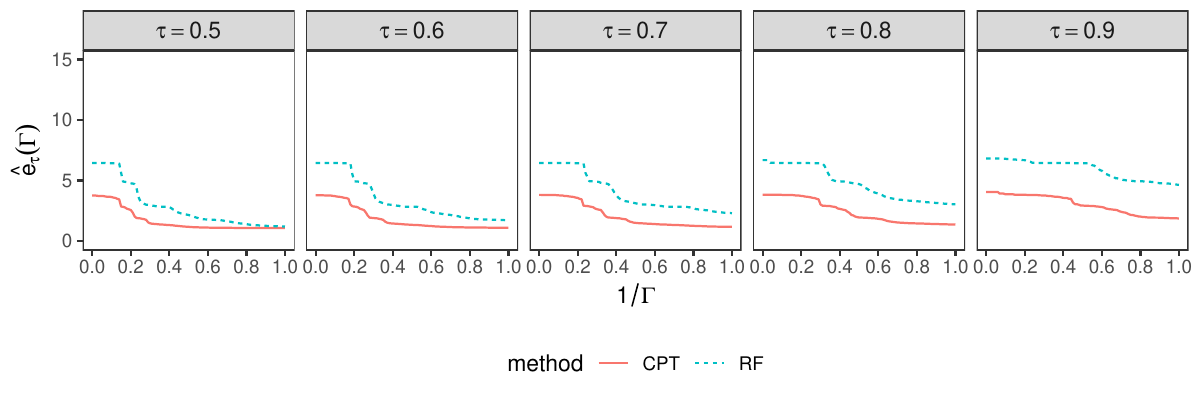}
    }\\
    \subfloat{
    \includegraphics[scale=0.7]{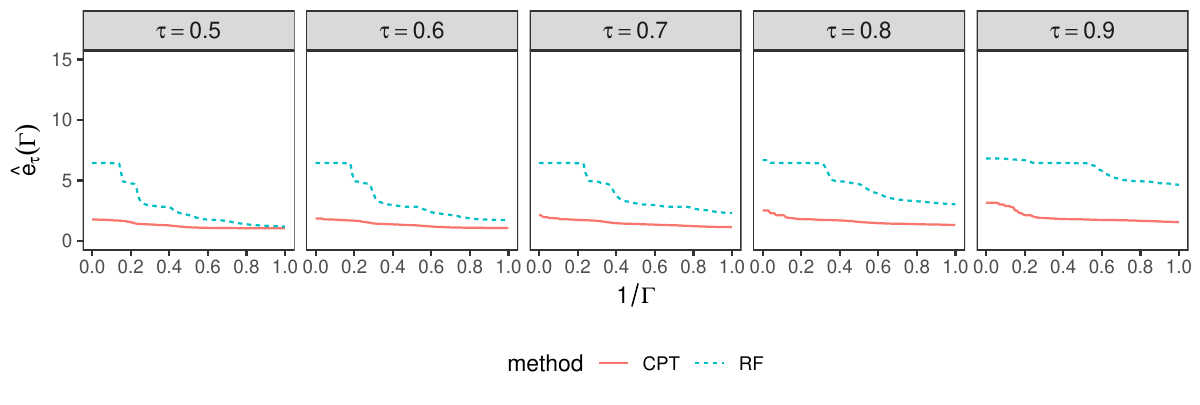}
    }
    \caption{The worst case upper prediction bound $\hat{e}_\tau(\Gamma)$ (as defined in (\ref{eq:WorstCase})) for (a) raw transfer error, (b) transfer shortfall, and (c) transfer deterioration of CPT and RF as a function of $\Gamma\in [1, \infty)$, discretized at $100/i (i = 0, 1, \ldots, 100)$, at different quantiles $\tau \in \{0.5, 0.6, 0.7, 0.8, 0.9\}$.} \label{fig:VaryGamma}
  \end{figure}

  \begin{figure}[h]
    \centering
     \subfloat{
    \includegraphics[scale=0.7]{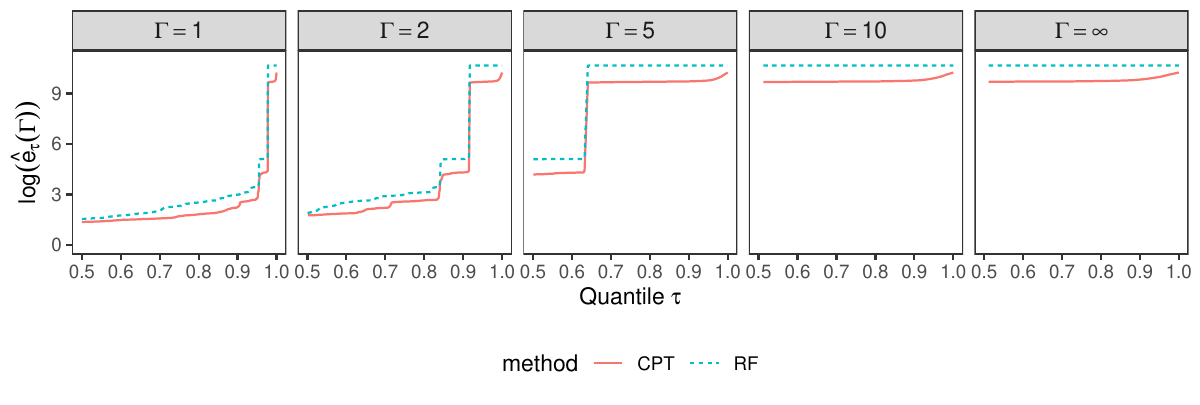}
    }\\
    \subfloat{
    \includegraphics[scale=0.7]{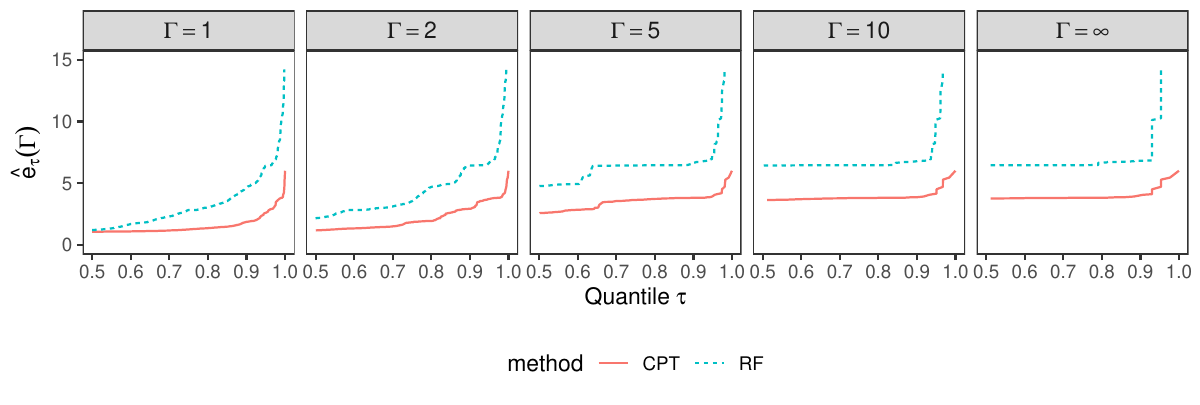}
    }\\
    \subfloat{
    \includegraphics[scale=0.7]{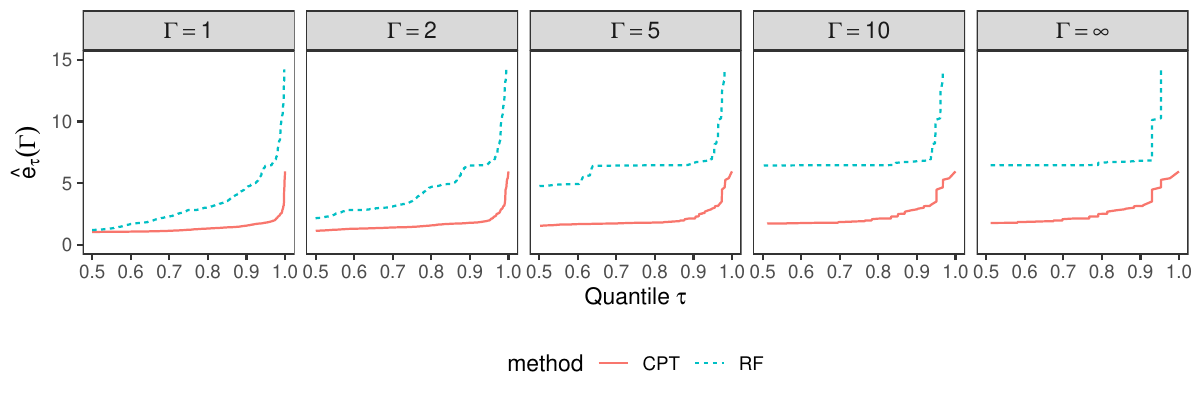}
    }
    \caption{The worst case upper prediction bound $\hat{e}_\tau(\Gamma)$ (as defined in (\ref{eq:WorstCase})) for (a) raw transfer error, (b) transfer shortfall, and (c) transfer deterioration of CPT and RF as a function of $\tau\in [0.5, 1]$ without discretization for $\Gamma \in \{1, 2, 5, 10, \infty\}$.}
    \label{fig:VaryTau}
  \end{figure}

\end{document}